\documentclass{statsoc}
\usepackage{setspace}
\usepackage[utf8]{inputenc}
\usepackage{graphicx}
\usepackage{natbib}
\usepackage{microtype}

\usepackage{color}
\usepackage{xcolor}
\usepackage{amsmath}
\usepackage{amssymb}
\usepackage{mathtools}
\usepackage{booktabs}
\usepackage{bm}

\usepackage[algo2e]{algorithm2e} 
\usepackage{algorithm}
\usepackage{verbatim}
\usepackage{setspace}
\usepackage{multirow}
\usepackage{lineno}

\usepackage{verbatim}

\usepackage{doi}
\usepackage{enumerate}
\usepackage{tikz} 
\usetikzlibrary{arrows.meta,calc}

\usepackage{cancel}

\newcommand{\bc}{\boldsymbol c}
\newcommand{\bg}{\boldsymbol{g}}
\newcommand{\bh}{\boldsymbol{h}}

\newcommand{\bs}{\boldsymbol{s}}
\newcommand{\bt}{\boldsymbol{t}}
\newcommand{\bu}{\boldsymbol{u}}
\newcommand{\bx}{\boldsymbol{x}}
\newcommand{\bz}{\boldsymbol{z}}

\newcommand{\bG}{\boldsymbol{G}}

\newcommand{\bQ}{\boldsymbol{Q}}
\newcommand{\bS}{\boldsymbol{S}}
\newcommand{\bT}{\boldsymbol{T}}

\newcommand{\bZ}{\boldsymbol{Z}}

\newcommand{\bvarepsilon}{\boldsymbol{\varepsilon}}

\newcommand{\byobs}{\boldsymbol{y}}

\newcommand{\argmin}{\mathop{\mathrm{arg\,min}}}  

\newcommand{\bbeta}{\boldsymbol{\beta}}

\newcommand{\btheta}{\boldsymbol{\theta}}
\newcommand{\bphi}{\boldsymbol{\phi}}

\newcommand{\bnu}{\boldsymbol{\nu}}
\newcommand{\btau}{\boldsymbol{\tau}}
\newcommand{\boldeta}{\boldsymbol{\eta}}

\newcommand{\var}{\mathsf{Var}}

\definecolor{cFeas}{HTML}{0072B2}
\colorlet{cPen}{cFeas}              
\definecolor{cPost}{gray}{0.45}
\colorlet{cTraj}{black}
\definecolor{cGrad}{RGB}{99,90,191}    

\newcommand{\normal}[1]{\mathsf{Normal}\left(#1\right)}

\DeclareMathOperator{\elpd}{ELPD}

\DeclareMathOperator{\rmse}{RMSE}
\DeclareMathOperator{\relrmse}{rel\text{-}RMSE}

\DeclareMathOperator{\Rtd}{R2D2}
\DeclareMathOperator{\Rtdprox}{Prox-R2}

\DeclareMathOperator{\rhs}{Regularized\ Horseshoe}
\DeclareMathOperator{\nig}{Normal-Gamma}

\title[Proximal R2 priors]{%
  Proximal R2 Priors for\\
  Structured Variance Allocation}
\author{Javier Enrique Aguilar}
\address{TU Dortmund, RC TRUST (UA Ruhr); corresponding author}
\email{javier.aguilarr@icloud.com}
\author{David Kohns}
\address{Aalto University}
\author{Eric Yanchenko}
\address{Akita International University}
\author{Shonosuke Sugasawa}
\address{Keio University}
\author{Kaoru Irie}
\address{University of Tokyo}
\author{Aki Vehtari}
\address{Aalto University}
\author[Aguilar \textit{et al.}]{Paul-Christian B\"urkner}
\address{TU Dortmund}

\newtheorem{theorem}{Theorem}[section]
\newtheorem{proposition}[theorem]{Proposition}
\newtheorem{lemma}[theorem]{Lemma}
\newtheorem{corollary}[theorem]{Corollary}

\newtheorem{assumption}{Assumption}[section]

\makeatletter
\long\def\@makecaption#1#2{%
  \ifx\statsoc@figure\@captype
    \vskip\abovecaptionskip
    \sbox\@tempboxa{\sffamily\footnotesize\textbf{#1.} #2\strut}%
    \ifdim \wd\@tempboxa >\hsize
      \sffamily\footnotesize\textbf{#1.} #2\strut\par
    \else
      \global\@minipagefalse
      \hb@xt@\hsize{\hfil\box\@tempboxa\hfil}%
    \fi
    \vskip\belowcaptionskip
  \else
    \parbox{\textwidth}{\sffamily\footnotesize\textbf{#1.} #2\strut}\aftergroup\\
  \fi}
\makeatother

\begin{document}


\begin{abstract}

Scientific domain knowledge in regression often concerns structural relationships among model terms, which are typically encoded as parameter constraints. Since such relationships rarely hold exactly in practice, prior distributions should encourage them while at the same time retaining support for scientifically plausible departures. Accounting for these relationships becomes more pressing in modern regression, where high-dimensionality is usually present and regularisation is essential. $R^2$ decomposition priors provide an interpretable foundation for regularisation by placing a prior on total explained variance and distributing the implied signal across model terms. We introduce the Proximal R2 prior, which uses proximal relaxations to encode these relationships as soft constraints on the variance allocation. Our formulation yields differentiable posteriors compatible with standard gradient based samplers. Theoretically, we prove that our prior preserves the characteristic marginal properties of global-local shrinkage priors. Moreover, we establish posterior consistency and contraction rates in growing dimensions, showing that the inclusion of structural relationships and proximal methods do not alter the shrinkage properties. Notably, these results demonstrate that the structural constraints in the model do not hinder posterior recovery. Simulation studies and real-data applications to interaction and generalised additive models demonstrate that our prior improves parameter recovery, out-of-sample prediction, and sampling efficiency over existing constrained and unconstrained shrinkage priors, even under structural misspecification.
\end{abstract}



\section{Introduction}
\label{sec:introduction}

Modern regression models frequently contain terms linked by scientifically meaningful relationships.  In an interaction model, for instance, an interaction is often considered more plausible when the main effects are active, a principle known as heredity \citep{chipman_bayesian_1997, griffin_hierarchical_2017}. In a generalised additive model, a parsimonious linear effect might be preferred unless the data provide sufficient evidence for a non-linear departure \citep{hastie_generalized_1986, marra2011practical}. Likewise, domain knowledge often suggests that certain treatment exposures or mechanisms should receive greater prior signal than others \citep{hahn_decoupling_2015,sewell_posterior_2025}. 
Such structural beliefs govern the relative relevance of model terms, yet they rarely hold exactly in practice. A useful prior should encourage the proposed structure without preventing the data from supporting scientifically credible departures from it \citep{robert_bayesian_2007}.

Incorporating such structural guidance is especially critical in high-dimensional regression problems, where the number of covariates is large relative to the sample size \citep{wainwright_high-dimensional_2019}. Bayesian inference responds to the complexity of these problems by designing prior distributions that encode structure \citep{castillo_needles_2012,castillo_bayesian_2015}. Global-local shrinkage priors provide a natural starting point to encourage proposed structures since they combine model-wide regularization with term-specific adaptation \citep{carvalho_horseshoe_2010, polson_local_2012, armagan_generalized_2011, bhattacharya_dirichletlaplace_2015, piironen_sparsity_2017,bai_beta_2021}. 
An array of shrinkage priors encode relationships among the local scales, facilitating regularization of one model term dependent on its relationship to others \citep{zellner_assessing_1986, griffin_structured_2024,aguilar2025dependencyawareshrinkagepriorshigh}. R2 type priors allow for interpretable model-level prior specification by separating priors beliefs about total explained variance, $R^2$, from signal allocation to model terms \citep{zhang_bayesian_2022}. This provides a principled route for incorporating structure into the model while still allowing domain knowledge to be easily elicited. Originally designed for ordinary linear regression, the framework has since been extended to multilevel \citep{aguilar_intuitive_2023, yanchenko_r2d2_2025}, time-series \citep{kohns_arr2_2025}, ordinal \citep{yanchenko_pseudo-r2d2_2025}, spatial \citep{yanchenko_spatial_2024} and grouped \citep{aguilar_r2_2025, yanchenko_group_2026} settings. Despite these advances, a general framework for encoding flexible parameter constraints and relational preferences directly into the signal allocation mechanism remains lacking.

To this end, we develop an R2 prior which embeds prespecified structural relationships among model terms and parameter constraints directly as preferences using proximal methods \citep{polson_proximal_2015}. In view of this, we name our prior distribution Proximal R2, which we denote by $\Rtdprox$. Our prior retains the model-level interpretation of the R2 prior similar to \citet{zhang_bayesian_2022}, while treating relationships as guidance, rather than properties that must hold exactly. Furthermore, the uncertainty related to constraints and relationships is expressed via proximal relaxations of the corresponding hard-constrained prior, i.e., one that would enforce the constraints strictly \citep{polson_proximal_2015, duan_bayesian_constraint_2020, presman2023distance, zhou_proximal_2024}. Our approach results in a family of softly constrained posteriors which approach the hard-constrained formulation as the amount of smoothing decreases; while still retaining support for scientifically plausible departures from the preferred relationships. Moreover, our models can be sampled efficiently with off-the-shelf state-of-the-art gradient-based samplers \citep{neal_mcmc_2011,hoffman_no-u-turn_2014}; a consequence of using proximal methods.

The combination of dependence and parametric constraints raises the question whether the  additional model structure preserves typical shrinkage prior behaviour, namely, strong shrinkage of noise and signal detection \citep{polson_shrink_2011,pas_theoretical_2021,tadesse_handbook_2021}. While explicit dependence among proportions of variance has been considered before \citep{aguilar_generalized_2025}, their interaction with explicit structural relationships has not been studied. In particular, we focus on problems whose structural relationships of interest can be expressed as sets of linear constraints (although our framework is sufficiently flexible to account for general forms). 
These constraints induce a non-standard joint distribution on the simplex, whose marginal allocation densities are generally unavailable in closed form, making their implications for shrinkage difficult to anticipate. We show that the proposed prior preserves the desired marginal shrinkage behaviour. Our analysis characterises how the global scale and the distribution of allocation shares interact to jointly determine the marginal coefficient density near zero and in the tails. 

We also establish posterior consistency and contraction for Gaussian regression with growing dimension, importantly, without requiring the true coefficients to obey the favoured structural constraints. To the best of our knowledge, these represent the first theoretical results for shrinkage behaviour and posterior consistency and contraction for continuous shrinkage priors coupled with proximal methods. The theoretical contribution also extends to the unconstrained logistic-normal R2 prior \citep{aguilar_generalized_2025}, for which the marginal origin, tail asymptotics, and posterior contraction rates had not previously been established.

We demonstrate the practical utility and scalability of the $\Rtdprox$ prior through simulations and real-data applications involving interaction and additive models. Across the considered regimes, $\Rtdprox$ consistently outperforms established shrinkage priors, regardless of whether they take constraints into account or not in their formulations. The performance gains are most pronounced under moderate-to-high signal-to-noise ratios, especially for identifying non-zero coefficients. Our method remains useful even when the constraints hold only approximately. Our results illustrate how structural knowledge, even when uncertain, can guide high-dimensional regularisation while retaining interpretable control in complex high-dimensional settings.

The paper is organized as follows. Section~\ref{sec:method} introduces the proposed prior while Section~\ref{sec:comp} describes the computational approach. The theoretical results are presented in Section~\ref{sec:theory} with Section~\ref{sec:experiments} presenting the simulation and real-data results. Finally, Section~\ref{sec:Discussion} discusses the findings and outlines directions for future research.

\section{Methodology}\label{sec:method}

\subsection{Model}
Given responses $y_i$, covariates $\bx_i^\prime=(x_{i1},\dots,x_{ip})$ and $\bz_i=h(\bx_i)\in\mathbb R^d$, which collects model terms constructed from the original covariates, consider the additive model
\begin{equation}
\label{eq:resp}
    y_i = \bz_i^\prime \bbeta + \varepsilon_i,
\end{equation}
where $\varepsilon_i\stackrel{\text{iid.}}{\sim}\mathsf{Normal}(0,\sigma^2)$ for $i=1,\dots,n$. In ordinary linear regression, $h(\bx_i)=\bx_i$ and $d=p$ while in models that contain added structure, $h(\cdot)$ may include interactions, spline basis functions, polynomial terms, or other transformations \citep{griffin_hierarchical_2017}.

\subsection{R2 priors}

We adopt the global--local shrinkage framework \citep{polson_local_2012,tadesse_handbook_2021}, in which the prior for each regression coefficient is represented as a scale mixture of normals \citep{andrews_scale_1974, west_scale_1987}. Specifically, conditional on the coefficient-specific mixing scales, the regression coefficients are independent and Gaussian:
\begin{equation}
\label{eq:normal-kernel}
    \beta_j \mid \sigma^2,\psi_j^2 \sim \mathsf{Normal}(0,\sigma^2\psi_j^2), \qquad j=1,\dots,d,
\end{equation}
where $\psi_j^2$ is the coefficient-specific prior scale, which we will decompose into global signal and term-specific allocation parameters.

Following \cite{zhang_bayesian_2022}, we construct a prior for a model-level coefficient-of-determination $R^2$ which induces a prior distribution on each individual $\psi_j^2$. Treating $\bz$ as random and independent of $\bbeta$, we define $R^2$ under the joint distribution induced by the coefficient prior and the covariate distribution as 
\begin{equation}\label{eq:bayes-r2-def}
    R^2 =  \operatorname{corr}^2(y, \bz^\prime\bbeta)= \frac{\mathsf{Var}(\bz^\prime\bbeta)} {\mathsf{Var}(\bz^\prime\bbeta)+\sigma^2}.
\end{equation}
Using the law of total variance, we obtain
\begin{equation}
\label{eq:linear-predictor-var-total}
\begin{aligned}
    \mathsf{Var}(\bz^\prime\bbeta) &= \mathbb{E}_{\bz}\{\mathsf{Var}(\bz^\prime\bbeta\mid\bz)\} + \mathsf{Var}_{\bz}\{\mathbb{E}(\bz^\prime\bbeta\mid\bz)\}
\end{aligned}
\end{equation}
where the second term is zero and
$
\mathsf{Var}(\bz^\prime\bbeta\mid\bz) = \sigma^2\sum_{j=1}^d z_j^2\psi_j^2
$
such that
\begin{equation}\label{eq:expected-var-linear-predictor}
\mathsf{Var}(\bz^\prime\bbeta)  = \sigma^2\sum_{j=1}^d \psi_j^2\mathbb{E}_{\bz}(z_j^2).
\end{equation}
We assume that $\mathbb{E}_{\bz}(z_j^2) < \infty$ for all $j$.  To separate the allocation of prior signal from the scale of the associated model term, we define
$
\rho_j = \mathbb{E}_{\bz}(z_j^2)
$
for $j=1,\dots,d$ \citep{aguilar_intuitive_2023},
and parametrise the coefficient-specific variances as
\begin{equation}\label{eq:lambda-rho-scaling}
    \psi_j^2 = \frac{\phi_j\omega^2}{\rho_j}, \qquad j=1,\dots,d,
\end{equation}
where $\bphi=(\phi_1,\dots,\phi_d)^\prime$ lies on the simplex, i.e $\phi_j\geq 0$ and $\sum_{j=1}^d\phi_j=1$. Then
$
    \sum_{j=1}^d \psi_j^2\mathbb{E}_{\bz}(z_j^2)
    =
    \omega^2
$
such that
\begin{equation}\label{eq:r2-omega}
    R^2 = \frac{\sigma^2\omega^2} {\sigma^2\omega^2+\sigma^2} = \frac{\omega^2}{\omega^2+1}.
\end{equation}
 Therefore, $\omega^2$ controls the total prior signal relative to the residual variance, while $\phi_j$ is the variance allocation share assigned to the $j$th model term. As in previous work on R2 priors, we let $R^2$ follow a Beta distribution, which induces a beta-prime prior distribution \citep{johnson1994continuous} on $\omega^2$  i.e., $R^2\sim\mathsf{Beta}(a,b)$ implies $\omega^2\sim\mathsf{BetaPrime}(a,b)$.

\subsection{ Logistic-normal allocation}
The R2 prior framework requires that $\bphi$ lie on the simplex so that the total prior signal $\omega^2$ is apportioned to the different model terms \citep{zhang_bayesian_2022}. Generally speaking, if $(T_1,\dots,T_d)$ is a vector of positive random variables, then we can define
\begin{equation}
    \phi_j = \frac{T_j}{\sum_{r=1}^dT_r}, \qquad j=1,\dots,d.
\end{equation}
For example, taking $T_j$'s to be independent gamma random variables with a common rate parameter induces a Dirichlet distribution on  $\boldsymbol{\phi}$ \citep[e.g.,][]{kruijer_adaptive_2010,lin_dirichlet_2016,zhang_bayesian_2022,aguilar_generalized_2025}. Alternatively, we can consider a latent Gaussian vector $\boldsymbol{\eta} \sim \mathsf{Normal} (\boldsymbol{\mu}_{\boldsymbol{\eta}}, \boldsymbol{\Sigma}_{\boldsymbol{\eta}})$ and map it to the simplex through the softmax operator, thus setting $T_j=\exp(\eta_j)$. The resulting distribution  of $\boldsymbol{\phi}$ is a logistic-normal \citep{aitchison_logistic-normal_1980, aguilar_generalized_2025, yanchenko_group_2026}. Through the covariance matrix $\boldsymbol{\Sigma}_{\eta}$, the logistic-normal distribution can accommodate a richer range of dependence structures among the components of $\boldsymbol{\phi}$ than the Dirichlet distribution, whose covariance structure is fully determined by its mean vector and concentration vector.

Although constraints could be imposed directly on $\boldsymbol{\phi}$, doing so entangles the scientific restrictions with the geometry of the simplex. In particular, the constraints must be compatible with both positivity and the unit-sum condition. Restricting any one component necessarily changes the feasible ranges of the others. Such restrictions can therefore alter the marginal behavior and dependence structure of the entire composition in ways that are difficult to anticipate. A construction based on positive variables $(T_1,\dots,T_d)$ does not fully resolve this issue. Even when the joint distribution is specified flexibly, normalization can obscure the implications for the induced distribution of $\boldsymbol{\phi}$ \citep{kruijer_adaptive_2010}. 

For these reasons, we adopt the logistic-normal construction and impose constraints on the latent Gaussian parameter $\boldsymbol{\eta}$ rather than directly on $\boldsymbol{\phi}$. The softmax transformation automatically enforces positivity and the unit-sum condition, thereby separating the compositional normalization from the scientific restrictions of interest. Moreover, the identity
$
\log\left(\phi_j/\phi_k\right)
=
\eta_j-\eta_k
$ provides direct control over relative allocations. Consequently, ordering constraints, bounds on relative magnitudes, and more general multiplicative relationships among the components of $\boldsymbol{\phi}$ can often be represented as linear constraints on contrasts of $\boldsymbol{\eta}$ \citep{aguilar_generalized_2025}. Finally, imposing restrictions on $\boldsymbol\eta$ directly, does not alter the  distribution of $R^2$ unlike working directly with the regression coefficients $\boldsymbol{\beta}$. 

\subsection{Motivating examples}
\label{subsec:examples}

We now describe several modelling situations in which constraints on the latent allocation parameters naturally arise.

\paragraph{Domain expertise.}
 One of the simplest examples comes from domain information where practitioners believe that the prior variance assigned to one term should be bounded by that assigned to another. In the latent scale, this means
\begin{equation}
\label{eq:domain-ordering}
    \eta_j\leq\eta_{k}.
\end{equation}
Since the softmax denominator is common to all coordinates, this is equivalent to $\phi_j\leq\phi_{k}$. More generally, if the prior allocation of term $j$ is expected to be at most $c>0$ times the allocation of term $k$, then we may impose
\begin{equation}
\label{eq:domain-ratio}
    \eta_j-\eta_{k}\leq \log c,
\end{equation}
which is equivalent to $\phi_j\leq c\phi_{k}$.

\paragraph{Heredity in interaction models.}
Interaction models are another common source of variance constraints \citep{chipman_bayesian_1997, griffin_hierarchical_2017}. Consider a regression model containing all main effects and second-order interactions,
\begin{equation}
\label{eq:interaction-model}
    y_i  = \sum_{j=1}^p x_{ij}\beta_j + \sum_{j<k} x_{ij}x_{ik}\beta_{jk} + \varepsilon_i.
\end{equation}
Letting
$
    \bz_i = \left(x_{i1},\dots,x_{ip},x_{i1}x_{i2},x_{i1}x_{i3},\dots,x_{i,p-1}x_{ip}\right)^\prime
$
and the vector of coefficients 
$\bbeta = \left(\beta_1,\dots,\beta_p,\beta_{12},\beta_{13},\dots,\beta_{p-1\,p}\right)^\prime
$
places this model in the general form of \eqref{eq:resp}, with
$d=p+p(p-1)/2$ model terms.

Heredity expresses the prior belief that an interaction should receive considerable signal only when the corresponding main effects do \citep{Haris_Witten_Interactions_2016}. On the latent scale, we encode this belief by requiring the interaction parameter to be no greater than the average of the corresponding main-effect parameters:
\begin{equation}\label{eq:latent-heredity}
    \eta_{jk}
    \leq
    \frac{1}{2}(\eta_j+\eta_k),
    \qquad j<k,
\end{equation}
which is equivalent to  $\phi_{jk} \leq \sqrt{\phi_j\phi_k}$ on the allocation scale. 
Therefore the interaction term can only be considered if both main effects are present, a condition known as strong heredity \citep{roshan_interactions}.


\paragraph{Generalised Additive Models (GAMs).}

Generalized additive models provide another setting in which structured prior allocations arise \citep{hastie_generalized_1986,griffin_hierarchical_2017}.  Consider a GAM
\begin{equation}
\label{eq:gam-model}
    y_i = \sum_{j=1}^p f_j(x_{ij})+\varepsilon_i
\end{equation}
where $f_j(\cdot)$ are unknown functions. A common representation uses basis expansions of the form
\begin{equation}
 f_j(x_{ij})  = x_{ij}\beta_j + \sum_{\ell=1}^{L_j}\beta_{j\ell}B_{j\ell}(x_{ij}),
\end{equation}
where the basis functions $B_{j\ell}(\cdot)$ are chosen to represent departures from linearity.  Collecting the linear terms $x_{ij}$ and nonlinear basis terms $B_{j\ell}(x_{ij})$ in $\bz_i$ shows that this model is a special case of \eqref{eq:resp}.

A natural hierarchy principle is that each nonlinear basis term should receive no more variance allocation than the corresponding linear term. On the latent scale, this can be written as
\begin{equation}
\label{eq:gam-individual-hierarchy}
    \eta_{j\ell}\leq\eta_j,
    \qquad j=1,\dots,p,\quad \ell=1,\dots,L_j,
\end{equation}
which is equivalent to $\phi_{j\ell}\leq\phi_j$. A group-level version constrains the average latent nonlinear allocation,
\begin{equation}\label{eq:gam-average-hierarchy}
    \frac{1}{L_j}\sum_{\ell=1}^{L_j}\eta_{j\ell}
    \leq
    \eta_j,
    \qquad j=1,\dots,p.
\end{equation}
This implies
\begin{equation}\label{eq:gam-geometric-hierarchy}
    \left(\prod_{\ell=1}^{L_j}\phi_{j\ell}\right)^{1/L_j}
    \leq
    \phi_j,
\end{equation}
that is, the geometric mean of the nonlinear allocations is bounded by the corresponding linear allocation.  

\subsection{Proposed prior}
\label{sec:proposed-prior}
For a chosen collection of structural relationships like those described above, let $\mathcal{E}\subseteq\mathbb{R}^d$ denote the set of latent allocation vectors that satisfy them. We write  $\mathsf{Normal}_{\mathcal{E}} (\boldsymbol{\mu}_{\eta},\boldsymbol{\Sigma}_{\eta})$ to denote the smoothly constrained Gaussian prior on $\boldsymbol{\eta}$ which incorporates the relationships defining $\mathcal{E}$. We show how these relationships are imposed in Section~\ref{sec:comp}.

Combining this latent distribution with the R2 hierarchy gives the proposed prior
\begin{equation}
\label{eq:proposed-prior}
\begin{aligned}
\beta_i&\mid\sigma^2,\phi_i,\omega^2 \stackrel{\text{ind.}}{\sim}\mathsf{Normal}(0,\sigma^2\phi_i \omega^2/\rho_i)
&& i=1,\dots,d,\\
\boldsymbol{\eta} &\sim \mathsf{Normal}_{\mathcal{E}} \left( \boldsymbol{\mu}_{\eta}, \boldsymbol{\Sigma}_{\eta} \right), \ \phi_i = \exp(\eta_i)/\sum_{r=1}^d\exp(\eta_r) && i=1,\dots,d,\\
\omega^2&\sim\mathsf{BetaPrime}(a,b), \ \sigma^2\sim \pi(\sigma^2)
\end{aligned}
\end{equation}
The softmax transformation induces a corresponding smoothly constrained logistic-normal distribution on $\boldsymbol{\phi}$, which we denote by $\mathsf{LN}_{\mathcal{E}} (\boldsymbol{\mu}_{\eta},\boldsymbol{\Sigma}_{\eta})$. When the population second moments $\rho_j$ are unavailable, we use $\widehat{\rho}_j = \frac{1}{n}\sum_{i=1}^n z_{ij}^2$ for  $j=1,\dots,d$. We refer to the resulting hierarchy as the $\Rtdprox$ prior.

\section{Computation}\label{sec:comp}

The hierarchy shown in Section\ref{sec:proposed-prior} requires a prior for the latent allocation vector that can encode the structural relationships represented by $\mathcal E$. This could, for instance, be achieved by truncating the prior directly to $\mathcal E$ and attempting to sample from it \citep{pakman_exact_2014,botev_normal_2017}, possibly even by using modern probabilistic programming languages \citep{ge_turing_2018,abril-pla_pymc_2023,stan_development_team_stan_2026}. We prefer, however, to induce a smooth version of the prior that favours configurations in $\mathcal E$ while retaining support outside it \citep{duan_bayesian_constraint_2020,zhou_proximal_2024}. Not only does direct truncation suffer from inefficient sampling in priors with hierarchical structure, but the soft enforcement safeguards against constraint misspecification by expressing the proposed structural information as a prior preference without ruling out alternative configurations. Equally important, this approach also yields a differentiable posterior that can be explored using state-of-the-art gradient-based MCMC methods \citep{neal_mcmc_2011,hoffman_no-u-turn_2014, zhou_proximal_2024}.

\subsection{Proximal MCMC}


We first describe an exact enforcement of the constraints before presenting the (smoothed) soft formulation. Let  $\btheta$ collect  all model parameters, $\btau$ the parameters that are subject to constraints and $\bnu$ be the unconstrained parameters such that $\btheta = (\btau', \bnu')'$. In our case, $\btau$ is the set of latent variables $\eta_j$'s which are constrained while $\bnu$ refers to all other parameters. We define the feasible set $\mathcal E\subseteq\mathbb R^d$ as the collection of parameters for which the structural relationships of interest are satisfied. Define the indicator function of $\mathcal{E}$, denoted as $\delta_\mathcal{E}(\bt)$ as
\begin{equation}
\label{eq:indicator}
    \delta_\mathcal{E}(\bt) =  
    \begin{cases}
        0 & \bt \in \mathcal{E},\\
        \infty & \bt \notin \mathcal{E}.
    \end{cases}
\end{equation}
Let $\pi(\btau)$ denote the prior for $\btau$ in the unconstrained setting which assigns positive mass to  $\mathcal E$. The corresponding hard-constrained prior is $\pi(\btau)\exp\left\lbrace-\delta_\mathcal{E}(\btau)\right \rbrace$. Since $ \exp\left\lbrace -\delta_\mathcal{E}(\btau) \right\rbrace =  \mathbf{1}_{\mathcal E}(\btau)$ this expression is simply $\pi(\btau)\mathbf{1}_{\mathcal E}(\btau)$, that is, the prior density restricted to the feasible set $\mathcal E$. Consequently, the prior and posterior assign zero density to parameter values outside $\mathcal E$. Working on the log scale, the constrained posterior of interest becomes
\begin{equation}
\label{eq:logposterior_approximation}
\begin{aligned}
    \log \pi( \btheta \mid \byobs) &\propto \log p(\byobs \mid \btheta ) + \log \pi(\btheta) 
    \propto \log p(\byobs \mid \btheta ) + \log \pi(\btau, \bnu)   - \delta_\mathcal{E}(\btau).
\end{aligned}
\end{equation}



The hard constraints included in the posterior complicate the sampling problem. Specialised algorithms are available for particular constrained targets, including truncated multivariate Gaussian distributions \citep{pakman_exact_2014, botev_normal_2017}, but a general sampler for, e.g., the full hierarchical posterior considered here, is unavailable. Moreover, the indicator introduces a nondifferentiable boundary that prevents direct use of gradient-based MCMC such as Hamiltonian Monte Carlo (HMC) \citep{neal_mcmc_2011}. Modern probabilistic programming languages such as Stan \citep{stan_development_team_stan_2026} can bypass this boundary for simple parameter-dependent bounds by transforming individual coordinates to an unconstrained space. However, these coordinate-wise transformations do not extend to general systems of linear inequalities in which parameters are constrained jointly. Moreover, they still enforce a hard truncation that assigns zero support outside $\mathcal{E}$.

To address these challenges and soften exact enforcement, we  adopt proximal methods \citep{pereyra_proximal_2016, durmus_efficientbayesiancomputationproximal_2016, zhou_proximal_2024}, replacing the indicator $\delta_\mathcal{E}(\btau)$ by a smooth approximation given by its Moreau-Yosida envelope \citep{neal_boyd_proximal_popular_2014,polson_proximal_2015,murphy_probabilistic_2022}. In this way, we quantify departures from the proposed structural relationships through the Euclidean distance to $\mathcal E$. 

Let $g : \Theta \to \mathbb{R} \cup \{+\infty\}$ be a proper, lower semicontinuous, and convex function.  For $\lambda > 0$, its Moreau–Yosida envelope is defined as
\begin{equation}
\label{eq:moreau-yosida-envelope}
    g^\lambda(\bt) = \inf_{\bs} \left\{ g(\bs) + \frac{1}{2\lambda}\|\bs-\bt\|^2 \right\}.
\end{equation}
The quadratic term renders the objective strongly convex, so the minimiser is unique. Moreover, $g^\lambda$ is continuously differentiable even when $g$ is not. The parameter $\lambda$ controls the amount of smoothing, where small values of $\lambda$ give a sharper approximation to $g$, while larger values produce a smoother but less faithful approximation \citep{polson_proximal_2015}.  

In our setting, $g=\delta_{\mathcal E}$, where $\mathcal E$ is the feasible set defined above. For a nonempty, closed, and convex set $\mathcal E$, the minimiser in \eqref{eq:moreau-yosida-envelope} is the Euclidean projection \citep{pereyra_proximal_2016} which we denote by $\mathcal{P}_\mathcal{E}$,
\begin{equation}
\label{eq:projection-operator}
\mathcal{P}_\mathcal{E}(\bt) =  \displaystyle \argmin_{\bs \in \mathcal{E}} \|  \bs - \bt   \|.
\end{equation}
In words, we find the point closest to $\bt$ that lies in the  constrained (allowable) space. The Moreau-Yosida envelope, $\delta_{\mathcal E}^{\lambda}(\cdot)$, of the indicator function $\delta_\mathcal{E}(\cdot)$ is then
\begin{equation}
\label{eq:my-envelope-indicator}
\delta_\mathcal{E}^\lambda(\bt) = \frac{1}{2 \lambda} \|  \bt - \mathcal{P}_\mathcal{E}(\bt)   \|^2,
\end{equation}
which vanishes when $\bt \in \mathcal E$ and increases quadratically with the distance from $\mathcal E$ otherwise. 

\cite{zhou_proximal_2024} propose a smooth, constrained posterior $\pi^\lambda(\btheta\mid\byobs)$ by replacing $\delta_\mathcal{E}$ with its Moreau-Yosida envelope, $\delta^\lambda_\mathcal{E}$, i.e.,
\begin{equation}
\label{eq:smoothed-posterior}
\begin{aligned}
    \log \pi^\lambda( \btheta \mid \byobs) &\propto \log p(\byobs \mid \btheta ) + \log \pi(\btau, \bnu)   - \delta^\lambda_\mathcal{E}(\btau).
\end{aligned}
\end{equation}

If $\bt\in\mathcal E$, then $\mathcal P_{\mathcal E}(\bt)=\bt$, so $\delta_{\mathcal E}^{\lambda}(\bt)=0$ meaning that feasible configurations receive no distance-based adjustment. Outside $\mathcal E$, parameter values remain in the support inherited from the unconstrained posterior but are increasingly penalized with their distance from the feasible set. Importantly, we use the projection $\mathcal{P}_\mathcal{E}(\cdot)$ only to evaluate the penalty $\delta_\mathcal{E}^\lambda(\cdot)$, since the sampled parameter is not replaced by the projection. If the unconstrained posterior is proper and assigns positive probability to $\mathcal E$, the softly constrained posterior is proper as well for every $\lambda>0$ and one can show that it converges in total variation to the hard-constrained posterior as $\lambda \downarrow 0$ \citep{zhou_proximal_2024}.

Soft enforcement therefore treats the relationships represented by $\mathcal E$ as substantive prior information rather than as inviolable properties of the model. This provides a safeguard against misspecification of the constraints, where configurations that violate the proposed relationships are discouraged, but can still receive posterior support when favoured by the likelihood. Moreover, replacing the nondifferentiable indicator with its differentiable Moreau--Yosida envelope makes the constraint component compatible with gradient-based samplers. Finally, we stress that we are not proposing a new MCMC algorithm, but rather applying the proximal MCMC construction to the latent parameters under constraints of the R2 prior \citep{pereyra_proximal_2016,duan_bayesian_constraint_2020, zhou_proximal_2024}. 

\subsection{Computing the projections}
\label{sec:projections}

The principal computation introduced by the softly constrained
prior is the evaluation of $\mathcal P_{\mathcal E}(\cdot)$. Because this projection determines both the Moreau--Yosida envelope and its gradient, it must be computed whenever the log posterior or its gradient is evaluated. Although the projection problem admits a closed-form solution for some constraint sets, in general it must be solved numerically
\citep{chaari_proximal_hmc_2015,guo_proximal_hmc_2024,
shukla_proximal_2025}. We consider both cases below before specializing to the  constraints introduced in Section~\ref{subsec:examples}.

\paragraph{Closed-form projections.}

For many convex sets $\mathcal E$ of practical interest, a closed-form projection exists. For the Euclidean ball $\mathcal E=\{\bs\in\mathbb R^d:\|\bs\|_2\leq r\}$, the projection is the rescaling $\mathcal{P}_{\mathcal{E}}(\bt) = r\,\bt / \max(r, \|\bt\|)$. For the box $\mathcal E= \{\bs \in \mathbb{R}^d \mid \ell_i \le s_i \le u_i\}$, the projection acts coordinate-wise, $[\mathcal{P}_{\mathcal{E}}(\bt)]_i = \max(\ell_i, \min(u_i, t_i))$. Projection onto an $\ell_1$ ball is also possible \citep{murphy_probabilistic_2022}. 

\paragraph{Quadratic programs.} 

When no closed-form expression is available, the projection must be obtained by solving an optimization problem. An important class for the present applications is formed by quadratic programs, which combine a quadratic objective with linear equality or inequality constraints \citep{neal_boyd_proximal_popular_2014}. Given $\bs  \in \mathbb{R}^d$ and a system of $m$ linear inequality constraints, a quadratic program has the form
\begin{equation}
\label{eq:quadratic_program}
\begin{aligned}
\min_{\bs \in \mathbb R^d} \quad 
& \tfrac{1}{2}\bs' \bQ \bs + \bc' \bs \\
\text{subject to} \quad 
& \bG' \bs \le \bh,
\end{aligned}
\end{equation}
where $\bQ \in \mathbb{R}^{d \times d}$ is positive semidefinite, $\bc \in \mathbb{R}^{d}$, $\bG \in \mathbb{R}^{d \times m}$, $\bh \in \mathbb{R}^m$.  When $\bQ$ is positive definite, the objective is strictly convex and the problem admits a unique minimiser \citep{murphy_probabilistic_2022}. 

All constraints described in Section~\ref{subsec:examples} can be expressed in this form.  In particular, if $\mathcal E=\{\bs\in\mathbb R^d:\bG^\prime\bs\leq\bh\}$, then projecting a point $\bt$ onto $\mathcal E$ gives Equation~\eqref{eq:quadratic_program} with $\bQ=\mathbf I_d$ and $\bc=-\bt$ (after dropping the constant $\tfrac12\bt^\prime\bt$). Thus, the projection itself is a strictly convex quadratic program. More specifically, for the interaction heredity constraint in \eqref{eq:latent-heredity}, let $d =p +\tfrac{1}{2}p(p-1)$ denote the total number of main effects and pairwise interactions. There is one constraint for each interaction, so $m=p(p-1)/2=d-p$. Taking $\bh=\mathbf{0}$, the column $\bg_{jk}$ of $\bG$ associated with interaction $(j,k)$ has entries $-1$ in the coordinates of $\eta_j$ and $\eta_k$, an entry $2$ in the coordinate of $\eta_{jk}$, and zeros elsewhere. The corresponding inequality is $\bg_{jk}^\prime\boldsymbol\eta =-\eta_j-\eta_k+2\eta_{jk}\leq0$.

To solve quadratic programs of the form \eqref{eq:quadratic_program}, we make use of Dykstra's algorithm \citep{dykstra_algorithm_1983}. This algorithm applies projections onto an intersection of convex sets $\mathcal{E} = \bigcap_i \mathcal{E}_i$ by iteratively projecting onto each individual set $\mathcal{E}_i$, adding correction terms to guarantee convergence to the true projection. The efficiency of the algorithm rests on the fact that projecting onto a single half-space $\mathcal{E}_i = \{\bs \in \mathbb{R}^d \mid \bg_i'\bs \leq h_i\}$ admits a closed-form solution. 
Given a point $\bt \in \mathbb{R}^d$, the projection is
\begin{equation}
\label{eq:halfspace-projection}
\mathcal{P}_{\mathcal{E}_i}(\bt) = \bt - \frac{\max \{0,\, \bg_i^\prime\bt-h_i \}}{\|\bg_i\|^2}\,\bg_i.
\end{equation}
If $\bt$ already satisfies the constraint $\bg_i'\bt \le 0$, the projection leaves it unchanged; otherwise,  it moves $\bt$ along the perpendicular direction $\bg_i$ to the boundary hyperplane. Since each subproblem reduces to evaluating Equation~\eqref{eq:halfspace-projection}, the algorithm is well suited to high-dimensional settings. The projection steps described above are summarised in Figure \ref{fig:proximal-disk}. We also give the full procedure in Algorithm~\ref{alg:dykstra} in Appendix~\ref{sec:supplement} as well as providing more sampling details and further discussing the difference between projecting the parameters into the constrained space and truncating the MCMC samples.

\paragraph{A two-variable heredity example.}
\label{para:two_var_example}

The interaction model with $p = 2$ has only one constraint, so the distance to the feasible set can be written in closed form, providing useful geometric intuition.  Let $\boldsymbol{\eta} = (\eta_1, \eta_2, \eta_{12})^\prime$ be the latent allocation vector for the two main effects and their interaction respectively. The heredity constraint~\eqref{eq:latent-heredity} reduces to a single inequality $-\eta_1-\eta_2+2\eta_{12} \leq 0$, defining the half-space $\mathcal{E} = \left\lbrace
\bs\in\mathbb R^3:
\bg^\prime\bs\leq0
\right\rbrace $ with $\bg = (-1, -1, 2)^\prime$ and $\|\bg\|^2 = 6$. 
When the constraint is violated ($\boldsymbol{\eta}\notin\mathcal E$), let $v= \bg^\prime\boldsymbol{\eta} =-\eta_1-\eta_2+2\eta_{12}>0$ denote the amount of violation. A single application of \eqref{eq:halfspace-projection} gives the projection
\begin{equation}
\label{eq:projection-p2}
\mathcal{P}_{\mathcal{E}}(\boldsymbol{\eta}) 
= \boldsymbol{\eta} - \frac{\bg'\boldsymbol{\eta}}{\|\bg\|^2}\bg
= \frac{1}{6}
\begin{pmatrix}
    5\eta_1 - \eta_2 + 2\eta_{12} \\
    -\eta_1 + 5\eta_2 + 2\eta_{12} \\
    2\eta_1 + 2\eta_2 + 2\eta_{12}
\end{pmatrix}.
\end{equation}
To reach the nearest point on the constraint boundary, the projection increases each main-effect coordinate symmetrically by $v/6$ while decreasing the interaction coordinate by $v/3$. This yields the projected interaction value of  $(\eta_1+\eta_2+\eta_{12})/3 < \eta_{12}$, which lies on the heredity boundary $\phi_{12}=\sqrt{\phi_1\phi_2}$ on the variance allocation scale.

In the softly constrained model, the projection is used to measure the distance to $\mathcal{E}$ rather than to replace the sampled value of $\boldsymbol{\eta}$. Specifically, the displacement $\boldsymbol{\eta} - \mathcal{P}_{\mathcal{E}}(\boldsymbol{\eta})$ gives the Euclidean distance $\|\boldsymbol{\eta} - \mathcal{P}_{\mathcal{E}}(\boldsymbol{\eta})\| = v/\sqrt{6}$. In this case, the Moreau--Yosida penalty is given by
\begin{equation}
\delta_{\mathcal{E}}^\lambda(\boldsymbol{\eta})
= \frac{\|\boldsymbol{\eta} - \mathcal{P}_{\mathcal{E}}(\boldsymbol{\eta})\|^2}{2\lambda}
= \frac{v^2}{12\lambda}
= \frac{1}{3\lambda}\left\{\log\left(\frac{\phi_{12}}{\sqrt{\phi_1\phi_2}}\right)\right\}^2.
\end{equation}

The penalty therefore grows quadratically with the log-ratio by which the interaction allocation exceeds the geometric mean of the main-effect allocations. During HMC sampling, each leapfrog step outside $\mathcal{E}$ applies the negative gradient $-\nabla\delta_{\mathcal{E}}^\lambda(\boldsymbol{\eta}) = -\lambda^{-1}(\boldsymbol{\eta} - \mathcal{P}_{\mathcal{E}}(\boldsymbol{\eta}))$ \citep{zhou_proximal_2024}, which pulls the interaction coordinate down by $v/(3\lambda)$ and each main-effect coordinate by $v/(6\lambda)$ toward $\mathcal{P}_{\mathcal{E}}(\boldsymbol{\eta})$.

\begin{figure}[t]
\centering
\begin{tikzpicture}[>=Stealth,scale=1,line cap=round]
 
  \fill[cFeas!8] (0,0) circle (2);
  \draw[cPen!45,dashed,thin] (0,0) circle (2.8);
  \draw[cPen!65,dashed,thin] (0,0) circle (2.4);
  \draw[cFeas,thick]         (0,0) circle (2);
 
  \draw[black!40,->] (-2.7,0) -- (3.9,0);
  \draw[black!40,->] (0,-2.7) -- (0,3.0);
  \node[right=1pt] at (3.9,0) {$\eta_1$};
  \node[above=1pt] at (0,3.0) {$\eta_2$};
 
  \begin{scope}[shift={(2.9,1.7)},rotate=-25]
    \foreach \aa/\bb in {0.35/0.2, 0.7/0.4, 1.05/0.6, 1.35/0.78} {
      \draw[cPost,thin] (0,0) ellipse [x radius=\aa cm, y radius=\bb cm];
    }
  \end{scope}
  \fill[cPost] (2.9,1.7) circle (1.4pt);   
  \node[cPost!80!black,font=\small] at (3.1,2.95) {unconstrained posterior};
 
  \draw[cTraj,dashed,semithick,->,shorten >=2.6pt]
    plot[smooth,tension=0.8] coordinates
    {(-0.7,0.1)(0.4,0.9)(1.4,1.4)(2.2,1.25)(2.55,0.6)(2.45,-0.1)(2.05,-0.55)};
 
  \draw[cGrad,thick,->] (2.2,1.25)  -- (1.81,1.03);
  \draw[cGrad,thick,->] (2.55,0.6)  -- (2.11,0.50);
  \draw[cGrad,thick,->] (2.45,-0.1) -- (2.00,-0.08);
  \node[cGrad!85!black,font=\small] at (1.15,0.35) {$-\nabla\delta^{\lambda}_{\mathcal{E}}$};
 
  \fill[cTraj] (-0.7,0.1) circle (2.2pt);
  \node[below=2pt,font=\small] at (-0.7,0.1) {$\boldsymbol{\eta}^{(t)}$};
  \draw[cTraj,fill=white,semithick] (2.05,-0.55) circle (2.2pt);
  \draw[black!55,thin] (2.13,-0.64) -- (2.72,-1.12);
  \node[anchor=west,font=\small,inner sep=2pt] at (2.72,-1.16)
        {proposal $\tilde{\boldsymbol{\eta}}$};
 
  \draw[black!55,thin] (-2.12,-2.22) -- (-1.72,-1.72);
  \node[cPen!80!black,font=\small] at (-2.55,-2.45)
        {smoothed penalty $\delta^{\lambda}_{\mathcal{E}}$};
 
\end{tikzpicture}

\caption{Hamiltonian Monte Carlo under a Moreau--Yosida circular constraint. Gray ellipses show the unconstrained posterior, whose mode lies outside
the feasible disk $\mathcal{E}=\{\boldsymbol{\eta}:\lVert\boldsymbol{\eta}\rVert
\le r\}$. The dashed circles are level sets of the Moreau--Yosida penalty $\delta^{\lambda}_{\mathcal{E}}$, which are concentric with $\mathcal{E}$. Their spacing is set by $\lambda$ (a small $\lambda$ shows closer level sets, whereas a large $\lambda$ shows spaced levels). The Hamiltonian trajectory leaves the feasible region and is redirected by the gradient of the Moreau–Yosida penalty $-\nabla\delta^{\lambda}_{\mathcal{E}}=-\lambda^{-1}(\boldsymbol{\eta} -\mathcal{P}_{\mathcal{E}}(\boldsymbol{\eta}))$ applied at every step (purple arrows). The proposal $\tilde{\boldsymbol{\eta}}$ may end slightly outside due to soft enforcement. }
\label{fig:proximal-disk}
\end{figure}
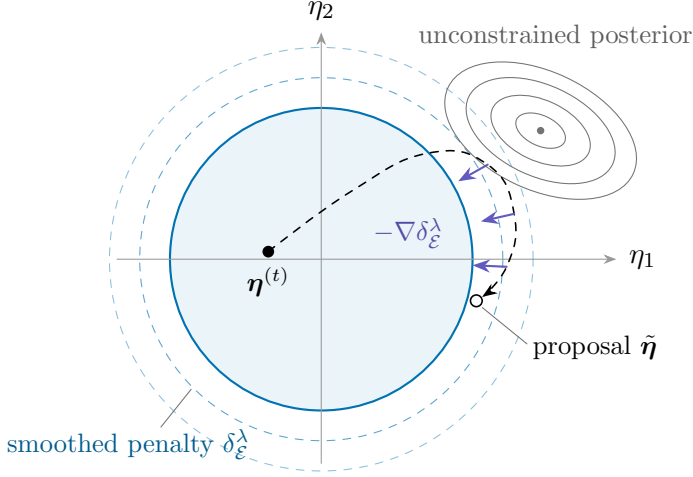

\subsection{Prior calibration of the smoothing parameter $\lambda$}

The value of $\lambda$ is typically specified to control the degree to which the constraints are enforced \citep{shukla_proximal_2025}. This leads practitioners to treat $\lambda$ as a tuning parameter that determines how probability mass is shifted from the unconstrained model to the softly constrained one. Smaller values will shift  mass toward configurations that satisfy the constraints more closely, whereas larger values allow to place more mass on departures (i.e the unconstrained model). We consider that the influence of $\lambda$ can additionally be studied in terms of the effect it has in model itself. To do this we propose quantities that are easier to interpret. We showcase our approach for interactions models.

We can specify $\lambda$ by measuring the largest local departure from heredity at the allocation level, which we denote by $R_{\max} = \max_{j<k}\frac{\phi_{jk}}{\sqrt{\phi_j\phi_k}}.$
The event $R_{\max}\leq1$ is equivalent to satisfying every heredity constraint, whereas $R_{\max}>r$ means that at least one interaction allocation exceeds the geometric mean of its main effects allocations by a factor greater than $r$. Therefore, for a fixed value of $\lambda$, we can use the probability $P(R_{\max}>r)$ to express how strongly the constraint is enforced.
\begin{figure}[t]
\includegraphics[width=\linewidth]{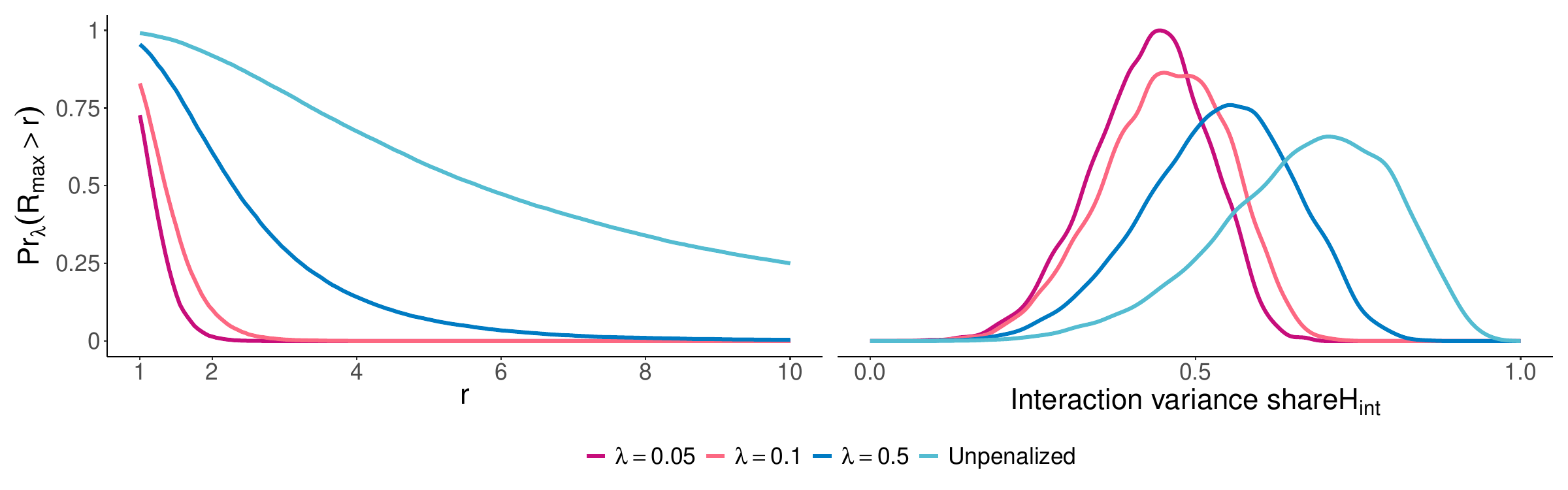}
\caption{Prior implications of the smoothing parameter for an interaction model with $p=5$, $\boldsymbol\mu_\eta=\boldsymbol 0$, and $\boldsymbol\Sigma_\eta=I$. The left panel shows the probability that the largest interaction-to-main-effects allocation ratio exceeds $r$. Panel shows the prior distribution of the interaction share of linear-predictor variance. }
   \label{fig:lambda-interaction-calibration}
\end{figure}

We can also consider the total allocation of variance to the interaction coefficients, namely $H_{\mathrm{int}} = \sum_{j<k}\phi_{jk}$. While $R_{\max}$ measures the largest local violation, the distribution of $H_{\mathrm{int}}$ shows how $\lambda$ shifts prior mass across different levels of total interaction allocation. As $\lambda$ decreases, large multiplicative violations become less probable and the distribution of $H_{\mathrm{int}}$ moves towards smaller values. Under an unconstrained exchangeable prior, we have $\mathbb{E}(H_{\mathrm{int}})= \frac{p-1}{p+1}$. This provides us with a reference point for how the value we specify for $\lambda$ moves the mass that is assigned to all two-way interactions. 

We recommend to evaluate candidate values of $\lambda$ by direct simulation from the prior using the approach of Section~\ref{sec:comp} \citep{gelman_bayesian_2020}. For each candidate, we plot the survival curve of $R_{\max}$ against the frequency and magnitude of heredity violations, while the distribution of $H_{\mathrm{int}}$ describes the corresponding allocation of prior variance. We exemplify this in Figure~\ref{fig:lambda-interaction-calibration}. In practice, one can specify a meaningful violation factor $r_0>1$ and an acceptable probability $\alpha_0$, and  retain values satisfying $P(R_{\max}>r_0)\leq\alpha_0$. The distribution of $H_{\mathrm{int}}$ can then be used to assess the implied balance between main effects and interactions. While smaller values of $\lambda$ enforce the constraints more, higher computational costs and lower quality sampling are bound to appear. Therefore, diagnostics and computation time should be compared across the retained candidates.
\section{Theoretical results}\label{sec:theory}

In this section, we prove various theoretical properties of the proposed $\Rtdprox$ prior including deriving the marginal prior distribution, showing posterior consistency and proving posterior contraction around the true model. All proofs are given in Appendix~\ref{sec:supplement}.

\subsection{Marginal prior distributions}

Shrinkage priors are designed to simultaneously preserve large signals and suppress small ones, and this behaviour is primarily studied through the marginal distributions of the regression coefficients $\beta_i$ \citep{castillo_bayesian_2015,pas_uncertainty_2017,ghosh_asymptotic_2017}. Ideally, the marginal density exhibits heavy tails, so that large coefficients escape shrinkage, and a spike at the origin, so that small coefficients are aggressively shrunk. Characterizing these two features of the marginal distribution is central to understanding whether a shrinkage prior behaves as intended and to understand the effect of hyperparameter choices \citep{pas_conditions_2016}. 

Without loss of generality, we set $\sigma^2 = 1$ and $\rho_i = 1$. 
Under the hierarchy $\beta_i \mid \phi_i, \omega^2 \sim \mathsf{Normal}(0, \phi_i \omega^2)$, $\omega^2 \sim \mathsf{BetaPrime}(a, b)$, and $\phi_i \in (0,1)$ independent of $\omega^2$, integrating over the allocation and global scale gives the marginal density
\begin{equation}
\label{eq:marginal}
p(\beta_i) = \int_0^\infty  \int_0^1 p(\beta_i \mid \phi_i, \omega^2)\, 
    p(\phi_i)\, p(\omega^2)\, d\phi_i\, d\omega^2,
\end{equation}
Characterizing $p(\beta_i)$ is challenging since our procedure renders the marginal density $p(\phi_i)$ not available in closed form. Although the prior on the latent vector $\boldsymbol\eta$ is specified explicitly, the (smoothed) constraints act jointly on its coordinates. Moreover, mapping to the simplex followed by marginalization does not yield an explicit density for $\phi_i$. Our objective, however, is not to derive this density directly, but to determine how different classes of allocation distributions for $\phi_i$ affect the tail and origin behaviour of $p(\beta_i)$. We consider allocations with polynomial boundary behaviour, such as the beta marginals induced by Dirichlet priors, and the constrained logistic-normal allocation used here. To this end, a useful representation is
\begin{equation}
\label{eq:beta-product}
\beta_i=\sqrt{\phi_i}\,\gamma_i, \qquad \gamma_i\mid\omega^2\sim\mathsf{Normal}(0,\omega^2),
\end{equation}
where $\gamma_i$ is independent of $\phi_i$. Thus, the beta-prime scale first determines the shape of the density of $\gamma_i$, and the allocation then randomly rescales it. We use this representation to characterise $p(\beta_i)$, first through its behaviour as $|\beta_i| \to \infty$ (tail behaviour), and then as $|\beta_i| \to 0$ (origin behaviour).
\begin{proposition}
\label{prop:tail}
Let $\omega^2 \sim \mathsf{BetaPrime}(a,b)$ with $a>0$ and $b>0$.  Let $\beta_i \mid \phi_i, \omega^2 \sim \mathsf{Normal}(0, \phi_i \omega^2)$  and $\phi_i \sim p(\phi_i)$ where $\phi_i \in (0,1]$ is a continuous 
random variable independent of $\omega^2$. Then, as $|s| \to \infty$,
\[
p_{\beta_i}(s) \sim C_{a,b}\, \mathbb{E}[\phi_i^{b}]\, |s|^{-(2b+1)},
\]
where $C_{a,b} > 0$.
\end{proposition}
Proposition~\ref{prop:tail} shows that the marginal density of $\beta_i$ has polynomial tails with decay exponent $2b+1$, controlled entirely by the second parameter $b$ of the beta-prime distribution. Smaller values of $b$ produce heavier tails and allow large signals to escape shrinkage more easily. The allocation prior, and hence the structural constraints, affects only the coefficient-specific multiplicative constant through $\mathbb E(\phi_i^b)$; it can change the relative tail weights across coefficients without changing their rate of decay.

We now consider the behaviour near the origin, which depends on how much mass the local allocation scale places near the origin. We first study allocations whose marginal densities behave like a power of $\phi_i$ near the boundary \citep{barndorff-nielsen_normal_1982}. This class includes the beta marginals induced by Dirichlet distributions and provides a direct comparison with existing R2 priors, revealing a competition between the global and local scales. The logistic-normal allocation we use does not have polynomial boundary behaviour because no finite power exponent describes its behaviour as $\phi_i \downarrow 0$ (Appendix~\ref{sec:supplement}). We therefore analyse it separately through a negative-moment condition. We show that, although the logistic-normal allocation does not itself generate a polynomial spike at zero, it does not alter the spike produced by the beta-prime global scale. Consequently, when ($a< 1/2$), the  marginal density of $\beta_i$ remains unbounded at the origin.

We first isolate the contribution of the global scale. The following lemma provides the origin asymptotic and shows how negative moments are involved when the allocation scale dominates.
\begin{lemma}
\label{lem:global-origin}
Let $\gamma\mid\omega^2\sim\mathsf{Normal}(0,\omega^2)$ and $\omega^2\sim\mathsf{BetaPrime}(a,b)$, where $a,b>0$. If $0<a<1/2$, then $p_\gamma(s)\sim A_{a,b}|s|^{2a-1}$, where $A_{a,b} > 0$. Moreover, if $0<c<\min(a,1/2)$, then $\mathbb E|\gamma|^{-2c} < \infty$.
\end{lemma}
Lemma~\ref{lem:global-origin} shows that the beta-prime scale mixture alone has an unbounded density at the origin when $a<1/2$. The next proposition shows how an allocation with polynomial boundary behaviour modifies this exponent.
\begin{proposition}
\label{prop:origin}
Suppose the assumptions of Proposition~\ref{prop:tail} hold and that $\phi_i$ is absolutely continuous with density satisfying $f_{\phi_i}(\phi) \sim C\phi^{\alpha - 1}$ as $\phi \downarrow 0$ for some $C>0, \alpha > 0$. Then as $s \to 0$:
\begin{enumerate}
    \item If $0 < a < \alpha$ and $a < 1/2$, then 
        $p_{\beta_i}(s) \sim A_{a,b}\, \mathbb{E}[\phi_i^{-a}]\, |s|^{2a-1}$.
    \item If $0 < \alpha < a$ and $\alpha < 1/2$, then 
        $p_{\beta_i}(s) \sim K_{a,b,\alpha}\, |s|^{2\alpha - 1}$.
    \item If $\alpha = a < 1/2$, then 
        $p_{\beta_i}(s) \sim D\, |s|^{2a-1} \log(1/|s|)$.
\end{enumerate}
Here  $A_{a,b}, K_{a,b,\alpha},D>0$ are constants.
\end{proposition}
Proposition~\ref{prop:origin} covers, for example, Dirichlet allocations because their marginals are beta distributed; thereby recovering the results of \cite{zhang_bayesian_2022} and \cite{aguilar_intuitive_2023}. In particular, the minimum of $a$ and $\alpha$ determines the polynomial exponent. When $a = \alpha$, the two mechanisms introduce an additional logarithmic factor.

This proposition is a comparison result for allocation densities with polynomial behaviour near the boundary, but it does not apply to the logistic-normal distribution, whose boundary behaviour cannot be represented by a finite value of $\alpha$ (Appendix~\ref{sec:supplement}). Instead, the relevant property is that all negative moments of a logistic-normal coordinate are finite. This ensures that the allocation does not alter the origin behaviour generated by the beta-prime global scale. We make this statement precise in the next two results.
\begin{lemma}
\label{lem:neg-moments}
Let $\phi \sim \mathsf{LN}^\lambda_\mathcal{E}(\boldsymbol{\mu}_\eta, \boldsymbol{\Sigma}_{\boldsymbol{\eta}})  $. Then for every $c>0$,  $\mathbb{E}\left[\phi_i^{-c}\right]<\infty$. 
\end{lemma}
The following result shows that the global-scale's influence on shrinkage is preserved.
\begin{proposition}
\label{prop:origin-ln}
Under the setup of Proposition~\ref{prop:tail} and~\ref{lem:neg-moments}, suppose that $0<a<1/2$ and $\mathbb{E}\left[\phi_i^{-a}\right]<\infty$. Then, as $s\to0$,
\begin{equation}
    p_{\beta_i}(s) \sim A_{a,b}\, \mathbb{E}\left[\phi_i^{-a}\right] |s|^{2a-1}.
\end{equation}
\end{proposition}

Combining Lemma~\ref{lem:neg-moments} with Propositions~\ref{prop:tail} and~\ref{prop:origin-ln} yields the desired characterization of the marginal prior actually used in this paper.

\begin{corollary}{}
\label{cor:marginal-ln}
Let $\bphi \sim \mathsf{LN}^\lambda_\mathcal{E}(\boldsymbol{\mu}_\eta, \boldsymbol{\Sigma}_{\boldsymbol{\eta}}) $ and let $\omega^2\sim\mathsf{BetaPrime}(a,b)$. Then,
\begin{equation}
p_{\beta_i}(s)\sim C_{a,b}\,\mathbb E[\phi_i^{\,b}]\,|s|^{-(2b+1)} \quad (|s|\to\infty).
\end{equation}
When $0 < a < \tfrac{1}{2}$, 
\begin{equation}
\qquad p_{\beta_i}(s)\sim A_{a,b}\,\mathbb E[\phi_i^{-a}]\,|s|^{2a-1} \quad (s\to0).
\end{equation}
In particular, the marginal prior is unbounded at the origin with a polynomial rate.
\end{corollary}

Taken together, these results show that the tail and origin behaviour of $p(\beta_i)$ are governed by different parts of the hierarchy, and that neither is compromised by the constraints. The hyperparameter $b$ exclusively controls the tail decay, while, for the logistic-normal prior, the spike at the origin is produced entirely by the global scale through $a<1/2$. Although the (smooth) logistic-normal does not itself generate a polynomial spike near $\phi_i = 0$, its finite negative moments ensure that we can still preserve the spike generated by the beta-prime global scale.  In the absence of constraints, these results also characterise the origin and tail behaviour of the logistic-normal R2 prior proposed by \citet{aguilar_generalized_2025}. 
This approach differs from global--local constructions in which the local scales can generate the singularity at zero \citep{pas_conditions_2016,bai_beta_2021}. 

We can use these results to translate prior beliefs about sparsity and signal magnitude into choices of $(a,b)$. A value $a<1/2$ combined with a relatively small $b$ favours many negligible coefficients while retaining heavy tails for potentially large signals, whereas weaker concentration near zero and lighter tails may be more appropriate for a denser collection of moderate effects. Candidate choices can be compared through their induced shrinkage factors and effective model sizes, which summarise the strength and extent of shrinkage in more interpretable terms \citep{piironen_sparsity_2017,aguilar2025dependencyawareshrinkagepriorshigh}. The allocation distribution and structural constraints then determine where the available prior variance is concentrated, such as between main effects and interactions.

\subsection{Posterior consistency}

We have shown that the constrained model still preserves the desired marginal shrinkage behaviour. We now establish strong posterior consistency for the smoothed constrained prior as the number of terms $d_n$ may grow with the sample size. Our result ensures asymptotic concentration around the true coefficient vector. We also show that our prior is robust to structural model misspecification, understood here as disagreement between the pattern favoured by the prior and the true signal. For instance, heredity, may be a plausible prior principle even when the true coefficient vector contains interactions whose parent effects are weak or even absent. 

We describe the growing-dimensional setting through a triangular-array version of \eqref{eq:resp}. That is, each sample size $n$ represents its own regression problem, in which the response vector, design matrix, and the dimension of the coefficient vector may change with $n$ \citep{Petrov_sums_1975,ghosh_bayesian_2003}. We also let the number of constraints vary with $n$, implying that the smoothing parameter $\lambda_n$ is free to change as well. In the following, we denote the unconstrained, softly constrained, and hard-constrained prior distributions by $\Pi$, $\Pi^\lambda$, and $\Pi_{\mathcal E}$, respectively, with corresponding densities $\pi_0$, $\pi^\lambda$, and $\pi_{\mathcal E}$. Their induced posterior distributions are denoted by $\Pi(,\cdot\mid\byobs)$, $\Pi^\lambda(,\cdot\mid\byobs)$, and $\Pi_{\mathcal E}(,\cdot\mid\byobs)$, with the same convention used for their densities.

Let
\begin{equation}
\label{eq:model-array}
    \byobs_n = \bZ_n\bbeta^*_n+\bvarepsilon_n,
    \qquad
    \bvarepsilon_n\sim\mathsf{Normal}(0,\sigma^2 I_n),
\end{equation}
where $\bZ_n\in\mathbb R^{n\times d_n}$ collects the (possibly transformed) model terms and $\sigma^2$ is known. The true coefficient vector is a sequence $\bbeta^*=\{\bbeta_n^*\}_{n\geq1}$ with $\bbeta_n^*\in\mathbb R^{d_n}$. Let $\Pr_{\bbeta^*}$ denote the joint sampling law of the data sequence $\{\byobs_n\}_{n\geq1}$. For each $n$, let $\Pi_n$ denote the marginal prior on $\bbeta_n$ induced by the hierarchy of Section~\ref{sec:comp}, i.e.,
\begin{equation}
\label{eq:beta-marginal-prior}
    \beta_{nj}\mid\bphi,\omega^2 \sim \mathsf{Normal}(0,\sigma^2\phi_j\omega^2/\rho_{nj}),     \qquad 
    \omega^2\sim\mathsf{BetaPrime}(a,b), \qquad     \bphi=\operatorname{softmax}(\boldsymbol\eta),
\end{equation}
 where $\rho_{nj}=n^{-1}\sum_{i=1}^n z_{ij}^2$. Throughout this subsection, $\boldsymbol\eta$ follows the smoothed prior $\Pi^{\lambda_n}$ for some $\lambda_n>0$. The sequence $\lambda_n$ is otherwise unrestricted. 

To establish posterior consistency, we need  the prior to retain enough mass in a shrinking neighbourhood of $\bbeta_n^*$ as the dimension grows \citep{armagan_posterior_2013}. We work under the following assumptions.
 
\begin{assumption}[Dimension]\label{asmp:dimension}
The dimension satisfies $d_n\log n = o(n)$.
\end{assumption}

\begin{assumption}[Design]\label{asmp:design}
Let $\Lambda_{n\min}$ and $\Lambda_{n\max}$ denote the smallest and largest singular values of $\bZ_n$, respectively. There exist $0<\Lambda_{\min}\leq\Lambda_{\max}<\infty$ such that  $\Lambda_{\min}<\liminf_n \Lambda_{n\min}/\sqrt n \le \limsup_n \Lambda_{n\max}/\sqrt n<\Lambda_{\max}$.
\end{assumption}

\begin{assumption}[Bounded truth]\label{asmp:truth}
There exists a constant $M_\beta<\infty$ such that $\sup_{n}\max_{j\leq d_n}|\beta^*_{nj}|\le M_\beta<\infty$.
\end{assumption}

\begin{assumption}[Hyperparameters]\label{asmp:hyper}
The hyperparameters $a,b>0$ are fixed. Moreover, there exists a constant $M_\mu<\infty$ such that $\sup_n\|\boldsymbol\mu_{\eta,n}\|_\infty\leq M_\mu<\infty$ and the eigenvalues of $\Sigma_{\eta,n}$ lie in $[\sigma_L^2,\sigma_U^2]$ for constants $0<\sigma_L^2\le\sigma_U^2<\infty$.
\end{assumption}

\begin{assumption}[Uniform strict feasibility]\label{asmp:feasibility}
There exist $M_\eta<\infty$, independent of $n$, such that for every $n$, there exists a point $\boldeta^*_n\in\mathbb{R}^{d_n}$ satisfying $\|\boldsymbol\eta^*_n\|_\infty\le M_\eta$ such that $\bg_{i}^\prime\boldsymbol\eta^*_n\le-\tfrac{1}{2}\|\bg_{i}\|_1$ for every constraint $i=1,\dots,m_n$, where $m_n$ is the number of constraints.
\end{assumption}

Assumptions~\ref{asmp:dimension}--\ref{asmp:truth} are standard in growing-dimensional consistency arguments \citep{armagan_posterior_2013}. The design condition~\ref{asmp:design} rules out asymptotically collinear or explosive designs and implies $\rho_{nj}\in[\Lambda_{\min}^2,\Lambda_{\max}^2]$ for every $j$ and all sufficiently large $n$, because the $\rho_{nj}$ are diagonal entries of $\bZ_n^\prime\bZ_n/n$. Together with the bounds on the truth and the hyperparameters, these conditions keep the coefficient and latent-allocation at a sufficiently stable size as the dimension grows. We link the regularity conditions to the geometry of the constraint set using uniform strict feasibility. The margin considered implies that $\boldsymbol\eta_n^*$ lies entirely in the feasible set. Thus, the penalty is zero for every $\lambda$ independently of $n$. We impose this condition on the allocation geometry and therefore require neither $\bbeta_n^*$ to be sparse nor force it to follow the structure favoured by the prior.

\begin{theorem}[Strong posterior consistency]
\label{thm:consistency}
Suppose Assumptions~\ref{asmp:dimension}--\ref{asmp:feasibility} hold, and let $\Pi_n$ be the marginal prior on $\bbeta_n$ defined by \eqref{eq:beta-marginal-prior}. If $\boldsymbol\eta$ has the smoothed prior $\Pi^{\lambda_n}$ for an arbitrary sequence $\lambda_n\in(0,\infty)$, then for every $\varepsilon>0$,
\begin{equation}
\Pi_n\left( \bbeta_n: \|\bbeta_n-\bbeta_n^*\|_2 >\varepsilon \,\middle|\, \byobs_n \right) 
 = \int_{R^{d_n}}
 \mathbf{1}\left\{
    \|\boldsymbol t-\bbeta_n^*\|_2>\varepsilon
\right\}
  \pi_n\left(\boldsymbol t\mid\byobs_n\right) \,d\boldsymbol t 
\longrightarrow 0,
\end{equation}
$\Pr_{\bbeta_n^*}$-almost surely as $n\to\infty$.
\end{theorem}

The theorem provides the desired consistency guarantee while clarifying the effect of structural model misspecification. Since the constraints are encoded in the prior on $\bphi$, the posterior of $\bbeta$ can still concentrate around the true coefficient vector. This happens even when the prior favours coefficient configurations that disagree with the specified structure, and remains true whether $\lambda_n$ shrinks, grows, or stays fixed, provided that it is positive.

The proof follows the consistency criterion of \cite{armagan_posterior_2013}, which requires at least $e^{-cn}$ mass in a shrinking Euclidean neighbourhood of the truth. The criterion allows the model dimension to satisfy $d_n=o(n)$. Here we verify a stronger quantitative bound, specifically, for some $\rho>0$ and every $\Delta>0$, we show that
\begin{equation}
-\log\Pi_n\left\{ \bbeta_n: \|\bbeta_n-\bbeta_n^*\|_2 < \Delta n^{-\rho/2} \right\} \leq C_\Delta d_n\log n  = o(n).
\end{equation}
The proof obtains this bound by restricting $\boldsymbol\eta$ to the feasible box supplied by Assumption~\ref{asmp:feasibility}, where every softmax coordinate is of order $d_n^{-1}$. Restricting the global scale $\omega^2$ to $[d_n,2d_n]$ offsets this decay and keeps each conditional coefficient variance bounded above and away from zero. The resulting Gaussian probabilities can then be bounded uniformly over a coordinatewise box of radius $\Delta/\sqrt{nd_n}$ around $\bbeta_n^*$, which is contained in the Euclidean ball of radius $\Delta n^{-1/2}$. Since the Moreau--Yosida penalty vanishes throughout the chosen feasible box, our calculations are unaffected by the rate at which $\lambda_n$ varies.

In shrinkage-prior theory, sparse recovery rates are commonly indexed by $q_n=\#\{j:\beta_{nj}^*\neq0\}$, the number of true nonzero coefficients \citep{castillo_needles_2012,pas_horseshoe_2014,pas_adaptive_2017,zhang_bayesian_2022}. Our  Theorem~\ref{thm:consistency} imposes no condition on $q_n$ and therefore covers both dense sequences, for which $q_n\asymp d_n$, and sparse sequences, for which $q_n\ll d_n$. Both cases remain subject to $d_n\log n=o(n)$ and the full-rank design condition, placing regimes with $d_n\geq n$ outside the scope of the result. Different joint priors for $\bphi$, particularly those that place more mass near sparse corners of the simplex, may support contraction rates indexed by $q_n$ \citep{zhang_bayesian_2022}. Establishing such adaptation under structural constraints remains a direction for future research.

\subsection{Posterior contraction}

We have established that the posterior eventually concentrates around the true coefficient vector, but it remains to show the rate at which this occurs. We strengthen the consistency result by deriving a posterior contraction rate in the same growing-dimensional regime. Let
\begin{equation}
\label{eq:rate}
    \varepsilon_n= \left( \frac{d_n\log n}{n} \right)^{1/2}.
\end{equation}
Under Assumption~\ref{asmp:dimension} we have $\varepsilon_n\to 0$ and $n\varepsilon_n^2=d_n\log n\to\infty$. We study both the Euclidean loss $\|\bbeta_n-\bbeta_n^*\|_2$ and the prediction norm $\|\bbeta_n-\bbeta_n^*\|_{n} \coloneqq n^{-1/2} \| \bZ_n(\bbeta_n-\bbeta_n^*) \|_2$.

\begin{theorem}[Posterior contraction]
\label{thm:contraction}
Suppose Assumptions~\ref{asmp:dimension}--\ref{asmp:feasibility} hold and $\sigma^2$ is known. Let $\Pi_n$ be the marginal prior \eqref{eq:beta-marginal-prior} on $\bbeta_n$ with $\boldsymbol\eta$ distributed according the smoothed prior $\Pi^{\lambda_n}$ for an arbitrary sequence $\lambda_n\in(0,\infty)$. Then there exists a sufficiently large constant $M<\infty$, such that
\begin{equation}
\Pi_n \left( \bbeta_n:\|\bbeta_n -\bbeta^*_n\|_2>M\varepsilon_n\;\middle|\;\byobs_n\right)
\longrightarrow 0,
\end{equation}
in $\Pr_{\bbeta_n^*}$ probability as $n\to\infty$. Likewise, changing $M$ if needed,
\begin{equation}
\Pi_n \left(\bbeta_n:\|\bbeta_n-\bbeta^*_n\|_n>M\varepsilon_n\;\middle|\;\byobs_n\right)
\longrightarrow 0,
\end{equation}
in the same sense.
\end{theorem}

Theorem~\ref{thm:contraction} establishes that the posterior concentrates within Euclidean and prediction neighbourhoods whose radii are proportional to $\varepsilon_n$. Our proof adapts Theorem~1 of \citet{ghosal_convergence_2007} to the present sequence of Gaussian regression models. To apply their Theorem, we verify the relevant conditions and additionally prove the prior-specific requirement that a suitable neighbourhood of $\bbeta_n^*$ receives sufficient prior mass. Finally, $M$ can be chosen independently of the smoothing sequence $\lambda_n$.

The theorem does not impose conditions on the positive smoothing sequence $\lambda_n$. This is, however, an asymptotic result and different values of $\lambda_n$ can still enforce the constraints with different strength, leading to different computational behaviour in finite samples. The result also covers models without structural preferences, i.e  taking $\mathcal E_n=\mathbb R^{d_n}$, making the Moreau--Yosida penalty zero; thus rendering the strict-feasibility assumption vacuous. Under the remaining assumptions, Theorem~\ref{thm:contraction} gives the same rate for the unconstrained logistic-normal R2 prior, hence extending the theoretical characterization of the prior proposed by \citet{aguilar_generalized_2025}.

To place this rate in context, the original R2D2 prior achieves a near-minimax rate of order ${q_n\log(d_n)/n}^{1/2}$ in sparse high-dimensional regression \citep{zhang_bayesian_2022}, while adaptive horseshoe results give related rates in sparse Gaussian sequence models \citep{pas_horseshoe_2014,pas_adaptive_2017,zhang_bayesian_2022}. These rates depend on the number $q_n$ of nonzero coefficients and can therefore be sharper when $q_n\ll d_n$. By contrast, Theorem~\ref{thm:contraction} makes no sparsity assumption, and its rate is indexed by the full dimension $d_n$. Under a well-conditioned design, this rate differs from the classical dense benchmark $(d_n/n)^{1/2}$ by a factor $\sqrt{\log n}$ \citep{lehmann_casella_1998,wainwright_high-dimensional_2019}. Hence for fixed $d_n$, the established contraction bound requires a sample size larger by roughly a logarithmic factor to guarantee the same accuracy as the classical dense benchmark. This results provides a contraction guarantee for the proposed structurally constrained prior under both sparse and dense truths, without requiring sparsity.

\subsection{Global-scale collapse}

\label{subsec:asymptotics}

The preceding results establish asymptotic recovery, but they do not describe the finite-sample shrinkage induced by the global scale. Since every $\phi_j \leq 1$, no coordinate can offset an arbitrarily small value of $\omega^2$. The following proposition quantifies this behaviour uniformly over the allocation simplex. Our results extend the global-scale analysis of grouped global-local shrinkage priors to the interaction setting \citep{boss_group_2023} .

\begin{proposition}
\label{prop:collapse}
Let $\bz\in\mathbb R^{n\times d}$ have full column rank and define $\hat\bbeta^{\mathrm{ols}}=(\bz^\prime \bz)^{-1}\bz^\prime\byobs$. For $\bphi$ in the open simplex and $\omega^2>0$, let $D=\operatorname{diag}(\phi_1\omega^2/\rho_1,\dots,\phi_d\omega^2/\rho_d)$.  The conditional posterior mean $\bar\bbeta=\mathsf E(\bbeta\mid\byobs,\omega^2,\bphi,\sigma^2)=(\bz^\prime \bz+D^{-1})^{-1}\bz^\prime\byobs$ does not depend on $\sigma^2$ and satisfies
\begin{equation}
\|\bar\bbeta\|_2 \le \Lambda_{\max}^2(\bz)\,\frac{\omega^2}{\rho_{\min}}\, \|\hat\bbeta^{\mathrm{ols}}\|_2, \qquad \|\hat\bbeta^{\mathrm{ols}}-\bar\bbeta\|_2 \ge \frac{\|\hat\bbeta^{\mathrm{ols}}\|_2} {1+\Lambda_{\max}^2(\bz)\,\omega^2/\rho_{\min}}.    
\end{equation}
where $\rho_{\min}=\min_{j}\rho_j>0$. Therefore, for every $\varepsilon\in(0,1)$,
\begin{equation}
\omega^2\le\frac{(1-\varepsilon)\,\rho_{\min}}{\varepsilon\,\Lambda_{\max}^2(\bz)} \quad\Longrightarrow\quad \|\hat\bbeta^{\mathrm{ols}} -\bar\bbeta\|_2\ge\varepsilon\,\|\hat\bbeta^{\mathrm{ols}}\|_2,    
\end{equation}
uniformly over the allocation simplex.
\end{proposition}

Unlike standard global--local priors with unbounded local scales (e.g., horseshoe), the simplex constraint in the R2 prior framework places a deterministic upper bound on every local variance. Proposition~\ref{prop:collapse} shows that the global scale retains uniform control over the conditional posterior mean, which collapses to the origin as $\omega^2\to0$ regardless of the allocation values. Nonetheless, this does not imply collapse of the marginal posterior mean, but shows the influence of the global-scale prior in preserving large signals.

\section{Experiments}\label{sec:experiments}
We now study the finite sample performance of $\Rtdprox$ on synthetic data using an interaction model, and real-world data fit with GAMs.

\subsection{Synthetic interaction models}
\label{sec:experiment_exp1}

We generate data from~\eqref{eq:resp} where the data-generating-processes (DGPs) differ by sparsity, dimensionality of the main effect, signal-to-noise and true heredity type. Our primary metrics of interest are the recovery of regression coefficients and predictive performance measured by the expected log predictive density score \citep[$\elpd$; ][]{vehtari_survey_2012}.  

Let $\mathcal{A}_{\mathrm{main}} = \{j \in \{1,\dotsc,p\} : \beta_{\mathrm{main},j} \neq 0\}$ denote the active set of main effects. We consider dense settings where $\#\mathcal{A}_{\mathrm{main}} = p$ and sparse settings with exactly six active main effects.
In the sparse setting, we set $\beta_{\mathrm{main},j} = 0.5$ for $j \in \{1,2,3,p-2,p-1,p\}$ and $\beta_{\mathrm{main},j} = 0$ for the remaining indices. We set $p \in \{10, 15, 20, 25\}$ and true DGP-$R^2 \in \left\{0.1,0.5,0.9\right\}$.  We consider two DGPs that differ in whether heredity holds or not.

\textbf{(H1)} \textit{Heredity on coefficients satisfied}: This setting enforces (weak) heredity in the regression coefficients.  For all pairs $j < k$, define $\bar{\beta}_{jk} = \tfrac{1}{2}(|\beta_{\mathrm{main},j}| + |\beta_{\mathrm{main},k}|)$. Then each element of $\bbeta_{\mathrm{int}}$ is drawn as
\begin{equation}
\beta_{\mathrm{int},jk} \sim \normal{0, \sigma^2_{\mathrm{int}}} \cdot \mathbf{1}\!\left\{|\beta_{\mathrm{int},jk}| \leq \bar{\beta}_{jk}\right\},
\end{equation}
i.e., a truncated normal centred at zero with support $[-\bar{\beta}_{jk},\, \bar{\beta}_{jk}]$. Note that $\beta_{\mathrm{main},j} = \beta_{\mathrm{main},k} = 0$ implies $\bar{\beta}_{jk} = 0$ and hence $\beta_{\mathrm{int},jk} = 0$ such that strong heredity is also enforced among the inactive main effects. 

\textbf{(H2)} \textit{Heredity violated}: 
Elements of $\bbeta_{\mathrm{int}}$ are assigned as follows:
        $$
            \beta_{\mathrm{int},jk} = 
            \begin{cases}
                0 & \text{if } \beta_{\mathrm{main},j} = 0 \text{ and } \beta_{\mathrm{main},k} = 0, \\
                \beta_{\mathrm{main},j} & \text{if } \beta_{\mathrm{main},j} \neq 0 \text{ and } \beta_{\mathrm{main},k} = 0, \\
                \beta_{\mathrm{main},k} & \text{if } \beta_{\mathrm{main},j} = 0 \text{ and } \beta_{\mathrm{main},k} \neq 0, \\
                \tfrac{1}{2}(|\beta_{\mathrm{main},j}| + |\beta_{\mathrm{main},k}|) & \text{otherwise}.
            \end{cases}
        $$
        Hence, weak heredity is violated when one of the main coefficients is non-zero while the other is zero. 
The covariate matrix ${\bf X}$ is generated from a normal distribution with zero mean and a Toeplitz correlation structure with correlation coefficient $\rho = 0.5$. 

Following common practice, we centre and scale the main covariates, say $\tilde x_{ij}$ and $\tilde x_{ik}$, but do not re-standardise the interaction columns $z_{i,jk}=\tilde x_{ij}\tilde x_{ik}$ \citep{echambadi_mean-centering_2007, gelman_scaling_2008, schielzeth_simple_2010}. For all DGPs, we set the number of observations $n = 100$. In this case, none of the DGPs match any of the priors, i.e., no true model included in the comparison.

We compare $\Rtdprox$ with a range of existing shrinkage priors. In particular,  we consider $\rhs$ of \citet{piironen_sparsity_2017}, the unconstrained $\Rtd$ prior of \citet{zhang_bayesian_2022}, as well as $\nig$ prior of \citet{griffin_hierarchical_2017} which enforces weak heredity exactly. This is achieved via a Normal-Gamma-Gamma hierarchy where $\beta \sim \normal{0, \Psi}$, and weak heredity is strictly enforced by defining the interaction variances as $\Psi_{jk} = \eta_{jk} \frac{\Psi_j + \Psi_k}{2}$, with the parent variances following $\Psi_j \sim \text{GG}(\lambda_1, c, d_1)$ and the latent scales following $\eta_{jk} \sim \text{GG}(\lambda_2, c, d_2)$, where $\text{GG}$ stands for the double gamma distribution\footnote{Following Example 5.1 of \citet{griffin_hierarchical_2017}, $c = 2$ and the hyperpriors are specified as $\lambda_1 \sim \text{Exponential}(1)$ and $\lambda_2 = r \lambda_1$, where $r \sim \text{Beta}(2, 6)$. $d$ is assigned a heavy-tailed prior with density $p(d) \propto (1+d)^{-2}$.}

We report all recovery results \emph{relative} to the unconstrained $\Rtd$, which we regard as the natural benchmark model. We prefer reporting performance metrics in relative terms as it removes scale, signal-to-noise and dimensionality effects across DGPs. For a coefficient group $g \in \{\mathrm{main}, \mathrm{int}\}$ partitioned into active and inactive sets $\mathcal{A}_g$ and $\mathcal{A}_g^c$, define the group-wise root-mean-squared error of model~$m$ at MCMC simulation replicate $s$ as
\begin{equation}
      \rmse_{m,g,\mathcal{S}}^{(s)} \;=\; \sqrt{\frac{1}{\#\mathcal{S}}\sum_{j \in \mathcal{S}}\bigl(\hat{\beta}^{(s)}_{m,g,j} - \beta^{(s)}_{g,j}\bigr)^{2}}, \qquad \mathcal{S} \in \{\mathcal{A}_g,\, \mathcal{A}_g^c\},
\end{equation}
where $\hat{\beta}^{(s)}_{m,g,j}$ denotes the posterior mean for model~$m$ and $\#\mathcal S$ is the number of elements in $\mathcal A$. The \emph{relative RMSE} of model~$m$ against the $\Rtd$ baseline is 
\begin{equation}
     \relrmse_{m,g,\mathcal{S}} \;=\; \frac{1}{S_{\mathrm{sim}}} \sum_{s=1}^{S_{\mathrm{sim}}} \frac{\rmse_{m,g,\mathcal{S}}^{(s)}}{\rmse_{\Rtd,g,\mathcal{S}}^{(s)}},
\end{equation}
where $S_{\mathrm{sim}} = 100$ is the number of replications. Values below one indicate an improvement over the $\Rtd$ baseline. Analogously, we report out-of-sample predictive performance via the $\elpd$-difference on held-out data (of size $n$),
\begin{equation}
    \Delta\elpd_{m}^{(s)} \;=\; \elpd_{m}^{(s)} - \elpd_{\Rtd}^{(s)},
\end{equation}
with positive values favouring model~$m$ over the baseline \citep{vehtari_pareto_2022}. \citet{mclatchie_robust_2023} and \citet{sivula2025uncertainty} provide theoretical support for treating $|\Delta\elpd| > 4$ as the threshold above which two models have meaningful difference in predictive performance (assuming it is scaled by the number of in-sample observations). 
Replication code for the simulations can be found at \url{https://github.com/davkoh/prox-r2}.
\subsubsection{Parameter recovery}\label{subsubsec:recovery}

We show the main recovery results in Figure~\ref{fig:rmse_recovery_interaction} and briefly summarise the findings before providing a more detailed discussion.

\begin{figure}[t]
    \centering
    \resizebox{1.05\textwidth}{!}{\input{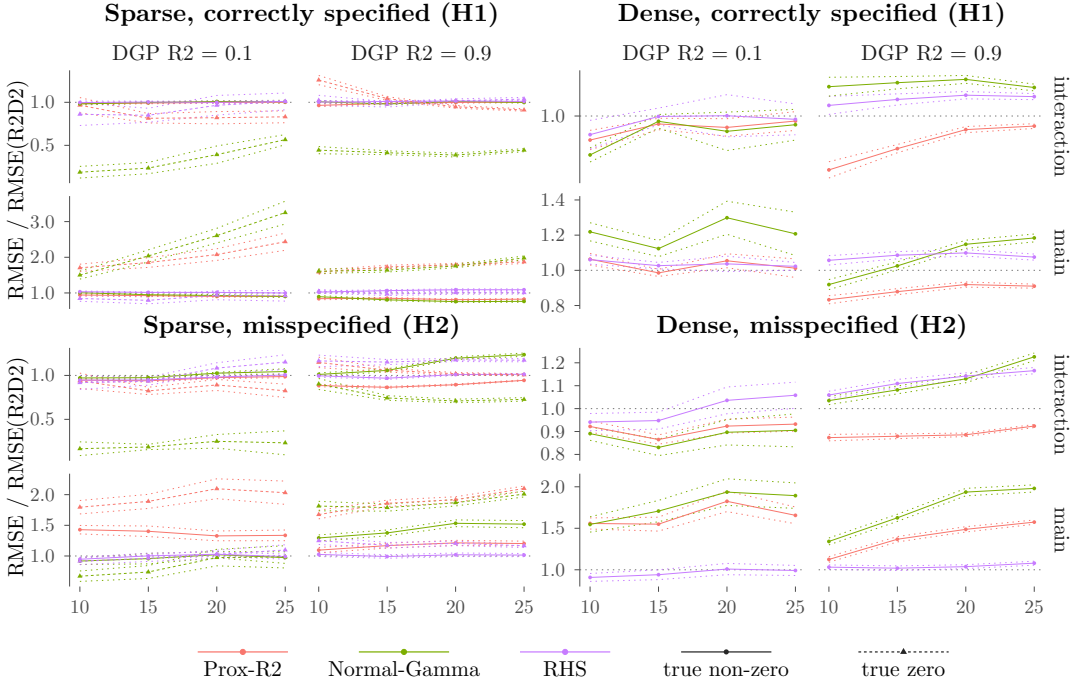}
    }
    \caption{Recovery of coefficients, measured by the  $\relrmse$, relative to the unconstrained $\Rtd$ baseline, for increasing main effects $p$. Panels correspond to the four DGPs and sub-panels to the low and high signal regimes (DGP-$R^2 = 0.1$ and $0.9$). Curves are averaged over the $S_{\mathrm{sim}}=100$ replications. The dashed line marks $\relrmse = 1$, and the dotted lines are the standard-error estimate. Tables~\ref{tab:mcmc_diagnostics_beta_p10}-\ref{tab:mcmc_diagnostics_beta_p25} in the Supplementary Materials report the associated MCMC diagnostics.}
    \label{fig:rmse_recovery_interaction}
\end{figure}

\paragraph{Summary.}
For interaction coefficients, $\Rtdprox$ performs as well as or better than the competing priors for true non-zero coefficients, most clearly in the Dense (H2) setting. This pattern is present across signal-to-noise ratios and under both heredity designs. In contrast, $\nig$ prior performs well for true zero coefficients, benefiting from its strong shrinkage towards the origin. For main-effect coefficients, $\Rtdprox$ improves on $\Rtd$ for non-zero coefficients when the DGP is dense and there is a high signal-to-noise ratio. In the remaining settings, priors that do not model heredity are comparable or perform slightly better for main effects, particularly for true zeros. This results from the main effect variance parameters being lower bounded by the heredity constraints, thus leading to underestimation of the posterior variance.

\paragraph{Sparse, correctly specified (H1).}
For interaction effects, all priors are essentially indistinguishable from the $\Rtd$ on the truly non-zero coefficients ($\relrmse \approx 1$ across $p$ and DGP-$R^2$ levels). On true zeros, $\Rtdprox$ and $\nig$ deliver markedly stronger shrinkage, with $\nig$ achieving the smallest error ($\relrmse \approx 0.2\text{--}0.5$) and $\Rtdprox$ also improving over $\Rtd$ as $p$ grows ($\relrmse$ falling from $\approx 0.95$ at $p=10$ to $\approx 0.8$ at $p=25$). As $p$ increases, the active main-effect set remains fixed while the number of interactions with at least one inactive parent grows quadratically, so the heredity constraint applies to an increasing share of the interaction block and shrinks the corresponding $\phi_{jk}$ toward zero.
For main effects, non-zero coefficients are recovered comparably across all models. On true-zero main coefficients, however, both  $\Rtdprox$ and $\nig$ models perform worse than $\Rtd$ ($\relrmse$ growing from $\approx 1.5$ at $p=10$ to $> 2.5$ at $p=25$) as imposing the heredity constraint at the variance level bounds the main-effect decomposition weights away from zero (see Section \ref{sec:projections}), leading to overestimation of the variance of the true zero coefficients.

\paragraph{Dense, correctly specified (H1).}
At low signal (DGP-$R^2=0.1$), the value of explicitly modelling weak heredity is muted since all constrained priors track the $\Rtd$ closely on both effect groups. Hence, recovery is sufficiently dominated by noise so that all priors behave similarly. With high signal for interaction effects (DGP-$R^2=0.9$), $\Rtdprox$ model attains the lowest $\relrmse$ (below one at $\approx 0.95$), whereas $\nig$ over-shrinks the dense signal and is the worst performer ($\relrmse \approx 1.05$). For main effects, $\Rtdprox$ again clearly performs best, achieving $\relrmse \approx 0.85\text{--}0.90$ while $\nig$ deteriorates to $\approx 1.18$ and the remaining competitors remain near $\Rtd$ performance. This likely results from combining structured variance allocation with direct control of total prior signal.

\paragraph{Sparse, misspecified (H2).}
When the heredity constraint is violated by the DGP, $\Rtdprox$ model nevertheless remains the most robust choice on non-zero interaction coefficients. For example, with large signal it attains $\relrmse \approx 0.85$, outperforming all competitors. This is consistent with the constraint acting as a \emph{soft} penalty toward the heredity region: informative enough to regularise beyond the unconstrained $\Rtd$, yet permissive enough to allow truly non-zero off-heredity interactions to escape shrinkage. For true-zero interactions, $\nig$ remains the strongest performer ($\relrmse \approx 0.2$), again reflecting its aggressive shrinkage behaviour. Similar to (H1), for main effects, non-zero coefficients are recovered uniformly well across priors, while $\Rtdprox$ cannot recover the true zeros as well. In contrast to (H1), however, the error for $\nig$ on true-zero main coefficients is larger in the low-signal DGP.

\paragraph{Dense, misspecified (H2).}
For interaction effects with low signal, both interaction priors improve modestly on the $\Rtd$ ($\relrmse \approx 0.85-0.90$). With large signal, $\Rtdprox$ again performs best ($\relrmse \approx 0.87$), possibly because the (incorrect) pull toward the heredity region still provides useful regularization even though the DGP violates the constraint. For main effects, the cost of imposing an incorrectly specified constraint is more apparent as both constrained priors deteriorate substantially relative to $\Rtd$, though $\Rtdprox$ does so less than $\nig$.$\rhs$ behaved similarly to the $\Rtd$~for sparse DGPs, but performed noticeably worse on average for the dense DGPs -- a finding consistent with \citet{kohns_arr2_2025}.

\subsubsection{Predictive performance}
\begin{figure}[t]
    \centering
    \resizebox{1.05\textwidth}{!}{
        \input{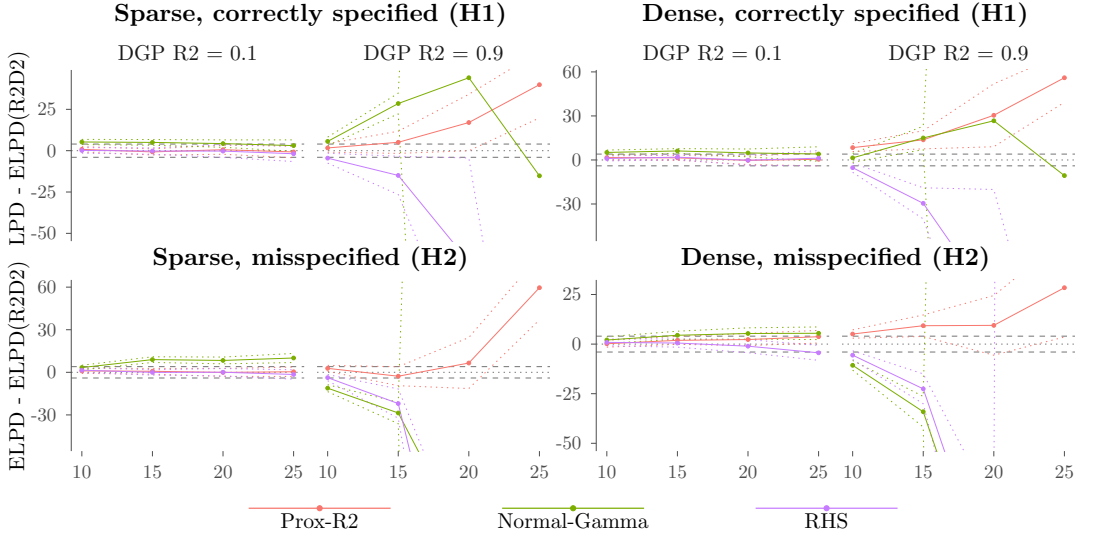}
    }
    \caption{Predictive performance, measured by the out-of-sample $\Delta\elpd$ relative to the unconstrained $\Rtd$ baseline, for increasing main effects, $p$. Panels correspond to the four DGPs and sub-panels to the low and high signal regimes (DGP-$R^2 = 0.1$ and $0.9$). Curves are averaged over the $S_{\mathrm{sim}}=100$ replicates. The dotted line marks $\Delta\elpd = 0$ and the dashed lines the $|\Delta\elpd| = 4$ threshold of \citet{sivula2025uncertainty}. The vertical axes are clipped because at high signal and large $p$ the $\nig$ and $\rhs$ curves fall far below the plotted range. Tables~\ref{tab:mcmc_diagnostics_beta_p10}-\ref{tab:mcmc_diagnostics_beta_p25} in the Supplementary Materials report the associated MCMC diagnostics.}
    \label{fig:elpd_interaction_r2_01_09}
\end{figure}

We show predictive performance, measured by the held-out $\Delta\elpd$ relative to the unconstrained $\Rtd$ baseline, in Figure~\ref{fig:elpd_interaction_r2_01_09}. The curves averaged over all $S_{\mathrm{sim}}=100$ simulation replications and mixing diagnostics are reported in Tables~\ref{tab:mcmc_diagnostics_beta_p10}-\ref{tab:mcmc_diagnostics_beta_p25} in the Supplementary Materials. Following \citet{sivula2025uncertainty}, we treat $|\Delta\elpd| > 4$ (dashed lines) as the threshold for a practically meaningful difference. 

At low signal-to-noise (DGP-$R^2 = 0.1$), all models are predictively indistinguishable from the $\Rtd$ across the four DGPs and all $p$ ($|\Delta\elpd| \approx 0$). The only exception is a marginal advantage for $\nig$ on the sparse, misspecified DGP, consistent with its strong recovery of the true-zero interaction coefficients.

At high signal-to-noise (DGP-$R^2 = 0.9$), the priors separate sharply. Under both correct specification (H1) and misspecification (H2), $\Rtdprox$ improves substantially over $\Rtd$ and scales favourably with $p$. $\nig$ is competitive at low dimension but falls far below the baseline by $p=25$, alongside $\rhs$. As shown in Section~\ref{subsubsec:recovery}, the robustness of $\Rtdprox$ appears to be driven mainly by its more accurate recovery of the non-zero interaction coefficients rather than the main effects. 

The poor performance of $\nig$ and $\rhs$ for large $p$ is predominantly stemming from computational rather than modelling issues. Both priors are comparable to $\Rtd$ at small $p$, where the chains of all models mix well (Table~\ref{tab:mcmc_diagnostics_beta_p10}, Supplementary Materials), while their deterioration tracks the loss of mixing in higher dimensions (Tables \ref{tab:mcmc_diagnostics_beta_p15}-\ref{tab:mcmc_diagnostics_beta_p25} in the Supplementary Materials): at high signal, $\nig$ and $\rhs$ have $\hat{R}$ up to $1.76$ and $1.49$, respectively, and both have effective sample sizes in the single digits. In summary, $\Rtdprox$ is competitive-to-best on predictive accuracy across all four DGPs, most robust under heredity misspecification, and computationally more reliable in regimes where the competing priors break down.

\subsubsection{Posterior probability of constraint realization}
We conduct a more detailed investigation of the posterior behaviour of the proposed and competing models. First, a key aspect of $\Rtdprox$ is the proximal MCMC step that projects the variance allocation parameters~$\bphi$ into the constrained set defined by the heredity condition, i.e., $\eta_{jk} \leq \tfrac{1}{2}(\eta_j + \eta_k)$. Ideally, the proximal step regularises the posterior so that it concentrates in the desired constrained region. Since the constraint is imposed on the latent variance decomposition parameters $\boldsymbol{\eta}$ rather than on the regression coefficients~$\bbeta$ directly, however, it is not guaranteed that the constraint will also be realised at the coefficient level, i.e., $|\beta_{jk}|\leq\frac12(|\beta_j|+|\beta_k|)$.

To investigate this, we report heatmaps of constraint realization probabilities for  $\Rtdprox$, $\Rtd$ and $\nig$ at $p=20$ under the sparse, correctly specified DGP~(H1) with DGP-$R^2 = 0.5$ in Figure~\ref{fig:feasibility_sparse_correct}. The analogous heatmaps for the remaining DGP permutations (sparse misspecified, dense correctly specified, and dense misspecified) are reported in the Supplemental Materials, Figures~\ref{fig:feasibility_sparse_misp}--\ref{fig:feasibility_dense_misp}, and exhibit the same qualitative pattern. Comparing these models reveals that the proximal step with $\lambda = 1$ achieves near-perfect constraint realization at the~$\bphi$ level. The regression coefficients, however, satisfy the constraints much less frequently and exhibit minimal differences with the other priors. This makes sense given that our prior satisfies geometric constraints on the allocation scale.

$\nig$ prior, in contrast, applies the heredity constraint strictly on the variances, yet even this results in only a marginal improvement in constraint realization at the coefficient level relative to $\Rtdprox$. This suggests that the soft constraint of our proximal approach achieves a comparable degree of constraint realization while retaining the computational advantages and high degree of interpretability of the R2 prior framework.

\begin{figure}[t]
    \centering
    \includegraphics[width=0.9\linewidth]{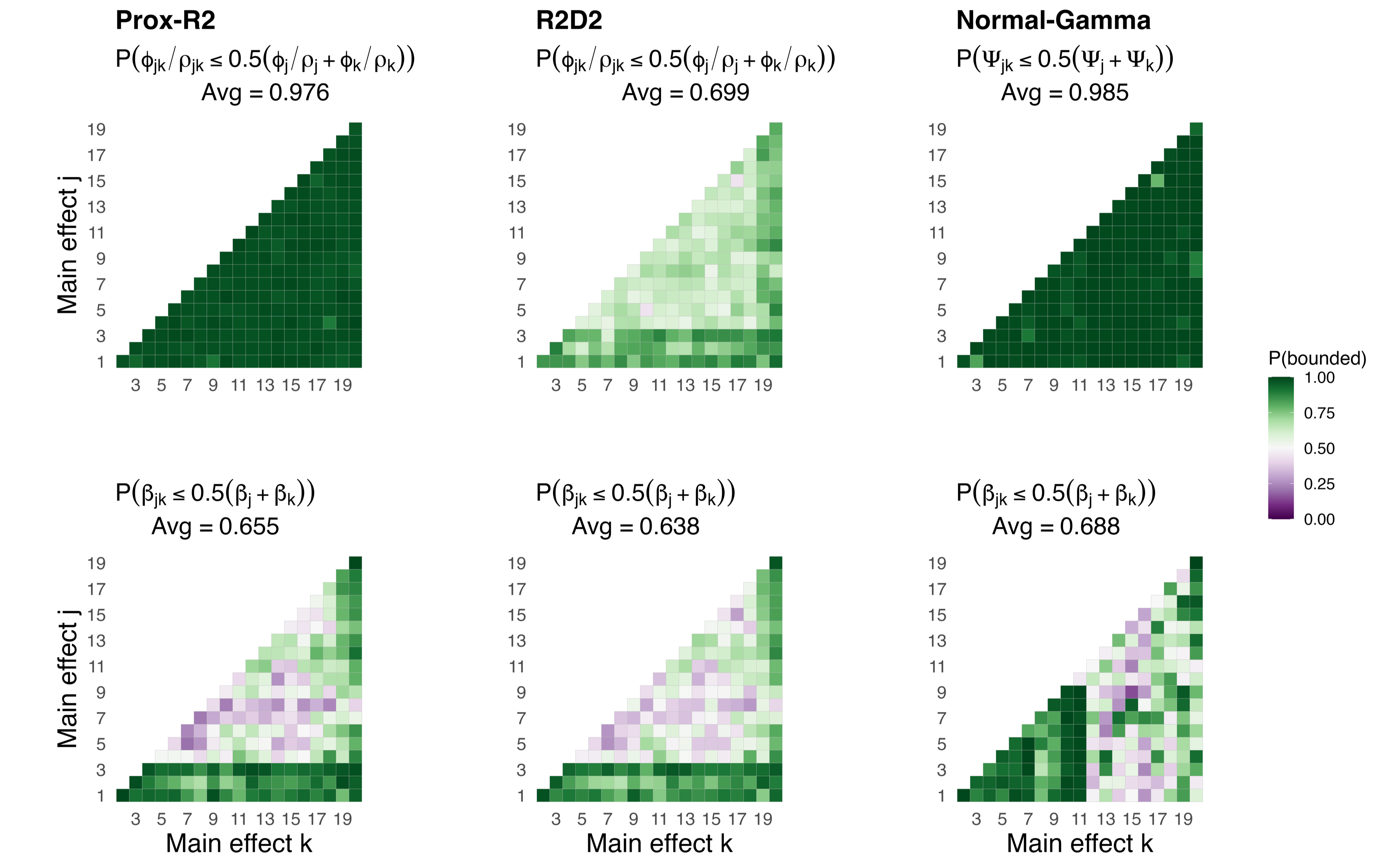}
    \caption{Constraint feasibility heatmaps for $p=20$, sparse, correctly specified DGP~(H1) with $R^2 = 0.5$. Panels compare  $\Rtdprox$ model against the unconstrained $\Rtd$ and $\nig$ models.}
    \label{fig:feasibility_sparse_correct}
\end{figure}

Additionally, in the Supplemental Materials, we report MCMC diagnostics in Tables~\ref{tab:mcmc_diagnostics_beta_p10}-\ref{tab:mcmc_diagnostics_beta_p25}. Overall, we find that all models struggle in high-dimensions and large signal. Since, we adjust the true $\sigma^2$ in order to achieve the desired DGP-$R^2$, the curvature of the likelihood sharpens, making exploration with gradient-based samplers more difficult. Nevertheless, we find that the geometry of the posterior induced by soft rather than hard penalization allows for improved sampling.

\subsection{Real-data analysis}

We investigate the performance of the proposed method on several real-world datasets. Inspired by \citet{griffin_hierarchical_2017}, we apply $\Rtdprox$ and competing priors to a generalised additive model, where we employ the linear-over-nonlinear allocation constraint of Equation~\eqref{eq:gam-average-hierarchy}.

\begin{figure}[t]
    \centering
    \includegraphics[width=0.9\linewidth]{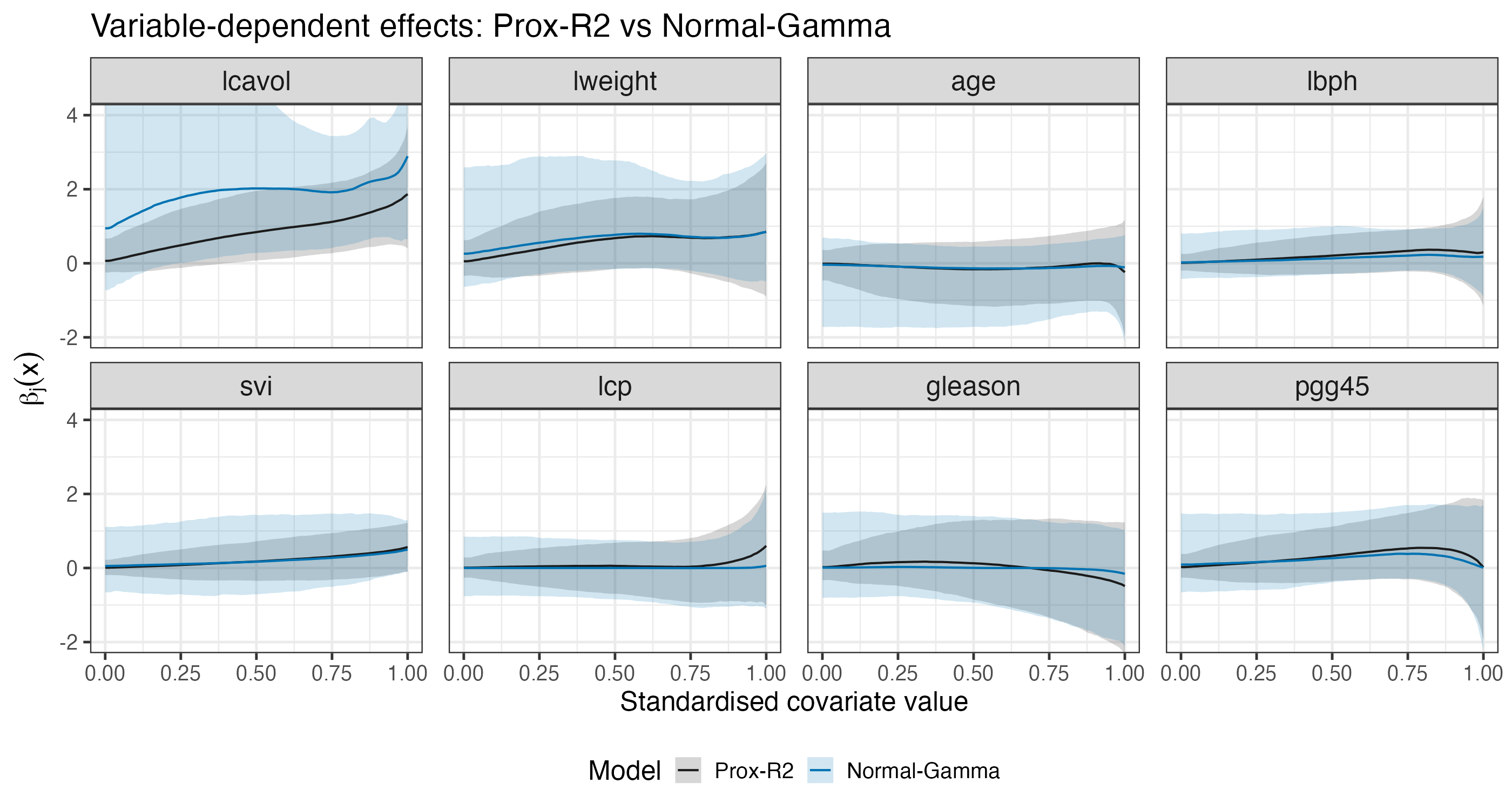}
    \caption{Non-linear effects estimates for the prostate cancer data set: $\Rtdprox$ vs.\ $\nig$}
    \label{fig:non_linear_effects}
\end{figure}

We first use the prostate cancer data set, prepared as in Example~2 of \citet{griffin_hierarchical_2017}, with the log prostate-specific antigen (\texttt{lpsa}) as the response and a spline basis expansion for each of the eight covariates. Here, we apply the group-level allocation constraint of Equation~\eqref{eq:gam-average-hierarchy}, which bounds the average latent allocation of the non-linear basis terms of covariate~$j$ by the allocation of its linear term, $\tfrac{1}{L_j}\sum_{\ell=1}^{L_j}\eta_{j\ell}\leq\eta_j$, and may be interpreted as a preference for linear over non-linear covariate effects. Since the non-linear effects are modelled as spline basis expansions, such models are very high-dimensional.

Figure~\ref{fig:non_linear_effects} shows the posteriors of the non-linear functions over the covariate domains. $\Rtdprox$ yields narrower pointwise posterior intervals than $\nig$, and a markedly different non-linear effect for the log cancer volume covariate (\texttt{lcavol}). As in the simulation section, this is partly explained by the substantially better sampling efficiency of the $\Rtdprox$ relative to $\nig$, reported in Table~\ref{tab:gam_efficiency}. $\Rtdprox$ mixes well (minimum bulk-ESS $\approx 3700$, $\hat{R}\leq 1.01$), whereas $\nig$ fails to converge ($\hat{R}=1.31$, minimum bulk-ESS $\approx 200$). Hence, the $\nig$ effect estimates are unreliable, and its convergence failure likely explains part of the discrepancy in Figure~\ref{fig:non_linear_effects}.

\begin{table}
\caption{Sampling efficiency on the prostate cancer GAM ($p=8$ covariates, each with $L_j = 60$ spline basis functions, $4$ chains of $1000$ warmup $+$ $1000$ sampling iterations). Min bulk-ESS and max $\hat{R}$ are taken across the regression coefficients $\boldsymbol{\beta}$.\label{tab:gam_efficiency}}
\begin{tabular}{lrrr}
\toprule
Model &  min bulk-ESS & max $\hat{R}$ & bulk-ESS/s \\
\midrule
$\Rtdprox$ &   3744 & 1.005 & 161.7 \\
$\nig$    &   207 & 1.307 &   0.2 \\
$\Rtd$    &   4608 & 1.007 & 112.5 \\
$\rhs$    &   3382 & 1.004 &  70.3 \\
\bottomrule
\end{tabular}
\end{table}

\subsubsection{Additional real-data results}

We further assess out-of-sample predictive performance on the Boston Housing \citep{harrison_hedonic_1978}, Ozone \citep{breiman_estimating_1985}, and Blood Glucose \citep{efron_least_2004} data sets. In each case, we include all $p$ main effects, all $p_{\mathrm{int}} = \binom{p}{2}$ interactions, and hold out $75\%$ of the observations, so that the number of interaction terms is large relative to the training sample. $\Rtdprox$ provides competitive performance on all three data sets, attaining the best point predictions on Boston Housing and the largest proportion of explained variance on the held-out data. Full results are reported in Appendix~\ref{app:predperf}.

\section{Discussion}\label{sec:Discussion}



Modern regression problems frequently contain structured collections of terms, such as main effects and interaction, linear and nonlinear basis expansions, or treatment groups, over which domain knowledge can naturally guide regularisation in high-dimensional settings. In this work, we introduced $\Rtdprox$, a prior distribution designed for the intermediate setting where we favour scientifically plausible relationships among variance allocations without treating them as inviolable rules. Our work connects two lines of prior research that have largely addressed different parts of the same modelling problem: (1) priors that take structure into account through dependence among the coefficient scales such as heredity, grouping, or ordered importance \citep{zellner_assessing_1986,chipman_bayesian_1997, griffin_hierarchical_2017, griffin_structured_2024, aguilar2025dependencyawareshrinkagepriorshigh}; and (2) R2 priors, which begin with an interpretable prior on overall model fit and then partition the implied signal across model terms \citep{zhang_bayesian_2022}. With our prior, we bring these perspectives together by expressing structural relationships as linear constraints on a latent space of variance allocations and relaxing them via proximal methods \citep{pereyra_proximal_2016}. In doing so, our prior uses these relationships to guide relative signal allocation across terms while preserving interpretable, model-level prior elicitation \citep{mikkola_prior_2021, bockting_simulation-based_2024, manderson_translating_2026}.

Broadly speaking, our work connects to the growing literature that leverages optimisation paradigms within Bayesian inference \citep{duan_bayesian_constraint_2020, presman2023distance, xu_bayesian_2023}. Although we use the same computational language, our method is differentiated by evaluating preferences over variance allocations directly rather than by defining a set of admissible coefficient values. The former perspective is particularly useful against structural misspecification, since we do not remove any region from the support of the coefficients. Optimisation has appeared in $R^2$ priors as a tool for selecting hyperparameters before model fitting \citep{yanchenko_r2d2_2025,yanchenko_pseudo-r2d2_2025} while we involve it here throughout posterior computation, yielding a continuously differentiable target that scales reliably with standard gradient-based sampler.

Our theoretical results complement the established theory of global-local shrinkage priors to coupled, proximally constrained variance allocations. Work on Horseshoe and R2D2 priors has shown how behaviour near zero, heavy tails, and posterior concentration can support recovery in sparse problems \citep{carvalho_horseshoe_2010, pas_horseshoe_2014, zhang_bayesian_2022}. These analyses rely on independent or exchangeable local behaviour, whereas our construction related the coefficient scales by design. Our results show that these modifications preserves the defining shrinkage behaviour of the prior. 
We also show that posterior recovery remains possible when the favoured structural relationship holds only approximately.

Our simulations provide evidence about when and to what degree favouring the structural relationship helps. In particular, under moderate to strong signal-to-noise ratio, we show efficient recovery of non-zero interactions, even when the true data-generating mechanism follows constraints only approximately.  At low signal, the competing priors are more difficult to distinguish. In general, when considering constraints, practitioners should formulate them as beliefs about relative regularisation taking into account how the prior distribution shifts mass from the standard exchangeable case to the one that reflects the preferred structure. To do this, one could resort to use of prior-predictive summaries as well as evaluating term-level behaviour. Likewise, the smoothing strength should be considered a modelling decision, capturing the degree to which we are willing to accept deviations from the constraints.

Our work faces several limitations which suggest promising directions for future research. First, our implementation is most convenient for closed convex sets with efficient projections at hand.  Extending $\Rtdprox$ to nonlinear or nonconvex relational constraints, or empirically evaluating it across the broader families of multilevel, generalized linear, and time-series models where R2 priors have been developed \citep{aguilar_intuitive_2023, yanchenko_r2d2_2025, kohns_arr2_2025}, will require new numerical strategies and validation. We have so far only worked with default hyperparameters for the logistic-normal allocation parameters, without considering the interaction that these could have the constraints. Jointly learning these hyperparameters, together with the smoothing strength, may provide a more flexible elicitation approach. Another promising direction is observational causal inference, where our framework could encode structural preferences over variance allocation among prognostic effects and treatment-effect modifiers, with particular attention to how prior specification and regularisation affect causal effect estimation \citep{hahn2019bayesianregressiontreemodels,linero_nonparametric_2024,orihara2026bayesiandoublyrobustcausal,riha2026select,aguilar2026dirichlettreeallocationpriorsbayesian}. Finally, a broader extension would replace R2 decomposition with other elicitable model-level or predictive quantities that could be allocated across terms, while still accounting for constraints. Simulation-based elicitation could provide one route to constructing such priors \citep{bockting_simulation-based_2024,manderson_translating_2026}.


\section{Acknowledgments}

Javier Enrique Aguilar and Paul-Christian Bürkner acknowledge  funding by the Deutsche Forschungsgemeinschaft (DFG, German Research Foundation) Project 500663361 as well as the computing time provided on the Linux HPC cluster at Technical University Dortmund (LiDO3), partially funded in the course of the Large-Scale Equipment Initiative by DFG Project 271512359. Paul-Christian Bürkner further acknowledges support of the DFG via the Collaborative Research Center 391 (Spatio-Temporal Statistics for the Transition of Energy and Transport) – 520388526. David Kohns acknowledges the computational resources provided by the Aalto Science-IT project, and the support of the Research Council of Finland Flagship programme: Finnish Center for Artificial Intelligence, Research Council of Finland project (340721), and the Finnish Foundation for Technology Promotion. Eric Yanchenko was supported by JSPS KAKENHI grant 24K22613. Shonosuke Sugasawa was supported by JSPS KAKENHI grant 24K12757 and 25H00699. Kaoru Irie was supported by JSPS KAKENHI grant 22K13374.

\appendix

\newpage
\clearpage

\section{Supplemental Materials}
\label{sec:supplement}

\subsection{Sampling the smoothed posterior}
\label{sec:sampling}

The smoothed posterior $\pi^\lambda(\btheta\mid\byobs)$ in \eqref{eq:smoothed-posterior} is proper for every $\lambda>0$ and can therefore be targeted by standard MCMC algorithms. The role of the Moreau--Yosida envelope depends on the sampler. In all cases, however, the Markov chain targets a single density, namely $\pi^\lambda$. The projection $\mathcal P_{\mathcal E}$ is evaluated only to compute the envelope or its gradient. Posterior draws are not projected onto $\mathcal E$.

\paragraph{Gradient-free samplers.}
Consider a Metropolis-Hastings update of $\btau$, conditional on the remaining parameters $\bnu$. A candidate $\btau^\star$ is drawn from a proposal distribution $q(\btau^\star\mid\btau)$ and accepted with probability $ \alpha(\btau,\btau^\star)$$ =
\min \left\{ 1, \frac{ \pi^\lambda(\btau^\star,\bnu\mid\byobs) q(\btau\mid\btau^\star)} { \pi^\lambda(\btau,\bnu\mid\byobs) q(\btau^\star\mid\btau)} \right\}.$ The envelope enters the acceptance probability only through the two scalar values $\delta^\lambda_{\mathcal E}(\btau^\star)$ and $\delta^\lambda_{\mathcal E}(\btau)$, each obtained from a projection evaluation via Equation~\eqref{eq:my-envelope-indicator}. The candidate point is generated without reference to $\mathcal E$, and derivatives of $\delta^\lambda_{\mathcal E}$ are not required. A candidate outside $\mathcal E$ is accepted or rejected according to its probability under the smoothed target $\pi^\lambda$.

\paragraph{Gradient-based samplers.}
For HMC \citep{neal_mcmc_2011}, the smoothed target determines the potential energy $ U(\btheta) = -\log\pi^\lambda(\btheta\mid\byobs).$ The Moreau--Yosida envelope contributes to this potential through the term $\delta^\lambda_{\mathcal E}(\btau)$, whose gradient is given by $\nabla\delta^\lambda_{\mathcal E}(\btau) = \frac{1}{\lambda} \left(\btau-\mathcal P_{\mathcal E}(\btau)\right).$ This gradient is zero when $\btau\in\mathcal E$. Outside $\mathcal E$, it points away from the feasible set, from the projection $\mathcal P_{\mathcal E}(\btau)$ toward $\btau$. The corresponding Hamiltonian force is the negative gradient $-\nabla\delta^\lambda_{\mathcal E}(\btau)$, which points back toward $\mathcal E$. Thus, in HMC or NUTS \citep{hoffman_no-u-turn_2014}, the constraint affects the leapfrog trajectory through the additional force represented by $-\nabla\delta^\lambda_{\mathcal E}(\btau)$. Each leapfrog step requires evaluating $\mathcal P_{\mathcal E}(\btau)$ in order to compute this.

Figure~\ref{fig:proximal-disk} shows the resulting dynamics for a circular constraint, where a trajectory that leaves $\mathcal E$ is driven back" toward it by the penalty each step and may terminate outside $\mathcal E$ because the constraint is enforced softly. As the shrinkage models of Section~\ref{sec:method} are fitted with gradient-based samplers, the differentiability provided by the Moreau--Yosida envelope is what renders the constraint compatible with these algorithms.

\subsection{ Algorithms }

\begin{algorithm}[H]
\caption{Dykstra's algorithm for projection onto the intersection of half-spaces
    $\mathcal{E} = \bigcap_{i=1}^m \mathcal{E}_i$, where 
    $\mathcal{E}_i = \{\bs \in \mathbb{R}^d \mid \bG_i^\prime\bs \le 0\}$}
\label{alg:dykstra}
\KwIn{$\bt \in \mathbb{R}^d$, constraint matrix $\bG \in \mathbb{R}^{d \times m}$, 
    tolerance $\varepsilon > 0$}
\KwOut{$\mathcal{P}_{\mathcal{E}}(\bt)$}
$\bs \leftarrow \bt$ \tcp*{initialise iterate}
$\bz_i \leftarrow \mathbf{0}$ for $i = 1, \ldots, m$ \tcp*{initialise correction vectors}
$\delta \leftarrow \infty$\;
\BlankLine
\While{$\delta \geq \varepsilon$}{
    $\bs^{\mathrm{prev}} \leftarrow \bs$\;
    \For{$i \leftarrow 1$ \KwTo $m$}{
        $\tilde{\bs} \leftarrow \bs + \bz_i$ 
            \tcp*{add back correction before projecting}
        $\bs \leftarrow \tilde{\bs} - \dfrac{\max(0,\, \bG_i^\prime\tilde{\bs})}{\|\bG_i\|^2}\,\bG_i$
            \tcp*{project onto $\mathcal{E}_i$}
        $\bz_i \leftarrow \tilde{\bs} - \bs$ 
            \tcp*{record displacement for next cycle}
    }
    $\delta \leftarrow \|\bs - \bs^{\mathrm{prev}}\|$\;
}
\Return{$\bs$}\;
\end{algorithm}

\subsection{Projection versus truncation}

A natural question is whether projecting a random variable onto a convex set $\mathcal{E}$ is equivalent to truncating it to $\mathcal{E}$. These are two ways to obtain a random variable supported on $\mathcal{E}$, but they lead to fundamentally different distributions.

Given a random variable $\bS$, truncation to $\mathcal{E}$ conditions on the event $\bS \in \mathcal{E}$, yielding a distribution with density proportional to
\begin{equation}
p_{\mathcal{E}}(\bs) \propto \mathbf{1}_{\mathcal{E}}(\bs)\, p_{\bS}(\bs).
\end{equation}

Projection, by contrast, maps every realization of $\bS$ to its nearest point in $\mathcal{E}$, producing the pushforward random variable $\bT = \mathcal{P}_{\mathcal{E}}(\bS)$. The two operations coincide for parameter values that lie in $\mathcal{E}$, but differ fundamentally for those outside it. Truncation discards the parameter values entirely, whereas projection maps them onto the boundary $\partial\mathcal{E}$. Because the projection is many-to-one on $\partial\mathcal{E}$, the pushforward distribution can accumulate positive probability mass on the boundary. On the other hand, a truncated distribution cannot do this, since it simply renormalises the original density on $\mathcal{E}$ and therefore assigns zero probability to sets of measure zero.

We can make this concrete with a one-dimensional example. Let $\mathcal{E} = [0,\infty)$ and $S \sim \normal{0,1}$. The projection is $\mathcal{P}_{\mathcal{E}}(s) = \max(0,s)$, so every $s < 0$ is mapped to zero. The pushforward $T = \mathcal{P}_{\mathcal{E}}(S)$ therefore has a point mass of $1/2$ at zero and a half-normal density on $(0,\infty)$. The truncated normal $S \mid S \ge 0$, by contrast, has density $p(s) = 2\phi(s)$ for $s \ge 0$ and no point mass anywhere, showing that the two distributions are clearly distinct.


\newpage
\section{Further simulation results}

\begin{figure}[t]
    \includegraphics[width=\textwidth]{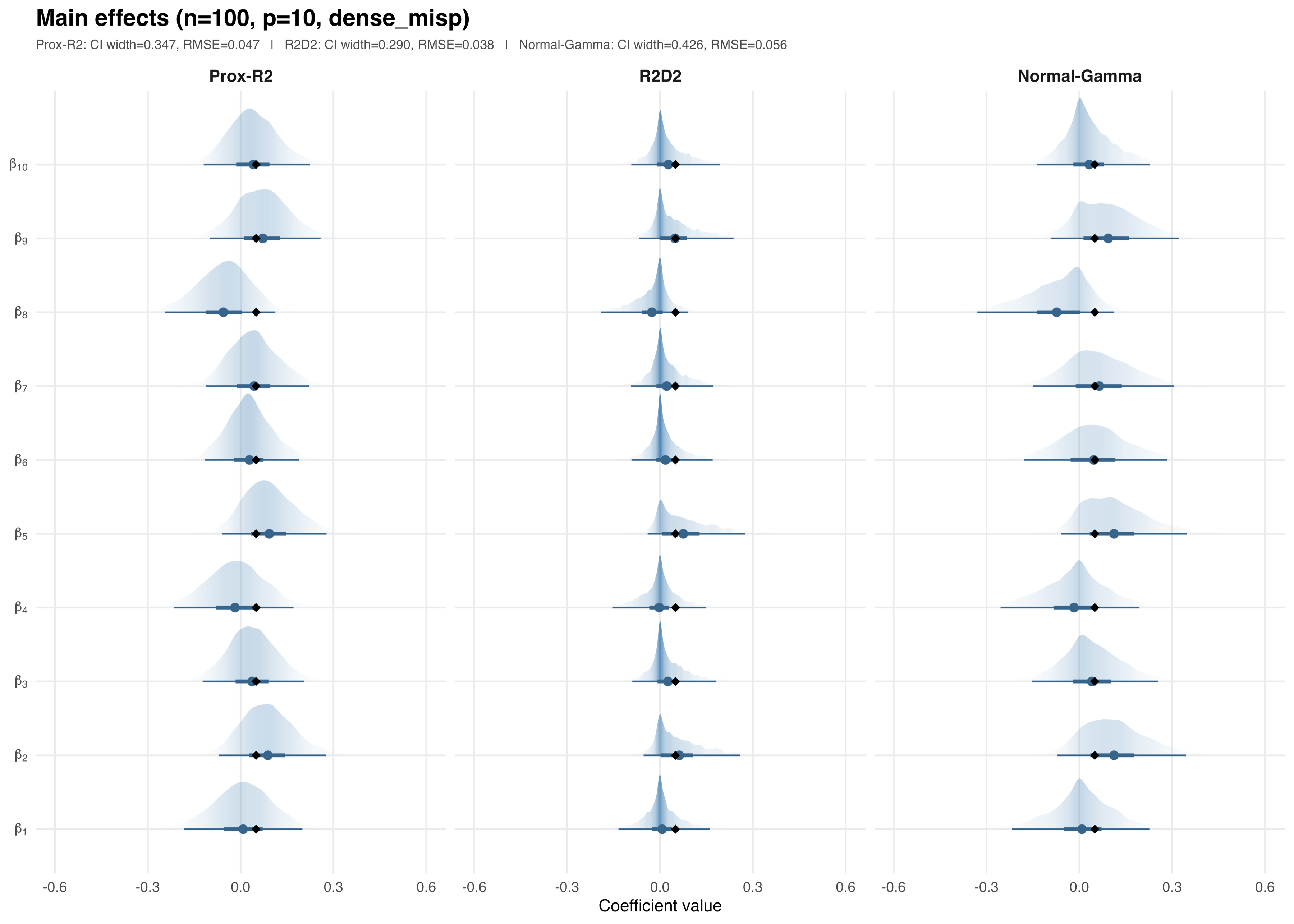}
    \caption{Main effect coefficient posteriors for misspecified, dense DGP (H2).}
    \label{fig:posteriors_dense_misp_main}
\end{figure}

\begin{figure}[t]
    \includegraphics[width=\textwidth]{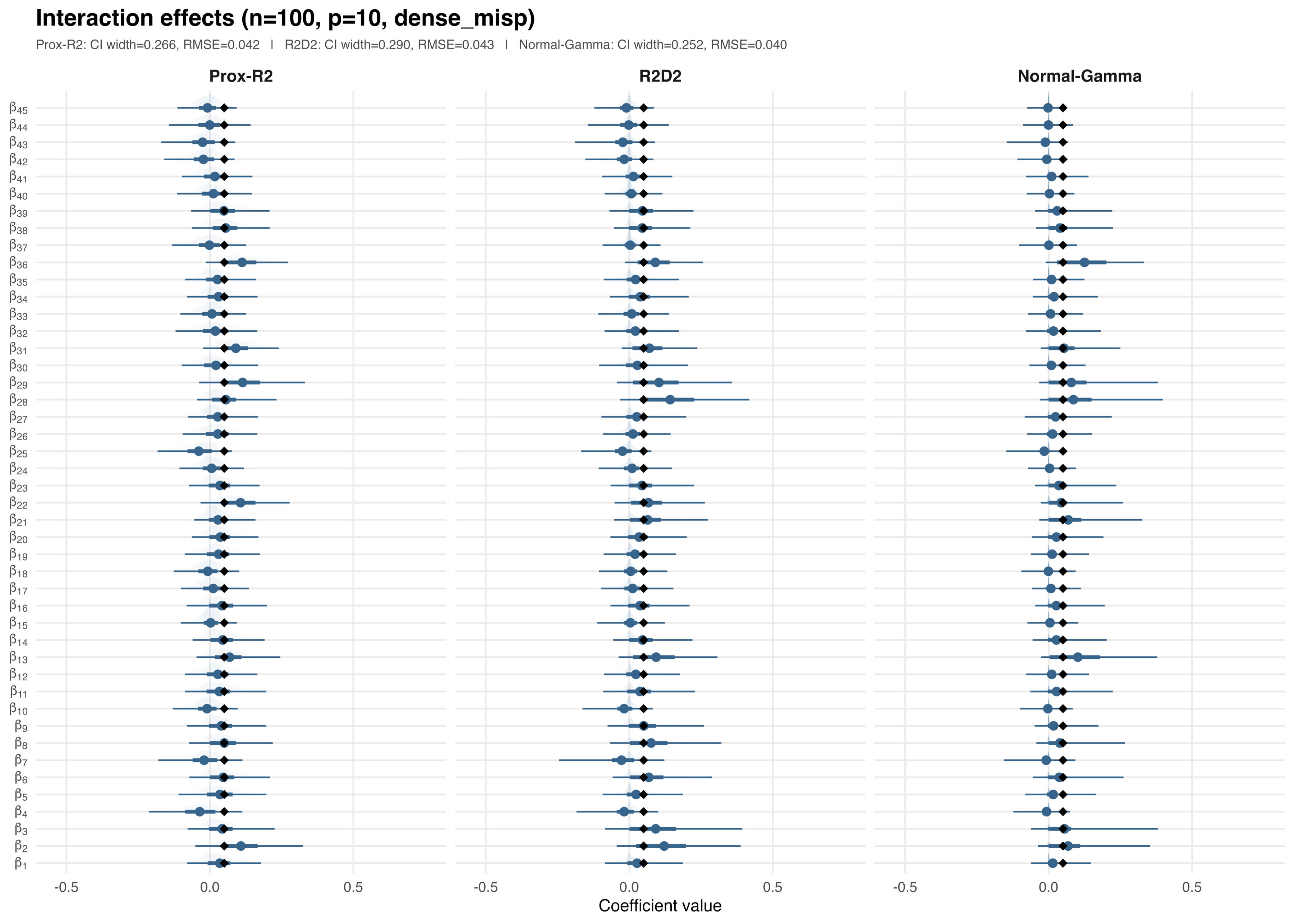}
    \caption{Interaction effect coefficient posteriors for misspecified, dense DGP (H2).}
    \label{fig:posteriors_dense_misp_int}
\end{figure}

\newpage

\subsection{MCMC diagnostics}
\label{sec:MCMC_diagnostics}

\begin{table}
\caption{MCMC diagnostics for $\boldsymbol{\beta}$ ($p = 10$). For each DGP condition and model: mean (over simulation repetitions) of the per-repetition maximum $\hat{R}$ and minimum bulk ESS across $\beta$ coefficients. Bold entries indicate $\bar{\hat{R}} > 1.01$.\label{tab:mcmc_diagnostics_beta_p10}}
\resizebox{\textwidth}{!}{%
\begin{tabular}{@{}llrrrrrrrr@{}}
\toprule
 &  & \multicolumn{2}{c}{$\Rtdprox$} & \multicolumn{2}{c}{$\nig$} & \multicolumn{2}{c}{$\Rtd$} & \multicolumn{2}{c}{$\rhs$} \\
\cmidrule(lr){3-4}\cmidrule(lr){5-6}\cmidrule(lr){7-8}\cmidrule(lr){9-10}
DGP & $R^2$ & $\bar{\hat{R}}$ & $\overline{\text{ESS}}$ & $\bar{\hat{R}}$ & $\overline{\text{ESS}}$ & $\bar{\hat{R}}$ & $\overline{\text{ESS}}$ & $\bar{\hat{R}}$ & $\overline{\text{ESS}}$ \\
\midrule
\multirow{3}{*}{Sparse, correct} & 0.1 & 1.004 & 3,397 & \textbf{1.147} & 1,184 & 1.003 & 2,553 & 1.003 & 2,674 \\
 & 0.5 & 1.003 & 2,595 & \textbf{1.046} & 751 & 1.003 & 1,604 & 1.003 & 1,797 \\
 & 0.9 & 1.003 & 2,486 & 1.008 & 1,120 & 1.004 & 1,508 & 1.003 & 1,665 \\
\midrule
\multirow{3}{*}{Dense, correct} & 0.1 & 1.004 & 3,515 & \textbf{1.167} & 918 & 1.003 & 2,620 & 1.003 & 2,241 \\
 & 0.5 & 1.003 & 2,738 & \textbf{1.085} & 489 & 1.003 & 1,636 & 1.003 & 1,890 \\
 & 0.9 & 1.003 & 2,460 & 1.007 & 1,105 & 1.004 & 1,417 & 1.004 & 1,628 \\
\midrule
\multirow{3}{*}{Sparse, misp.} & 0.1 & 1.004 & 3,580 & \textbf{1.171} & 1,372 & 1.003 & 2,842 & 1.003 & 3,416 \\
 & 0.5 & 1.003 & 2,804 & \textbf{1.011} & 1,354 & 1.003 & 1,768 & 1.003 & 2,588 \\
 & 0.9 & 1.003 & 2,925 & 1.003 & 2,395 & 1.004 & 1,466 & 1.003 & 1,967 \\
\midrule
\multirow{3}{*}{Dense, misp.} & 0.1 & 1.004 & 3,631 & \textbf{1.174} & 1,002 & 1.003 & 2,895 & 1.003 & 2,786 \\
 & 0.5 & 1.003 & 2,885 & \textbf{1.022} & 1,030 & 1.003 & 1,793 & 1.003 & 1,979 \\
 & 0.9 & 1.003 & 3,136 & 1.003 & 2,045 & 1.004 & 1,573 & 1.004 & 1,515 \\
\bottomrule
\end{tabular}
}
\end{table}

\begin{table}
\caption{MCMC diagnostics for $\boldsymbol{\beta}$ ($p = 15$). For each DGP condition and model: mean (over simulation repetitions) of the per-repetition maximum $\hat{R}$ and minimum bulk ESS across $\beta$ coefficients. Bold entries indicate $\bar{\hat{R}} > 1.01$.\label{tab:mcmc_diagnostics_beta_p15}}
\resizebox{\textwidth}{!}{%
\begin{tabular}{@{}llrrrrrrrr@{}}
\toprule
 &  & \multicolumn{2}{c}{$\Rtdprox$} & \multicolumn{2}{c}{$\nig$} & \multicolumn{2}{c}{$\Rtd$} & \multicolumn{2}{c}{$\rhs$} \\
\cmidrule(lr){3-4}\cmidrule(lr){5-6}\cmidrule(lr){7-8}\cmidrule(lr){9-10}
DGP & $R^2$ & $\bar{\hat{R}}$ & $\overline{\text{ESS}}$ & $\bar{\hat{R}}$ & $\overline{\text{ESS}}$ & $\bar{\hat{R}}$ & $\overline{\text{ESS}}$ & $\bar{\hat{R}}$ & $\overline{\text{ESS}}$ \\
\midrule
\multirow{3}{*}{Sparse, correct} & 0.1 & 1.004 & 3,150 & \textbf{1.232} & 1,606 & 1.004 & 2,503 & 1.004 & 2,881 \\
 & 0.5 & 1.004 & 2,400 & \textbf{1.106} & 591 & 1.004 & 1,636 & 1.003 & 1,711 \\
 & 0.9 & 1.003 & 1,988 & \textbf{1.021} & 755 & 1.005 & 1,155 & 1.005 & 1,293 \\
\midrule
\multirow{3}{*}{Dense, correct} & 0.1 & 1.004 & 3,071 & \textbf{1.251} & 1,226 & 1.004 & 2,432 & 1.004 & 2,489 \\
 & 0.5 & 1.004 & 2,522 & \textbf{1.155} & 483 & 1.003 & 1,772 & 1.003 & 1,870 \\
 & 0.9 & 1.003 & 2,016 & \textbf{1.019} & 773 & 1.005 & 1,208 & 1.005 & 1,293 \\
\midrule
\multirow{3}{*}{Sparse, misp.} & 0.1 & 1.004 & 3,087 & \textbf{1.226} & 1,558 & 1.004 & 2,362 & 1.004 & 4,354 \\
 & 0.5 & 1.003 & 2,673 & \textbf{1.027} & 1,002 & 1.003 & 1,972 & 1.003 & 2,385 \\
 & 0.9 & 1.003 & 2,166 & 1.005 & 1,598 & 1.005 & 1,321 & 1.004 & 1,666 \\
\midrule
\multirow{3}{*}{Dense, misp.} & 0.1 & 1.004 & 3,149 & \textbf{1.230} & 1,352 & 1.004 & 2,388 & 1.003 & 2,705 \\
 & 0.5 & 1.003 & 2,710 & \textbf{1.059} & 886 & 1.003 & 2,003 & 1.003 & 2,128 \\
 & 0.9 & 1.003 & 2,360 & 1.005 & 1,486 & 1.004 & 1,411 & 1.004 & 1,548 \\
\bottomrule
\end{tabular}
}
\end{table}

\begin{table}
\caption{MCMC diagnostics for $\boldsymbol{\beta}$ ($p = 25$). For each DGP condition and model: mean (over simulation repetitions) of the per-repetition maximum $\hat{R}$ and minimum bulk ESS across $\beta$ coefficients. Bold entries indicate $\bar{\hat{R}} > 1.01$.\label{tab:mcmc_diagnostics_beta_p25}}
\resizebox{\textwidth}{!}{%
\begin{tabular}{@{}llrrrrrrrr@{}}
\toprule
 &  & \multicolumn{2}{c}{$\Rtdprox$} & \multicolumn{2}{c}{$\nig$} & \multicolumn{2}{c}{$\Rtd$} & \multicolumn{2}{c}{$\rhs$} \\
\cmidrule(lr){3-4}\cmidrule(lr){5-6}\cmidrule(lr){7-8}\cmidrule(lr){9-10}
DGP & $R^2$ & $\bar{\hat{R}}$ & $\overline{\text{ESS}}$ & $\bar{\hat{R}}$ & $\overline{\text{ESS}}$ & $\bar{\hat{R}}$ & $\overline{\text{ESS}}$ & $\bar{\hat{R}}$ & $\overline{\text{ESS}}$ \\
\midrule
\multirow{3}{*}{Sparse, correct} & 0.1 & 1.005 & 2,467 & \textbf{1.257} & 1,515 & 1.005 & 2,873 & 1.005 & 2,205 \\
 & 0.5 & 1.005 & 2,552 & \textbf{1.170} & 593 & 1.005 & 2,197 & 1.005 & 1,595 \\
 & 0.9 & 1.005 & 1,820 & \textbf{1.064} & 333 & 1.005 & 1,283 & 1.007 & 885 \\
\midrule
\multirow{3}{*}{Dense, correct} & 0.1 & 1.005 & 2,514 & \textbf{1.317} & 1,298 & 1.005 & 2,919 & 1.005 & 2,170 \\
 & 0.5 & 1.005 & 2,789 & \textbf{1.244} & 480 & 1.004 & 2,434 & 1.004 & 1,646 \\
 & 0.9 & 1.004 & 2,174 & \textbf{1.091} & 310 & 1.006 & 1,358 & \textbf{1.013} & 847 \\
\midrule
\multirow{3}{*}{Sparse, misp.} & 0.1 & 1.005 & 2,513 & \textbf{1.263} & 1,941 & 1.005 & 2,872 & 1.006 & 3,303 \\
 & 0.5 & 1.005 & 2,996 & \textbf{1.064} & 963 & 1.004 & 2,625 & 1.005 & 2,405 \\
 & 0.9 & 1.005 & 2,228 & \textbf{1.043} & 540 & 1.005 & 1,664 & \textbf{1.011} & 1,236 \\
\midrule
\multirow{3}{*}{Dense, misp.} & 0.1 & 1.005 & 2,543 & \textbf{1.228} & 1,498 & 1.005 & 2,949 & 1.005 & 2,272 \\
 & 0.5 & 1.005 & 3,120 & \textbf{1.115} & 774 & 1.004 & 2,759 & 1.004 & 2,063 \\
 & 0.9 & 1.004 & 2,350 & \textbf{1.032} & 479 & 1.009 & 1,691 & \textbf{1.011} & 1,035 \\
\bottomrule
\end{tabular}
}
\end{table}

\newpage

\subsection{Constraint feasibility: Different DGPs}

\begin{figure}[htbp]
    \centering
    \includegraphics[width=\linewidth]{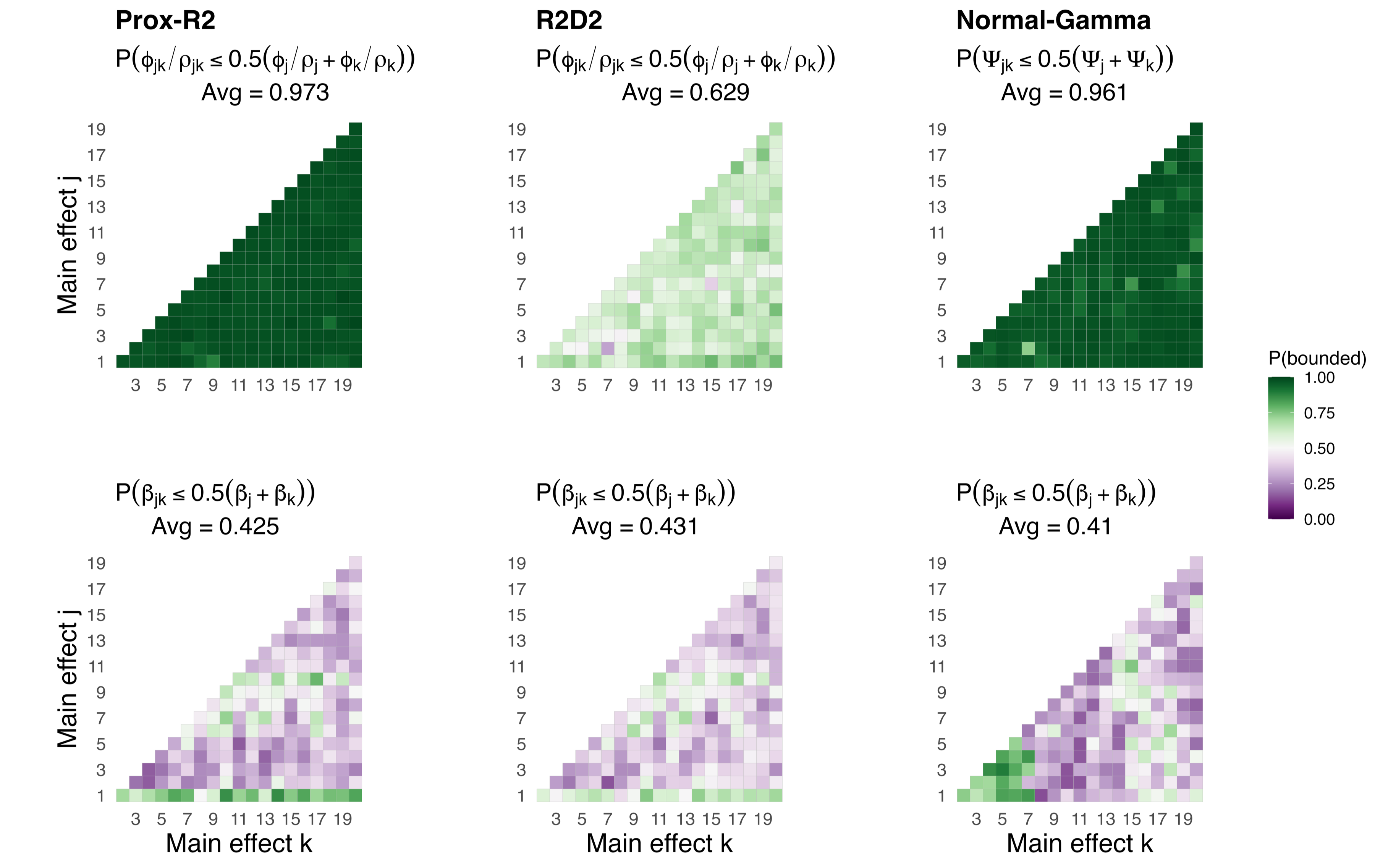}
    \caption{Constraint feasibility heatmaps for $p=20$, sparse, misspecified DGP~(H2) with $R^2 = 0.5$. Panels compare the $\Rtdprox$ model against the unconstrained $\Rtd$ and the $\nig$ models.}
    \label{fig:feasibility_sparse_misp}
\end{figure}

\begin{figure}[htbp]
    \centering
    \includegraphics[width=\linewidth]{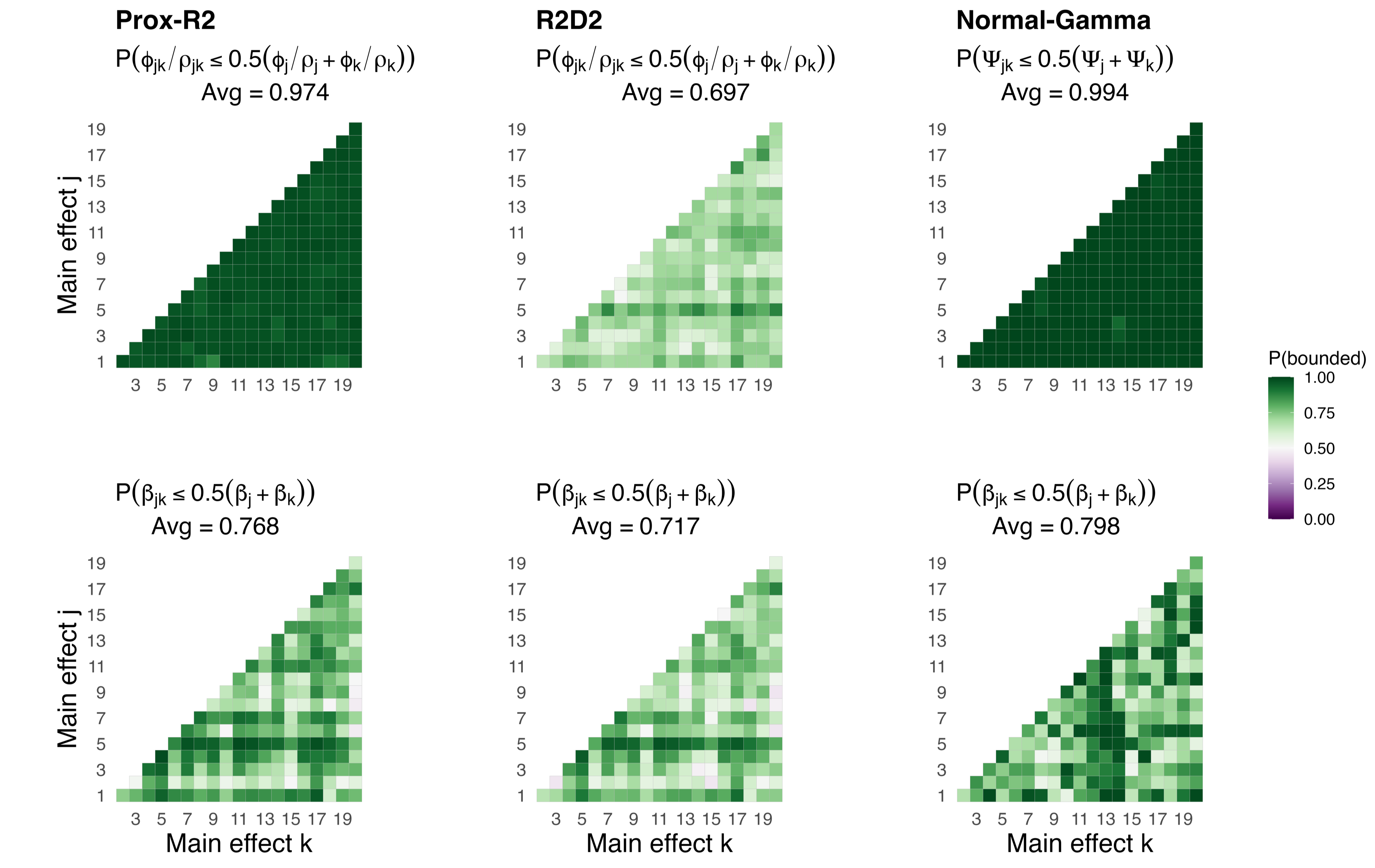}
    \caption{Constraint feasibility heatmaps for $p=20$, dense, correctly specified DGP~(H1) with $R^2 = 0.5$. Panels compare the $\Rtdprox$ model against the unconstrained $\Rtd$ and the $\nig$ models.}
    \label{fig:feasibility_dense_correct}
\end{figure}

\begin{figure}[htbp]
    \centering
    \includegraphics[width=\linewidth]{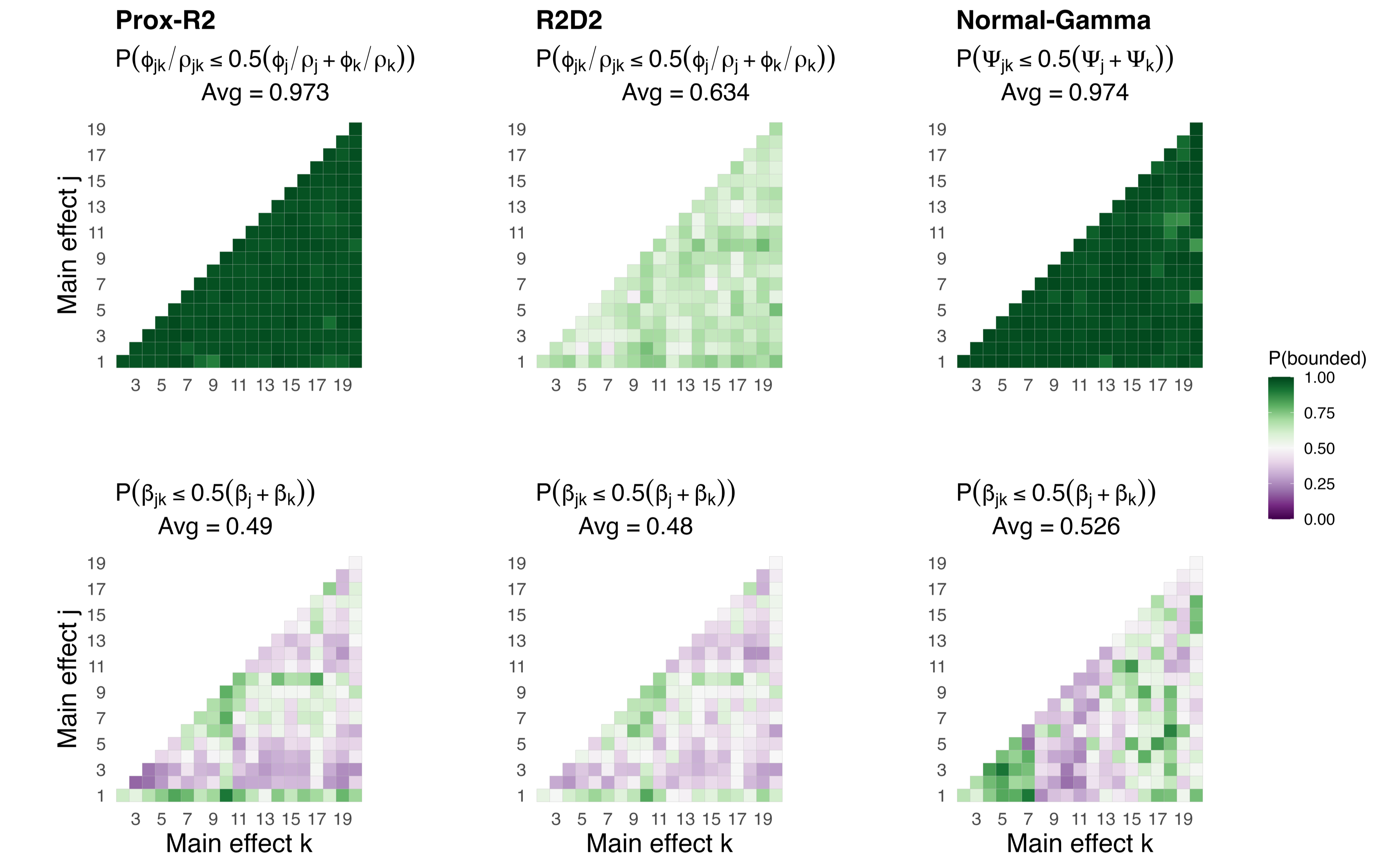}
    \caption{Constraint feasibility heatmaps for $p=20$, dense, misspecified DGP~(H2) with $R^2 = 0.5$. Panels compare the $\Rtdprox$ model against the unconstrained $\Rtd$ and the $\nig$ models.}
    \label{fig:feasibility_dense_misp}
\end{figure}

\begin{table}
\caption{Empirical coverage of the central 90\% posterior interval $[q_{5}, q_{95}]$ for $\boldsymbol{\beta}$ ($p = 10$). Coverage is pooled over coefficients and simulation repetitions, split by coefficient type (main vs.\ interaction) and zero status (true non-zero vs.\ true zero). Bold entries deviate from the nominal 0.90 by more than 0.05.\label{tab:coverage_beta_p10}}
\resizebox{\textwidth}{!}{%
\begin{tabular}{@{}lllrrrrrrrr@{}}
\toprule
 &  &  & \multicolumn{2}{c}{$\Rtdprox$} & \multicolumn{2}{c}{$\nig$} & \multicolumn{2}{c}{$\Rtd$} & \multicolumn{2}{c}{$\rhs$} \\
\cmidrule(lr){4-5}\cmidrule(lr){6-7}\cmidrule(lr){8-9}\cmidrule(lr){10-11}
DGP & $R^2$ & Coef. & non-0 & zero & non-0 & zero & non-0 & zero & non-0 & zero \\
\midrule
\multirow{6}{*}{Sparse, correct} & \multirow{2}{*}{0.1} & main & 0.93 & \textbf{1.00} & 0.87 & \textbf{1.00} & \textbf{0.78} & \textbf{1.00} & \textbf{0.70} & \textbf{1.00} \\
 &  & int. & 0.93 & \textbf{1.00} & \textbf{0.48} & \textbf{1.00} & 0.93 & \textbf{1.00} & 0.90 & \textbf{1.00} \\
 & \multirow{2}{*}{0.5} & main & 0.91 & \textbf{0.98} & 0.91 & \textbf{0.98} & \textbf{0.84} & \textbf{1.00} & \textbf{0.84} & \textbf{1.00} \\
 &  & int. & 0.88 & \textbf{0.99} & \textbf{0.68} & \textbf{1.00} & 0.88 & \textbf{0.99} & 0.87 & \textbf{0.99} \\
 & \multirow{2}{*}{0.9} & main & 0.91 & \textbf{0.96} & 0.91 & 0.95 & 0.90 & \textbf{0.99} & 0.89 & \textbf{0.98} \\
 &  & int. & 0.88 & \textbf{0.97} & \textbf{0.82} & \textbf{0.99} & 0.86 & \textbf{0.98} & 0.86 & \textbf{0.98} \\
\midrule
\multirow{6}{*}{Dense, correct} & \multirow{2}{*}{0.1} & main & \textbf{0.98} & -- & 0.93 & -- & 0.92 & -- & \textbf{0.82} & -- \\
 &  & int. & \textbf{0.95} & -- & \textbf{0.29} & -- & \textbf{0.95} & -- & 0.89 & -- \\
 & \multirow{2}{*}{0.5} & main & 0.95 & -- & 0.90 & -- & 0.90 & -- & 0.88 & -- \\
 &  & int. & 0.88 & -- & \textbf{0.57} & -- & 0.88 & -- & 0.86 & -- \\
 & \multirow{2}{*}{0.9} & main & 0.91 & -- & 0.89 & -- & 0.88 & -- & 0.87 & -- \\
 &  & int. & 0.87 & -- & \textbf{0.78} & -- & 0.85 & -- & \textbf{0.84} & -- \\
\midrule
\multirow{6}{*}{Sparse, misp.} & \multirow{2}{*}{0.1} & main & \textbf{0.99} & \textbf{1.00} & \textbf{0.98} & \textbf{1.00} & \textbf{1.00} & \textbf{1.00} & \textbf{0.99} & \textbf{1.00} \\
 &  & int. & \textbf{1.00} & \textbf{1.00} & \textbf{0.40} & \textbf{1.00} & \textbf{1.00} & \textbf{1.00} & \textbf{1.00} & \textbf{1.00} \\
 & \multirow{2}{*}{0.5} & main & \textbf{0.97} & \textbf{0.99} & \textbf{0.96} & \textbf{0.99} & \textbf{0.96} & \textbf{1.00} & \textbf{0.97} & \textbf{1.00} \\
 &  & int. & \textbf{0.99} & \textbf{0.99} & 0.94 & \textbf{1.00} & \textbf{0.99} & \textbf{0.99} & \textbf{0.99} & \textbf{0.99} \\
 & \multirow{2}{*}{0.9} & main & 0.93 & \textbf{0.96} & 0.91 & 0.94 & 0.91 & \textbf{0.99} & 0.92 & \textbf{0.99} \\
 &  & int. & \textbf{0.97} & \textbf{0.98} & 0.94 & \textbf{0.99} & \textbf{0.96} & \textbf{0.98} & \textbf{0.96} & \textbf{0.98} \\
\midrule
\multirow{6}{*}{Dense, misp.} & \multirow{2}{*}{0.1} & main & \textbf{0.99} & -- & \textbf{0.99} & -- & \textbf{1.00} & -- & \textbf{0.99} & -- \\
 &  & int. & \textbf{1.00} & -- & \textbf{0.56} & -- & \textbf{1.00} & -- & \textbf{0.99} & -- \\
 & \multirow{2}{*}{0.5} & main & \textbf{0.98} & -- & \textbf{0.97} & -- & \textbf{0.98} & -- & \textbf{0.98} & -- \\
 &  & int. & \textbf{1.00} & -- & 0.94 & -- & \textbf{0.99} & -- & \textbf{0.99} & -- \\
 & \multirow{2}{*}{0.9} & main & 0.95 & -- & 0.92 & -- & 0.95 & -- & 0.94 & -- \\
 &  & int. & \textbf{0.98} & -- & \textbf{0.95} & -- & \textbf{0.96} & -- & \textbf{0.95} & -- \\
\bottomrule
\end{tabular}
}
\end{table}

\begin{table}
\caption{Empirical coverage of the central 90\% posterior interval $[q_{5}, q_{95}]$ for $\boldsymbol{\beta}$ ($p = 15$). Coverage is pooled over coefficients and simulation repetitions, split by coefficient type (main vs.\ interaction) and zero status (true non-zero vs.\ true zero). Bold entries deviate from the nominal 0.90 by more than 0.05.\label{tab:coverage_beta_p15}}
\resizebox{\textwidth}{!}{%
\begin{tabular}{@{}lllrrrrrrrr@{}}
\toprule
 &  &  & \multicolumn{2}{c}{$\Rtdprox$} & \multicolumn{2}{c}{$\nig$} & \multicolumn{2}{c}{$\Rtd$} & \multicolumn{2}{c}{$\rhs$} \\
\cmidrule(lr){4-5}\cmidrule(lr){6-7}\cmidrule(lr){8-9}\cmidrule(lr){10-11}
DGP & $R^2$ & Coef. & non-0 & zero & non-0 & zero & non-0 & zero & non-0 & zero \\
\midrule
\multirow{6}{*}{Sparse, correct} & \multirow{2}{*}{0.1} & main & \textbf{0.77} & \textbf{1.00} & \textbf{0.72} & \textbf{1.00} & \textbf{0.46} & \textbf{1.00} & \textbf{0.39} & \textbf{1.00} \\
 &  & int. & 0.88 & \textbf{1.00} & \textbf{0.35} & \textbf{1.00} & 0.89 & \textbf{1.00} & \textbf{0.83} & \textbf{1.00} \\
 & \multirow{2}{*}{0.5} & main & \textbf{0.84} & \textbf{1.00} & 0.89 & \textbf{1.00} & \textbf{0.73} & \textbf{1.00} & \textbf{0.73} & \textbf{1.00} \\
 &  & int. & \textbf{0.83} & \textbf{1.00} & \textbf{0.57} & \textbf{1.00} & \textbf{0.84} & \textbf{1.00} & \textbf{0.82} & \textbf{1.00} \\
 & \multirow{2}{*}{0.9} & main & 0.90 & \textbf{0.98} & 0.91 & \textbf{0.97} & 0.88 & \textbf{1.00} & 0.86 & \textbf{1.00} \\
 &  & int. & \textbf{0.84} & \textbf{1.00} & \textbf{0.76} & \textbf{1.00} & \textbf{0.84} & \textbf{1.00} & \textbf{0.83} & \textbf{1.00} \\
\midrule
\multirow{6}{*}{Dense, correct} & \multirow{2}{*}{0.1} & main & \textbf{0.99} & -- & \textbf{0.95} & -- & 0.89 & -- & \textbf{0.71} & -- \\
 &  & int. & 0.94 & -- & \textbf{0.25} & -- & 0.95 & -- & \textbf{0.85} & -- \\
 & \multirow{2}{*}{0.5} & main & \textbf{0.97} & -- & 0.95 & -- & 0.91 & -- & 0.87 & -- \\
 &  & int. & 0.87 & -- & \textbf{0.50} & -- & 0.88 & -- & 0.85 & -- \\
 & \multirow{2}{*}{0.9} & main & \textbf{0.96} & -- & 0.93 & -- & 0.92 & -- & 0.89 & -- \\
 &  & int. & 0.86 & -- & \textbf{0.75} & -- & 0.86 & -- & \textbf{0.84} & -- \\
\midrule
\multirow{6}{*}{Sparse, misp.} & \multirow{2}{*}{0.1} & main & \textbf{1.00} & \textbf{1.00} & \textbf{0.99} & \textbf{1.00} & \textbf{1.00} & \textbf{1.00} & \textbf{1.00} & \textbf{1.00} \\
 &  & int. & \textbf{1.00} & \textbf{1.00} & \textbf{0.37} & \textbf{1.00} & \textbf{1.00} & \textbf{1.00} & \textbf{1.00} & \textbf{1.00} \\
 & \multirow{2}{*}{0.5} & main & \textbf{0.98} & \textbf{1.00} & \textbf{0.97} & \textbf{1.00} & \textbf{0.97} & \textbf{1.00} & \textbf{0.98} & \textbf{1.00} \\
 &  & int. & \textbf{1.00} & \textbf{1.00} & 0.92 & \textbf{1.00} & \textbf{0.99} & \textbf{1.00} & \textbf{1.00} & \textbf{1.00} \\
 & \multirow{2}{*}{0.9} & main & \textbf{0.96} & \textbf{0.98} & 0.94 & \textbf{0.98} & \textbf{0.95} & \textbf{1.00} & \textbf{0.95} & \textbf{1.00} \\
 &  & int. & \textbf{0.99} & \textbf{1.00} & \textbf{0.96} & \textbf{1.00} & \textbf{0.98} & \textbf{1.00} & \textbf{0.98} & \textbf{1.00} \\
\midrule
\multirow{6}{*}{Dense, misp.} & \multirow{2}{*}{0.1} & main & \textbf{1.00} & -- & \textbf{1.00} & -- & \textbf{1.00} & -- & \textbf{1.00} & -- \\
 &  & int. & \textbf{1.00} & -- & \textbf{0.66} & -- & \textbf{1.00} & -- & \textbf{1.00} & -- \\
 & \multirow{2}{*}{0.5} & main & \textbf{0.99} & -- & \textbf{0.98} & -- & \textbf{0.99} & -- & \textbf{0.99} & -- \\
 &  & int. & \textbf{1.00} & -- & 0.95 & -- & \textbf{1.00} & -- & \textbf{1.00} & -- \\
 & \multirow{2}{*}{0.9} & main & \textbf{0.98} & -- & \textbf{0.96} & -- & \textbf{0.97} & -- & \textbf{0.97} & -- \\
 &  & int. & \textbf{0.99} & -- & \textbf{0.98} & -- & \textbf{0.99} & -- & \textbf{0.98} & -- \\
\bottomrule
\end{tabular}
}
\end{table}

\begin{table}
\caption{Empirical coverage of the central 90\% posterior interval $[q_{5}, q_{95}]$ for $\boldsymbol{\beta}$ ($p = 25$). Coverage is pooled over coefficients and simulation repetitions, split by coefficient type (main vs.\ interaction) and zero status (true non-zero vs.\ true zero). Bold entries deviate from the nominal 0.90 by more than 0.05.\label{tab:coverage_beta_p25}}
\resizebox{\textwidth}{!}{%
\begin{tabular}{@{}lllrrrrrrrr@{}}
\toprule
 &  &  & \multicolumn{2}{c}{$\Rtdprox$} & \multicolumn{2}{c}{$\nig$} & \multicolumn{2}{c}{$\Rtd$} & \multicolumn{2}{c}{$\rhs$} \\
\cmidrule(lr){4-5}\cmidrule(lr){6-7}\cmidrule(lr){8-9}\cmidrule(lr){10-11}
DGP & $R^2$ & Coef. & non-0 & zero & non-0 & zero & non-0 & zero & non-0 & zero \\
\midrule
\multirow{6}{*}{Sparse, correct} & \multirow{2}{*}{0.1} & main & \textbf{0.67} & \textbf{1.00} & \textbf{0.74} & \textbf{1.00} & \textbf{0.26} & \textbf{1.00} & \textbf{0.28} & \textbf{1.00} \\
 &  & int. & \textbf{0.84} & \textbf{1.00} & \textbf{0.37} & \textbf{1.00} & 0.85 & \textbf{1.00} & \textbf{0.77} & \textbf{1.00} \\
 & \multirow{2}{*}{0.5} & main & \textbf{0.67} & \textbf{1.00} & \textbf{0.82} & \textbf{1.00} & \textbf{0.45} & \textbf{1.00} & \textbf{0.55} & \textbf{1.00} \\
 &  & int. & \textbf{0.72} & \textbf{1.00} & \textbf{0.47} & \textbf{1.00} & \textbf{0.74} & \textbf{1.00} & \textbf{0.72} & \textbf{1.00} \\
 & \multirow{2}{*}{0.9} & main & \textbf{0.79} & \textbf{1.00} & 0.86 & \textbf{0.99} & \textbf{0.73} & \textbf{1.00} & \textbf{0.74} & \textbf{1.00} \\
 &  & int. & \textbf{0.71} & \textbf{1.00} & \textbf{0.65} & \textbf{1.00} & \textbf{0.74} & \textbf{1.00} & \textbf{0.71} & \textbf{1.00} \\
\midrule
\multirow{6}{*}{Dense, correct} & \multirow{2}{*}{0.1} & main & \textbf{1.00} & -- & \textbf{0.99} & -- & 0.93 & -- & \textbf{0.71} & -- \\
 &  & int. & \textbf{0.96} & -- & \textbf{0.41} & -- & \textbf{0.97} & -- & 0.87 & -- \\
 & \multirow{2}{*}{0.5} & main & \textbf{0.98} & -- & \textbf{0.96} & -- & \textbf{0.85} & -- & \textbf{0.79} & -- \\
 &  & int. & 0.88 & -- & \textbf{0.46} & -- & 0.89 & -- & \textbf{0.85} & -- \\
 & \multirow{2}{*}{0.9} & main & \textbf{0.97} & -- & \textbf{0.95} & -- & 0.90 & -- & 0.87 & -- \\
 &  & int. & 0.86 & -- & \textbf{0.67} & -- & 0.87 & -- & \textbf{0.84} & -- \\
\midrule
\multirow{6}{*}{Sparse, misp.} & \multirow{2}{*}{0.1} & main & \textbf{1.00} & \textbf{1.00} & \textbf{1.00} & \textbf{1.00} & \textbf{1.00} & \textbf{1.00} & \textbf{1.00} & \textbf{1.00} \\
 &  & int. & \textbf{1.00} & \textbf{1.00} & \textbf{0.31} & \textbf{1.00} & \textbf{1.00} & \textbf{1.00} & \textbf{1.00} & \textbf{1.00} \\
 & \multirow{2}{*}{0.5} & main & \textbf{0.99} & \textbf{1.00} & \textbf{0.98} & \textbf{1.00} & \textbf{0.99} & \textbf{1.00} & \textbf{0.99} & \textbf{1.00} \\
 &  & int. & \textbf{1.00} & \textbf{1.00} & 0.89 & \textbf{1.00} & \textbf{0.99} & \textbf{1.00} & \textbf{1.00} & \textbf{1.00} \\
 & \multirow{2}{*}{0.9} & main & \textbf{0.98} & \textbf{1.00} & \textbf{0.96} & \textbf{0.99} & \textbf{0.96} & \textbf{1.00} & \textbf{0.97} & \textbf{1.00} \\
 &  & int. & \textbf{0.99} & \textbf{1.00} & \textbf{0.95} & \textbf{1.00} & \textbf{0.99} & \textbf{1.00} & \textbf{0.99} & \textbf{1.00} \\
\midrule
\multirow{6}{*}{Dense, misp.} & \multirow{2}{*}{0.1} & main & \textbf{1.00} & -- & \textbf{1.00} & -- & \textbf{1.00} & -- & \textbf{1.00} & -- \\
 &  & int. & \textbf{1.00} & -- & \textbf{0.85} & -- & \textbf{1.00} & -- & \textbf{1.00} & -- \\
 & \multirow{2}{*}{0.5} & main & \textbf{1.00} & -- & \textbf{1.00} & -- & \textbf{1.00} & -- & \textbf{1.00} & -- \\
 &  & int. & \textbf{1.00} & -- & 0.95 & -- & \textbf{1.00} & -- & \textbf{1.00} & -- \\
 & \multirow{2}{*}{0.9} & main & \textbf{1.00} & -- & \textbf{0.98} & -- & \textbf{1.00} & -- & \textbf{0.99} & -- \\
 &  & int. & \textbf{1.00} & -- & \textbf{0.99} & -- & \textbf{1.00} & -- & \textbf{1.00} & -- \\
\bottomrule
\end{tabular}
}
\end{table}

\clearpage

\subsection{Predictive performance on benchmark data}
\label{app:predperf}

\begin{table}
\caption{Hold-out predictive performance on three real datasets (75\% of each held out). RMSE, elpd, and $R^2$ are evaluated on the held-out test set. RMSE: posterior-mean
predictions.
train-$R^2$: posterior-mean Bayesian $R^2$.
test-$R^2$: hold-out Bayesian $R^2$.\label{tab:predperf}}
\resizebox{\textwidth}{!}{%
\begin{tabular}{lrrrrr}
\toprule
Model & RMSE &  $\text{train-}R^2$ & $\text{test-}R^2$ & $\text{elpd}_{\text{diff}}$ & $\text{SE}_{\text{diff}}$ \\
\midrule
\multicolumn{6}{l}{\textit{Boston Housing}} \\
\cmidrule(l){1-1}
\quad $\Rtdprox$ & 4.0 & 0.86 & 0.79 & 2.9 & 1.4 \\
\quad $\nig$ & 4.53 & 0.87 & 0.73 & 7.5 & 2.9 \\
\quad $\Rtd$ & 4.30 &  0.87 & 0.74 & 0.0 & 0.0 \\
\quad $\rhs$ & 4.84 &  0.87 & 0.70 & 2.4 & 2.2 \\
\addlinespace
\multicolumn{6}{l}{\textit{Ozone}} \\
\cmidrule(l){1-1}
\quad $\Rtdprox$ & 4.21 &  0.80 & 0.73 & 0.8 & 0.8 \\
\quad $\nig$ & 4.24 &  0.81 & 0.73 & -0.4 & 1.3 \\
\quad $\Rtd$ & 4.21 &  0.81 & 0.73 & 0.0 & 0.0 \\
\quad $\rhs$ & 4.22 &  0.81 & 0.73 & 0.3 & 0.8 \\
\addlinespace
\multicolumn{6}{l}{\textit{Blood Glucose}} \\
\cmidrule(l){1-1}
\quad $\Rtdprox$ & 57.79 &  0.61 & 0.45 & 2.2 & 1.0 \\
\quad $\nig$ & 57.09 &  0.58 & 0.47 & 5.7 & 2.7 \\
\quad $\Rtd$ & 58.16 &  0.60 & 0.45 & 0.0 & 0.0 \\
\quad $\rhs$ & 57.91 &  0.56 & 0.45 & 1.4 & 1.3 \\
\bottomrule
\end{tabular}
}
\end{table}

We apply the interaction models alongside the same set of competing priors as in the simulation section to the Boston Housing \citep{harrison_hedonic_1978}, Ozone \citep{breiman_estimating_1985}, and Blood Glucose \citep{efron_least_2004} data sets. In each data set, the model includes $p$ main effects and $p_{\mathrm{int}} = \binom{p}{2}$ interaction effects, giving $(p,\,p_{\mathrm{int}}) = (13,\,78)$ for Boston Housing, $(9,\,36)$ for Ozone, and $(10,\,45)$ for Blood Glucose. To make estimation more challenging, we hold out $75\%$ of the observations as a true out-of-sample test set and fit each model on the remaining $25\%$, so that $p_{\mathrm{int}}$ is large relative to the training sample. Table~\ref{tab:predperf} summarises the results.

Overall, the results confirm the good performance of the $\Rtdprox$ model. While competitive across all data sets, it outperforms on point predictions for the Boston Housing data set and shows the largest fraction of explained variance on the test-data. 

As in the simulation study of Section~\ref{sec:experiments}, and consistent with the diagnostics reported in Tables~\ref{tab:mcmc_diagnostics_beta_p10}--\ref{tab:mcmc_diagnostics_beta_p25}, the $\Rtdprox$ model again shows a large benefit in terms of sampling efficiency.

\newpage
\section{Proofs}

\subsection{Marginal prior behaviour}

Throughout this subsection, let $\gamma\mid\omega^2\sim\mathsf{Normal}(0,\omega^2)$ and $\omega^2\sim\mathsf{BetaPrime}(a,b)$, independently of $\phi_i$. Then $\beta_i$ has the product representation $\beta_i\overset{d}=\sqrt{\phi_i}\,\gamma$.

\begin{proposition}
Let $\omega^2 \sim \mathsf{BetaPrime}(a,b)$ with $a>0$ and $b>0$. Let $\beta_i \mid \phi_i , \omega^2 \sim \mathsf{Normal} (0,\phi_i \omega^2)$ and let $\phi_i \sim p(\phi_i)$ where $\phi_i \in[0,1]$ is a continuous random variable independent of $\omega^2$. Then, as $|s|\to\infty$,
\begin{equation}
p_{\beta_i}(s)\sim C_{a,b}\,\mathbb{E}[\phi_i^{\,b}]\,|s|^{-(2b+1)}.    
\end{equation}
where $C_{a,b}=\frac{2^{b}\Gamma \left(b+\tfrac12\right)}{\sqrt{\pi}\,B(a,b)}$.
\end{proposition}

\begin{proof}

The beta-prime density for $\omega^2$ is
\begin{equation}
f_{\omega^2}(t) =\frac{t^{a-1}(1+t)^{-(a+b)}}{B(a,b)}, \qquad t>0.
\end{equation}
Integrating the conditional normal density of $\gamma$ over $\omega^2$
\begin{equation}
p_\gamma(s)
=\int_0^\infty
\frac{1}{\sqrt{2\pi t}}
\exp\left(-\frac{s^2}{2t}\right)
f_{\omega^2}(t)\,dt.
\end{equation}
For $s \neq 0$, we let $u = s^2/(2t)$ and have so that $t=s^2/(2u)$ and $dt=-s^2(2u^2)^{-1}du$. Using
\begin{equation}
f_{\omega^2}(t)
=\frac{t^{-(b+1)}}{B(a,b)}(1+t^{-1})^{-(a+b)},
\end{equation}
and simplifying all powers of $s$ and $u$ gives the identity
\begin{equation}
\label{eq:gamma-tail-exact}
|s|^{2b+1}p_\gamma(s) =\frac{2^b}{\sqrt\pi B(a,b)} \int_0^\infty e^{-u}u^{b-1/2} \left(1+\frac{2u}{s^2}\right)^{-(a+b)}du.
\end{equation}
To justify the limit, take an arbitrary sequence $\{s_m\}$ with $|s_m|\to\infty$ and define, for $u>0$,
\begin{equation}
f_m(u) =e^{-u}u^{b-1/2} \left(1+\frac{2u}{s_m^2}\right)^{-(a+b)}.
\end{equation}
For each fixed $u>0$,
\begin{equation}
f_m(u)\longrightarrow e^{-u}u^{b-1/2}.
\end{equation}
Also, $0\leq f_m(u)\leq g(u)$ with $g(u)=e^{-u}u^{b-1/2},$ since the final factor in $f_m$ lies in $(0,1]$. Moreover, the dominating function is integrable since $\int_0^\infty g(u)\,du=\Gamma(b+1/2)<\infty$ for $b >0$. By dominated convergence,
\begin{equation}
\begin{aligned}
\lim_{m\to\infty}|s_m|^{2b+1}p_\gamma(s_m)  &=\frac{2^b}{\sqrt\pi B(a,b)} \lim_{m\to\infty}\int_0^\infty f_m(u) du   =\frac{2^b}{\sqrt\pi B(a,b)} \int_0^\infty\lim_{m\to\infty}f_m(u) du\\
&=\frac{2^b\Gamma(b+1/2)}{\sqrt\pi B(a,b)} =C_{a,b},
\end{aligned}
\end{equation}

showing that $\lim_{|s| \to \infty} |s|^{2b+1} p_{\gamma_i}(s) \sim C_{a,b}$.

Conditioning on $\phi_i = \phi$ leads to 
\begin{equation}
p_{\beta_i}(s\mid\phi) =\phi^{-1/2}p_\gamma(s/\sqrt\phi).
\end{equation}
Now define
\begin{equation}
R(r)=r^{2b+1}p_\gamma(r), \qquad r>0.
\end{equation}
Equation~\eqref{eq:gamma-tail-exact} implies both $R(r)\to C_{a,b}$ as $r\to\infty$ and $0\le R(r)\le C_{a,b}$ for every $r>0$. Therefore, 
\begin{equation}
|s|^{2b+1}p_{\beta_i}(s) =\mathbb{E}\left[ \phi_i^b R\left(\frac{|s|}{\sqrt{\phi_i}}\right) \right].
\end{equation}
Now, let $|s_m|$ be an arbitrary sequence such that $|s_m|\to\infty$ and define  
\begin{equation}
H_m(\phi) =\phi^b R\left(\frac{|s_m|}{\sqrt\phi}\right).
\end{equation}
For every fixed $\phi\in(0,1]$,
\begin{equation}
H_m(\phi)\longrightarrow C_{a,b}\phi^b.
\end{equation}
Moreover
\begin{equation}
0\le H_m(\phi)\le C_{a,b}\phi^b\le C_{a,b}.
\end{equation}
Hence the $C_{a,b}$ is an integrable dominating function under the probability law of $\phi_i$. By dominated convergence,
\begin{equation}
\begin{aligned}
\lim_{m\to\infty}|s_m|^{2b+1}p_{\beta_i}(s_m) &=\lim_{m\to\infty}\mathbb{E}\{H_m(\phi_i)\} =  \mathbb{E}\left\{\lim_{m\to\infty}H_m(\phi_i)\right\}=C_{a,b}\mathbb{E}(\phi_i^b),
\end{aligned}
\end{equation}
which is the required tail asymptotic.
\end{proof}

We now demonstrate Proposition~\ref{prop:origin}. We first show the following lemma, which exhibits the negative moments needed when the allocation $\phi$ dominates.

\begin{lemma}
Let $\gamma\mid\omega^2\sim\mathsf{Normal}(0,\omega^2)$ and $\omega^2\sim\mathsf{BetaPrime}(a,b)$, where $a,b>0$. If $0<a<1/2$, then $p_\gamma(s)\sim A_{a,b}|s|^{2a-1}$,
where $A_{a,b} > 0$. Moreover, for $0<c<\min(a,1/2)$, $\mathbb E|\gamma|^{-2c} < \infty$
\end{lemma}

\begin{proof}
Starting from the normal mixture integral presented above and letting $u=s^2/(2t)$ yields
\begin{equation}
\label{eq:gamma-origin-exact}
p_\gamma(s) =\frac{2^{-a}|s|^{2a-1}}{\sqrt\pi B(a,b)} \int_0^\infty e^{-u}u^{-a-1/2} \left(1+\frac{s^2}{2u}\right)^{-(a+b)}du.
\end{equation}
Let $\{s_m\}$ be an arbitraty sequence of nonzero numbers satisfying $s_m\to0$. Define
\begin{equation}
f_m(u)=e^{-u}u^{-a-1/2} \left(1+\frac{s_m^2}{2u}\right)^{-(a+b)}.
\end{equation}
For every $u>0$,
\begin{equation}
f_m(u)\longrightarrow e^{-u}u^{-a-1/2}.
\end{equation}
Additionally,
\begin{equation}
0\le f_m(u)\le g(u), \qquad g(u)=e^{-u}u^{-a-1/2}.
\end{equation}
The condition $a<1/2$ makes $g$ integrable both near zero and at infinity since $\int_0^\infty g(u)\,du=\Gamma(1/2-a)<\infty$. Dominated convergence yields
\begin{equation}
\begin{aligned}
\lim_{m\to\infty}|s_m|^{1-2a}p_\gamma(s_m) 
&=\frac{2^{-a}}{\sqrt\pi B(a,b)} \lim_{m\to\infty}\int_0^\infty f_m(u) du =\frac{2^{-a}}{\sqrt\pi B(a,b)} \int_0^\infty\lim_{m\to\infty}f_m(u) du\\
&=\frac{2^{-a}\Gamma(1/2-a)}{\sqrt\pi B(a,b)} =A_{a,b}.
\end{aligned}
\end{equation}
Since $\gamma=\sqrt{\omega^2}Z$, where $Z\sim\mathsf{Normal}(0,1)$ is independent of $\omega^2$, by Tonelli's theorem 
\begin{equation}
\mathbb{E}|\gamma|^{-2c} = \mathbb{E} \{(\omega^2)^{-c}\} \mathbb{E}|Z|^{-2c}.
\end{equation}
The beta-prime factor is given by
\begin{equation}
\begin{aligned}
\mathbb{E}\{(\omega^2)^{-c}\} &=\frac{1}{B(a,b)} \int_0^\infty t^{a-c-1}(1+t)^{-(a+b)}dt =\frac{B(a-c,b+c)}{B(a,b)},
\end{aligned}
\end{equation}
which is finite when $c<a$. Using the symmetry of $Z$ and letting $x = z^2/2$ we have
\begin{equation}
\mathbb{E}|Z|^{-2c} =\frac{2}{\sqrt{2\pi}} \int_0^\infty z^{-2c}e^{-z^2/2}dz = \frac{2^{-c}}{\sqrt\pi} \int_0^\infty x^{-c-1/2}e^{-x}dx =\frac{2^{-c}\Gamma(1/2-c)}{\sqrt\pi},
\end{equation}
which is finite when $c<1/2$. Multiplying the two factors proves the claimed formula.

We also establish a uniform upper bound near zero, which we will need below. Fix $0<c<\min(a,1/2)$. On $0< t \leq 1$, the beta-prime density is bounded by $K_c t^{c-1}$ for a finite constant $K_c$, since $t^{a-1}\leq t^{c-1}$ when $c<a$. This follows by  split the normal-mixture integral at $t=1$ leading to $I_1(s)$ and $I_2(s)$. Precisely, the contribution from $0<t\leq 1$ satisfies
\begin{equation}
\begin{aligned}
I_1(s) &\leq\frac{K_c}{\sqrt{2\pi}} \int_0^1 t^{c-3/2} \exp\left(-\frac{s^2}{2t}\right)dt  
\leq\frac{K_c2^{-c}\Gamma(1/2-c)}{\sqrt\pi} |s|^{2c-1}.
\end{aligned}
\end{equation}
For $t\geq 1$, the exponential part is at most one, so
\begin{equation}
I_2(s) \leq\frac{1}{\sqrt{2\pi}} \int_1^\infty t^{-1/2}f_{\omega^2}(t)dt<\infty.
\end{equation}
Since $2c-1<0$, the quantity $|s|^{2c-1}$ diverges as $s \to 0$. This implies that the bounded term is also $O(|s|^{2c-1})$. Combining $I_1$ and $I_2$ gives
\begin{equation}
\label{eq:gamma-origin-upper}
p_\gamma(s)=O(|s|^{2c-1}) \qquad  \textrm{as } s\to0.
\end{equation}  
\end{proof}


\begin{proposition}
Consider the assumptions of Proposition~\ref{prop:tail}. Assume, in addition, that $\phi_i$ is absolutely continuous with density satisfying $f_{\phi_i}(\phi) \sim C\phi^{\alpha - 1}$ as $\phi  \to 0$ for some $\alpha > 0$. Then as $s \to 0$, the following statements hold:
\begin{enumerate}
    \item If $0 < a < \alpha$ and $a < 1/2$, then 
        $p_{\beta_i}(s) \sim A_{a,b}\, \mathbb{E}[\phi_i^{-a}]\, |s|^{2a-1}$.
    \item If $0 < \alpha < a$ and $\alpha < 1/2$, then 
        $p_{\beta_i}(s) \sim K_{a,b,\alpha}\, |s|^{2\alpha - 1}$.
    \item If $\alpha = a < 1/2$, then 
        $p_{\beta_i}(s) \sim D\, |s|^{2a-1} \log(1/|s|)$.
\end{enumerate}
Here  $A_{a,b}, K_{a,b,\alpha},D>0$ are constants.
\end{proposition}

\begin{proof}
The three possible orderings of $a$ and $\alpha$ lead to different asymptotic regimes and are treated separately.

\textbf{Case I:} Global scale $\alpha$ dominates $ a < \alpha$ and $a < 1/2$. Define
\begin{equation}
q(r)=r^{1-2a}p_\gamma(r), \qquad r>0.
\end{equation}
By Equation~\eqref{eq:gamma-origin-exact}  $q(r)\to A_{a,b}$ as $r\downarrow0$. Additionally, it also gives the global bound $0\le q(r)\le A_{a,b}$ whenever $r > 0$, because the factor $(1+r^2/(2u))^{-(a+b)}$ is at most one.

Since $a<\alpha$, $\mathbb{E}(\phi_i^{-a})$ is finite. Choose $\delta>0$ such that $f_{\phi_i}(\phi)\le2C\phi^{\alpha-1}$ for $0<\phi\le\delta$, then
\begin{equation}
\int_0^\delta\phi^{-a}f_{\phi_i}(\phi)d\phi \leq 2C \int_0^\delta\phi^{\alpha-a-1}d\phi<\infty.
\end{equation}
Moreover, the integral over $[\delta,1]$ is finite since $\phi^{-a}$ is bounded.

Using the initial conditional scaling identity gives
\begin{equation}
|s|^{1-2a}p_{\beta_i}(s) =\mathbb{E}\left[ \phi_i^{-a}q\left(\frac{|s|}{\sqrt{\phi_i}}\right) \right].
\end{equation}
Let $s_m \to 0$ and define
\begin{equation}
H_m(\phi)=\phi^{-a}q(|s_m|/\sqrt\phi).
\end{equation}
For each fixed $\phi>0$, 
\begin{equation}
H_m(\phi)\longrightarrow A_{a,b}\phi^{-a}.
\end{equation}
The bound on $q$ gives
\begin{equation}
0\le H_m(\phi)\le A_{a,b}\phi^{-a}.
\end{equation}
The right-hand side is integrable by the preceding calculation. By dominated convergence we have
\begin{equation}
\begin{aligned}
\lim_{m\to\infty}|s_m|^{1-2a}p_{\beta_i}(s_m) &=\lim_{m\to\infty}\mathbb{E}\{H_m(\phi_i)\}  =\mathbb{E}\{\lim_{m\to\infty}H_m(\phi_i)\} =A_{a,b}\mathbb{E}(\phi_i^{-a}).
\end{aligned}
\end{equation}

\textbf{Case II:} Local scale $a$ dominates $ \alpha < a$ and $ \alpha < 1/2$. 
Let $U=\sqrt{\phi_i}$ and let $\phi=u^2$. The density of $U$ is
\begin{equation}
h(u)=2u f_{\phi_i}(u^2), \qquad 0<u<1.
\end{equation}
Hence, as $u \downarrow 0$
\begin{equation}
h(u)\sim2C u^{2\alpha-1}.
\end{equation}
Since $\beta_i=U\gamma$, then for $s>0$,
\begin{equation}
\label{eq:product-density-origin}
p_{\beta_i}(s) =\int_0^1h(u)u^{-1}p_\gamma(s/u)du.
\end{equation}
Choose $\delta\in(0,1)$ small enough that $\ell(u)=\frac{h(u)}{u^{2\alpha-1}}$ is bounded on $(0,\delta]$. The same asymptotic gives $\ell(u)\to2C$ as $u\downarrow0$. In the part of Equation~\eqref{eq:product-density-origin} with $u\le\delta$ let $r=s/u$. Normalizing by $s^{2\alpha-1}$ gives
\begin{equation}
J_s=\int_{s/\delta}^\infty p_\gamma(r)r^{-2\alpha}\ell(s/r)dr.
\end{equation}
Let $s_m\downarrow0$ and define functions on $(0,\infty)$ by
\begin{equation}
F_m(r) =
\begin{cases}
p_\gamma(r)r^{-2\alpha}\ell(s_m/r),
&r\geq s_m/\delta,\\
0,&0<r<s_m/\delta.
\end{cases}
\end{equation}
For each fixed $r>0$, the indicator eventually equals one and $s_m/r\to0$. Therefore
\begin{equation}
F_m(r)\longrightarrow2C p_\gamma(r)r^{-2\alpha}.
\end{equation}
If $L_\delta=\sup_{0<u\le\delta}\ell(u)$, then 
\begin{equation}
0 \leq F_m(r)\leq L_\delta p_\gamma(r)r^{-2\alpha}.
\end{equation}
The proposed dominating function is integrable since
\begin{equation}
\int_0^\infty p_\gamma(r)r^{-2\alpha}dr =\frac12\mathbb{E}|\gamma|^{-2\alpha}<\infty,
\end{equation}
where finiteness follows from $\alpha<\min(a,1/2)$ and Lemma~\ref{lem:global-origin}. Dominated convergence thus gives
\begin{equation}
\begin{aligned}
\lim_{m\to\infty}J_{s_m} &=\lim_{m\to\infty}\int_0^\infty F_m(r)\,dr =\int_0^\infty\lim_{m\to\infty}F_m(r)\,dr\\  
&=2C\int_0^\infty p_\gamma(r)r^{-2\alpha}dr =C\mathbb{E}|\gamma|^{-2\alpha} =C\frac{2^{-\alpha}\Gamma(1/2-\alpha)}{\sqrt\pi} \frac{B(a-\alpha,b+\alpha)}{B(a,b)}.
\end{aligned}
\end{equation}
It remains to show that the integral over $u>\delta$ is smaller. Choose $c$ satisfying
\begin{equation}
\alpha<c<\min(a,1/2).
\end{equation}
By Equation~\eqref{eq:gamma-origin-upper} there are constants $K$ and $s_0$ such that $p_\gamma(x)\le Kx^{2c-1}$ whenever $0<x<s_0$. For sufficiently small $s$, the $s/u$ is below $s_0$ uniformly over $u\in[\delta,1]$. Therefore
\begin{equation}
\begin{aligned}
\int_\delta^1h(u)u^{-1}p_\gamma(s/u)du &\le Ks^{2c-1} \int_\delta^1h(u)u^{-2c}du =O(s^{2c-1}).
\end{aligned}
\end{equation}
The integral in $u$ is finite because $u\ge\delta$. Moreover,
\begin{equation}
\frac{s^{2c-1}}{s^{2\alpha-1}} =s^{2(c-\alpha)}\longrightarrow0.
\end{equation}
Thus the $u>\delta$ part is negligible, proving the second regime and its constant.

\textbf{Case III:}  Equal scales $\alpha = a < 1/2$. 

Retain $h$, $\ell$, and $\delta$ from the preceding argument. The contribution from $u>\delta$ is $O(s^{2a-1})$ by the global bound on $q$, and hence is negligible relative to $s^{2a-1}\log(1/s)$. For $u\le\delta$, the same change of variables gives
\begin{equation}
\frac{1}{s^{2a-1}} \int_0^\delta h(u)u^{-1}p_\gamma(s/u)du =\int_{s/\delta}^\infty p_\gamma(r)r^{-2a}\ell(s/r)dr.
\end{equation}
Let $\varepsilon>0$. By taking $\delta$ smaller if necessary, we may ensure that, for all $0<u\le\delta$,
\begin{equation}
2C-\varepsilon\le\ell(u)\le2C+\varepsilon.
\end{equation}
Because $ r^{1-2a} p_\gamma(r)\to A_{a,b}$, there is $r_0>0$ such that, for $0<r\le r_0$,
\begin{equation}
\frac{A_{a,b}-\varepsilon}{r} \leq p_\gamma(r)r^{-2a} \leq \frac{A_{a,b}+\varepsilon}{r}.
\end{equation}
When $s$ is small enough that $s/\delta<r_0$, both bounds apply throughout $[s/\delta,r_0]$. Hence the integral over this interval is bounded below by
\begin{equation}
(2C-\varepsilon)(A_{a,b}-\varepsilon) \int_{s/\delta}^{r_0}\frac{dr}{r}
\end{equation}
and above by the same expression with both minus signs replaced by plus signs. The logarithmic integral is
\begin{equation}
\int_{s/\delta}^{r_0}\frac{dr}{r} =\log(1/s)+\log(r_0\delta).
\end{equation}
The remaining range $r \geq r_0$ contributes only $O(1)$ after division by $s^{2a-1}$. Indeed, $\ell$ is bounded and $\int_{r_0}^\infty p_\gamma(r)r^{-2a}dr<\infty$. Divide the upper and lower bounds by $\log(1/s)$, let $s\downarrow0$, and then let $\varepsilon\downarrow0$. This yields
\begin{equation}
\frac{p_{\beta_i}(s)}{s^{2a-1}\log(1/s)}
\longrightarrow2CA_{a,b}.
\end{equation}
Symmetry of the mixture extends the result from $s\downarrow0$ to $s \to 0$ through either sign.

\end{proof}

We now show that the negative moments of the logistic-normal are finite. 

\begin{lemma}
Let $\boldsymbol{\eta}\sim \mathsf{Normal}(\boldsymbol{\mu}_\eta,\Sigma_\eta)$ on $\mathbb{R}^d$, where $\Sigma_\eta$ is positive definite, and define $ \phi_i = \frac{\exp(\eta_i)}{\sum_{r=1}^{d}\exp(\eta_r)}$. Then for for every $c>0$,  $\mathbb{E}\left[\phi_i^{-c}\right]<\infty$ . The same conclusion holds under the smoothed prior $\Pi^\lambda$ for every $\lambda > 0$.   
\end{lemma}

\begin{proof}
Fix $c > 0$. It can be shown that under the softmax 
\begin{equation}
\phi_i^{-c} =\left\{\sum_{r=1}^d e^{\eta_r-\eta_i}\right\}^c.
\end{equation}
Moreover, this implies 
\begin{equation}
\phi_i^{-c} \leq d^c\sum_{r=1}^d e^{c(\eta_r-\eta_i)}.
\end{equation}
Since the contrast $\eta_r-\eta_i$ is univariate Gaussian, its moment generating function is finite at $c$, therefore 
\begin{equation}
\begin{aligned}
\mathbb{E}(\phi_i^{-c}) &\le d^c\sum_{r=1}^d
\mathbb{E}\{e^{c(\eta_r-\eta_i)}\}
=d^c\sum_{r=1}^d \exp\left[ c(\mu_r-\mu_i) +\frac{c^2}{2}\var(\eta_r-\eta_i) \right]<\infty.
\end{aligned}
\end{equation}

Let $\Pi^{\lambda}_\mathcal{E}$ denote the smoothed prior on $\mathcal{E}$ and $Z_\lambda$ denote its normalizing constant and define the weight  $w_{\lambda}(\boldsymbol\eta) = \exp\left\{ -\frac{\operatorname{dist} (\boldsymbol\eta,\mathcal E)^2}{2\lambda} \right\}$, then
\begin{equation}
\mathbb{E}_{\Pi^{\lambda}}(\phi_i^{-c}) =\frac{\mathbb{E}\{\phi_i^{-c} w_{\lambda}(\boldsymbol\eta) \}}{Z_\lambda}.
\end{equation}
Since $0<w_\lambda\le1$ and $Z_\lambda>0$,
\begin{equation}
\mathbb{E}_{\Pi^{\lambda}}(\phi_i^{-c}) \leq\frac{\mathbb{E}(\phi_i^{-c})}{Z_\lambda}<\infty.
\end{equation}
\end{proof}


\begin{proposition}
Under the setup of Proposition~\ref{prop:tail}, suppose that $0<a<1/2$ and $\mathbb{E}\left[\phi_i^{-a}\right]<\infty$. Then, as $s\to0$,
\begin{equation}
    p_{\beta_i}(s) \sim A_{a,b}\, \mathbb{E}\left[\phi_i^{-a}\right] |s|^{2a-1},
\end{equation}
\end{proposition}

\begin{proof}
Let $q(r)=r^{1-2a}p_\gamma(r)$. As established in the proof of Lemma~\ref{lem:global-origin}, as $r\downarrow0$ then $q(r)\longrightarrow A_{a,b}$ and Equation~\eqref{eq:gamma-origin-exact} gives the uniform bound $0\leq q(r)\leq A_{a,b}$ for $r> 0$. We have shown that the conditional scaling identity implies
\begin{equation}
|s|^{1-2a}p_{\beta_i}(s) =\mathbb{E}\left[ \phi_i^{-a}q\left(\frac{|s|}{\sqrt{\phi_i}}\right) \right].
\end{equation}

Proceeding as before, take any sequence of nonzero values $s_m\to0$ and define
\begin{equation}
H_m(\phi) = \phi^{-a}q(|s_m|/\sqrt\phi).
\end{equation}
For every fixed $\phi>0$,
\begin{equation}
H_m(\phi)\longrightarrow A_{a,b}\phi^{-a}.
\end{equation}
The uniform bound on $q$ yields
\begin{equation}
0\leq H_m(\phi)\leq A_{a,b}\phi^{-a}.
\end{equation}
The function on the right is integrable by assumption. By dominated convergence,
\begin{equation}
\begin{aligned}
\lim_{m\to\infty}|s_m|^{1-2a}p_{\beta_i}(s_m)
&=\lim_{m\to\infty}\mathbb{E}\{H_m(\phi_i)\}  =\mathbb{E}\{\lim_{m\to\infty}H_m(\phi_i)\} =A_{a,b}\mathbb E(\phi_i^{-a}).
\end{aligned}
\end{equation}
\end{proof}

\begin{corollary}{}
Let $\phi \sim \mathsf{LN}^\lambda_\mathcal{E}(\boldsymbol{\mu}_\eta, \boldsymbol{\Sigma}_{\boldsymbol{\eta}}) $ and let $\omega^2\sim\mathsf{BetaPrime}(a,b)$. Then,
\begin{equation}
p_{\beta_i}(s)\sim C_{a,b}\,\mathbb E[\phi_i^{\,b}]\,|s|^{-(2b+1)} \quad (|s|\to\infty)
\end{equation}
When $0 < a < \tfrac{1}{2}$, 
\begin{equation}
\qquad p_{\beta_i}(s)\sim A_{a,b}\,\mathbb E[\phi_i^{-a}]\,|s|^{2a-1} \quad (s\to0).
\end{equation}
In particular, the marginal prior is then unbounded at the origin with a polynomial rate.
\end{corollary}

\begin{proof}
The smoothed allocation remains supported on the open simplex, so $0<\phi_i\leq 1$ and $\mathbb{E}_{\Pi^{\lambda}}(\phi_i^b)\leq1$. Hence, Proposition~\ref{prop:tail} applies and gives the displayed tail asymptotic.

Lemma~\ref{lem:neg-moments} gives $\mathbb{E}_{\Pi^{\lambda}}(\phi_i^{-c})<\infty$ for every $c>0$. In particular, choosing $c=a$ verifies the assumption of Proposition~\ref{prop:origin-ln}. That proposition gives the displayed origin asymptotic. Clearly, the factor $|s|^{2a-1}$ diverges as $s \to 0$ when $a \leq 1/2$, proving that the marginal density is unbounded at the origin.
\end{proof}


\subsection{Posterior consistency}

For completeness, we restate the model, the assumptions, and the theorem before giving the proof. We assume that the data follow the model
\begin{equation}
\byobs_n=\bZ_n\bbeta_n^*+\bvarepsilon_n, \qquad \bvarepsilon_n\sim\mathsf{Normal}(\mathbf 0,\sigma^2I_n),
\end{equation}
where $\sigma^2$ is known. Conditionally on $\bphi$ and $\omega^2$, the prior coordinates are independent and satisfy
\begin{equation} 
\beta_{nj}\mid\bphi,\omega^2 \sim\mathsf{Normal} \left(0,\frac{\sigma^2\phi_j\omega^2}{\rho_{nj}}\right),
\end{equation}
where $\rho_{nj}=\frac1n\sum_{i=1}^nz_{ij}^2$.

\noindent\textbf{Assumption~\ref{asmp:dimension} (Dimension).}
The dimension satisfies $d_n\log n = o(n)$.

\medskip
\noindent\textbf{Assumption~\ref{asmp:design} (Design).}
Let $\Lambda_{n\min}$ and $\Lambda_{n\max}$ denote the smallest and largest singular values of $\bZ_n$, respectively. There exist $0<\Lambda_{\min}\leq\Lambda_{\max}<\infty$ such that  $\Lambda_{\min}<\liminf_n \Lambda_{n\min}/\sqrt n \le \limsup_n \Lambda_{n\max}/\sqrt n<\Lambda_{\max}$.

\medskip
\noindent\textbf{Assumption~\ref{asmp:truth} (Bounded Truth).}
There exists a constant $M_\beta<\infty$ such that $\sup_{n}\max_{j\leq d_n}|\beta^*_{nj}|\le M_\beta<\infty$.

\medskip
\noindent\textbf{Assumption~\ref{asmp:hyper} (Hyperparameters).}
The hyperparameters $a,b>0$ are fixed. Moreover, there exists a constant $M_\mu<\infty$ such that $\sup_n\|\boldsymbol\mu_{\eta,n}\|_\infty\leq M_\mu<\infty$ and the eigenvalues of $\Sigma_{\eta,n}$ lie in $[\sigma_L^2,\sigma_U^2]$ for constants $0<\sigma_L^2\le\sigma_U^2<\infty$.

\medskip
\noindent\textbf{Assumption~\ref{asmp:feasibility}
(Uniform strict feasibility).}
There exist $M_\eta<\infty$, independent of $n$, such that for every $n$, there exists a point $\boldeta^*_n\in\mathbb{R}^{d_n}$ satisfying $\|\boldsymbol\eta^*_n\|_\infty\le M_\eta$ such that $\bg_{i}^\prime\boldsymbol\eta^*_n\le-\tfrac{1}{2}\|\bg_{i}\|_1$ for every constraint $i=1,\dots,m_n$, where $m_n$ is the number of constraints.

\begin{theorem}[Strong posterior consistency]
Suppose Assumptions~\ref{asmp:dimension}--\ref{asmp:feasibility} hold, and let $\Pi_n$ be the marginal prior on $\bbeta_n$ defined by \eqref{eq:beta-marginal-prior}. If $\boldsymbol\eta$ has the smoothed prior $\Pi^{\lambda_n}$ for an arbitrary sequence $\lambda_n\in(0,\infty)$, then for every $\varepsilon>0$,
\begin{equation}
\Pi_n\left( \bbeta_n: \|\bbeta_n-\bbeta_n^*\|_2 >\varepsilon \,\middle|\, \byobs_n \right) 
 = \int_{\mathbb{R}^{d_n}}
 \mathbf{1}\left\{
    \|\boldsymbol t-\bbeta_n^*\|_2>\varepsilon
\right\}
  \pi_n\left(\boldsymbol t\mid\byobs_n\right) \,d\boldsymbol t 
\longrightarrow 0,
\end{equation}
$\Pr_{\bbeta_n^*}$-almost surely as $n\to\infty$.
\end{theorem}

The proof uses the following sufficient condition from \cite{armagan_posterior_2013}. We restate it here. 

\begin{lemma}
\label{lem:adl}
Suppose that $d_n=o(n)$ and Assumption~\ref{asmp:design} holds. Fix $\varepsilon>0$. If, for some $\rho>0$,
\begin{equation}
\Pi_n\left( \|\bbeta_n-\bbeta_n^*\|_2<\Delta n^{-\rho/2}\right)>e^{-cn}
\end{equation}
for every
\begin{equation}
0<\Delta<\frac{\varepsilon^2\Lambda_{\min}^2}{48\Lambda_{\max}^2}
\end{equation}
and every
\begin{equation}
0<c<
\frac{\varepsilon^2\Lambda_{\min}^2}{32\sigma^2}-\frac{3\Delta\Lambda_{\max}^2}{2\sigma^2},
\end{equation}
for all sufficiently large $n$, then the posterior is strongly consistent.
\end{lemma}

We now prove the prior-mass estimate needed to apply the criterion.

\begin{lemma}
\label{lem:prior-mass}
For every fixed $\Delta>0$, there is a constant $C_\Delta<\infty$, independent of $n$ and $\lambda_n$, such that
\begin{equation}
-\log\Pi_n\left(
\|\bbeta_n-\bbeta_n^*\|_2<\Delta n^{-1/2}
\right)
\le C_\Delta d_n\log n
\end{equation}
for all sufficiently large $n$.
\end{lemma}

\begin{proof}
Set $r_n = \frac{\Delta}{ \sqrt{nd_n}}$. Consider the coordinatewise event 
\begin{equation}
C_n = \bigcap_{j=1}^{d_n} \left\{ |\beta_{nj}-\beta_{nj}^*|<r_n \right\}.
\end{equation}
Suppose that $\bbeta_n \in C_n$, then for every $j=1,\ldots, d_n$, we have
\begin{equation}
|\beta_{nj}-\beta_{nj}^*|<r_n.    
\end{equation}
Therefore 
\begin{equation}
\|\bbeta_n-\bbeta_n^*\|_2^2 =\sum_{j=1}^{d_n}(\beta_{nj}-\beta_{nj}^*)^2 <d_nr_n^2=\frac{\Delta^2}{n}.
\end{equation}
This shows that $C_n \subseteq \left\{ \boldsymbol{\beta}_n: \|\boldsymbol{\beta}_n-\boldsymbol{\beta}_n^*\|_2 < \Delta n^{-1/2} \right\}$, hence the coordinate wise box is contained in the desired Euclidean ball. 

Consider the latent vector $\boldeta$ on $B_n=\prod_{j=1}^{d_n} \left[\eta_{nj}^* -\frac12,\eta_{nj}^*+\frac12\right]$ and the global scale $\omega^2$ on the interval $I_n=[d_n,2d_n]$. 

We claim that the box $B_n$ lies inside the feasible set $\mathcal{E}_n$. To see this, let $\boldeta  = \boldeta_n^*+ \bu$ with $\|\boldsymbol u\|_\infty\le1/2$ then, for every constraint vector $\bg_{ni}$,
\begin{equation}
\begin{aligned}
\bg_{ni}^\prime\boldsymbol\eta &=\bg_{ni}^\prime\boldsymbol\eta_n^*
+\bg_{ni}^\prime\boldsymbol u\\
&\le-\frac12\|\bg_{ni}\|_1+|\bg_{ni}^\prime\boldsymbol u|\\
&\le-\frac12\|\bg_{ni}\|_1+\|\bg_{ni}\|_1\|\boldsymbol u\|_\infty \leq 0,
\end{aligned}
\end{equation}
where we have used H\"olders inequality for the $\ell_1$ and $\ell_\infty$ norms. Since this inequality holds for every $i=1,\ldots,m_n$, every $\boldsymbol\eta\in B_n$ satisfies all the constraints defining $\mathcal E_n$. Therefore $B_n\subseteq\mathcal E_n$.

Note that on $B_n$, by the triangle inequality, every coordinate satisfies $|\eta_{nj}|\leq M_\eta+1/2$. Therefore, 
\begin{equation}
e^{-M_\eta-1/2} \leq e^{\eta_{nj}} \leq e^{M_\eta+1/2}.
\end{equation}
This shows that 
\begin{equation}
d_n e^{-M_\eta-1/2} \leq \sum_{j=1}^{d_n} e^{\eta_{nj}} \leq d_n e^{M_\eta+1/2}, 
\end{equation}
which implies
\begin{equation}
\frac{c_0}{d_n} \leq \phi_j\leq\frac{c_1}{d_n},
\qquad j=1,\ldots,d_n,
\end{equation}
where $c_0 = e^{-2M_\eta-1}$ and $c_1=e^{2M_\eta+1}$.

By Assumption~\ref{asmp:design} all the eigenvalues of $\bZ_n^\prime\bZ_n/n$ eventually lie in $[\Lambda_{\min}^2,\Lambda_{\max}^2]$. By the Rayleigh inequality, for every real symmetric matrix $A_n$ we have $\lambda_{\min}(A_n)\|x\|_2^2 \leq x'A_n x \leq \lambda_{\max}(A_n)\|x\|_2^2.$ Taking $A_n=\frac{1}{n}Z_n' Z_n$ and  $x = e_j$ yields 
\begin{equation}
    \Lambda_{\min}^2 \leq \rho_{nj} \leq \Lambda_{\max}^2,
\end{equation}
for all $j$ and sufficiently large $n$.
Therefore the conditional variance $v_{nj} = \frac{\sigma^2\phi_j\omega^2}{\rho_{nj}}$ is bounded uniformly in $j$ and $n$ on the set $B_n \times I_n$; that is we have
\begin{equation}
v_L\le v_{nj}\le v_U,
\end{equation}
where 
\begin{equation}
v_L=\frac{\sigma^2c_0}{\Lambda_{\max}^2},
\qquad
v_U=\frac{2\sigma^2c_1}{\Lambda_{\min}^2}.
\end{equation}
Conditionally on $(\boldsymbol\eta,\omega^2)\in B_n\times I_n$, the coefficients are independent. For one coordinate,
\begin{equation}
\begin{aligned}
\Pr\left( |\beta_{nj}-\beta_{nj}^*|<r_n \,\middle|\,\boldsymbol\eta,\omega^2 \right)= \int_{\beta_{nj}^* -r_n}^{\beta_{nj}^*+r_n} \frac{1}{\sqrt{2\pi v_{nj}}} \exp\left(-\frac{x^2}{2v_{nj}}\right)dx.
\end{aligned}
\end{equation}
Note that in this interval, $|x|\le M_\beta+r_n$. Using $v_{nj}\le v_U$ in the normalization constant and $v_{nj}\ge v_L$ in the exponential gives
\begin{equation}
\Pr\left( |\beta_{nj}-\beta_{nj}^*|<r_n\,\middle|\,\boldsymbol\eta,\omega^2\right) \geq \frac{2r_n}{\sqrt{2\pi v_U}} \exp\left\{-\frac{(M_\beta+r_n)^2}{2v_L}\right\}.
\end{equation}
Multiplying this over the $d_n$ conditionally independent coordinates and taking negative logarithm results in
\begin{equation}
-\log\Pr\left( \bigcap_{j=1}^{d_n} \{|\beta_{nj}-\beta_{nj}^*|<r_n\} \,\middle|\,\boldsymbol\eta,\omega^2 \right) \leq d_n\left[ \log\left(\frac{\sqrt{2\pi v_U}}{2r_n}\right) +\frac{(M_\beta+r_n)^2}{2v_L} \right].
\end{equation}
Since $r_n = \frac{\Delta }{\sqrt{n d_n}}$, we have $ r_n \to0$. This implies that $r_n \leq 1$ for sufficiently large $n$ and so $( M_\beta+r_n)^2 \leq (M_\beta+1)^2$.

We now bound $\log(1/r_n)$. By Assumption~\ref{asmp:dimension}. $\frac{d_n\log n}{n}\longrightarrow 0$, thus, for all sufficiently large $n$, $\frac{d_n\log n}{n}<1$. For such $n$, and taking $n$ sufficiently large that $\log n>1$, we obtain $d_n < \frac{n}{\log n} < n$, implying $\log d_n\leq\log n$. Hence,
\begin{equation}
\log(1/r_n) = \frac12\log n+\frac12\log d_n+\log(1/\Delta) \leq \log n+|\log(\Delta)|.
\end{equation}
Therefore
\begin{equation}
 \log\left(\frac{\sqrt{2\pi v_U}}{2r_n}\right) +\frac{(M_\beta+r_n)^2}{2v_L} =  \log(1/r_n) + \log\left(\frac{\sqrt{2\pi v_U}}{2}\right) + \frac{(M_\beta+r_n)^2}{2v_L} \leq \log n + C_{1, \Delta},
\end{equation}
where, for example, one may take $C_{1,\Delta} = |\log\Delta| + \log\left(\frac{\sqrt{2\pi v_U}}{2}\right) + \frac{(M_\beta+1)^2}{2v_L}$. 
The constant $C_{1,\Delta}$ is independent of $n$ and of the particular value of $(\boldsymbol\eta,\omega^2)\in B_n\times I_n$. Thus, 
\begin{equation}
-\log \Pr\left( \bigcap_{j=1}^{d_n} \left\{ |\beta_{nj}-\beta_{nj}^*|<r_n \right\} \,\middle|\, \eta,\omega^2 \right) \leq d_n\left\{ \log n+C_{1,\Delta} \right\}
\end{equation}

We now lower bound the prior probability of $B_n$. Let $\Pi_{\boldsymbol\eta,n}^{0}$ denote the unconstrained Gaussian law $\Pi_{\boldsymbol\eta,n}^{0} = \mathsf{Normal}_{d_n} \left( \boldsymbol\mu_{\eta,n}, \Sigma_{\eta,n} \right),$ and let $\pi_{\boldsymbol\eta,n}^{0}$ denote its density
\begin{equation}
\pi_{\boldsymbol\eta,n}^{0}(\boldsymbol\eta) =\frac{1}{(2\pi)^{d_n/2}|\Sigma_{\eta,n}|^{1/2}}\exp\left[ -\frac12(\boldsymbol\eta-\boldsymbol\mu_{\eta,n})' \Sigma_{\eta,n}^{-1}(\boldsymbol\eta-\boldsymbol\mu_{\eta,n})
\right].
\end{equation}

Let $\lambda_{n1},\ldots,\lambda_{nd_n}$ denote the eigenvalues of $\Sigma_{\eta,n}$. By assumption, $\sigma_L^2\leq\lambda_{nj}\leq\sigma_U^2$, therefore the following bound holds
\begin{equation}
\label{eq:det-bound}
|\Sigma_{\eta,n}| = \prod_{j=1}^{d_n}\lambda_{nj} \leq(\sigma_U^2)^{d_n}.
\end{equation}
This implies
\begin{equation}
\label{eq:inverse-cov-bound}
|\Sigma_{\eta,n}|^{-1/2}
\geq
(\sigma_U^2)^{-d_n/2}.
\end{equation}

Since $\lambda_{nj}\geq\sigma_L^2$, it follows that $\Sigma_{\eta,n}^{-1} \preceq \sigma_L^{-2}I_{d_n}.$ Equivalently, for every $x\in\mathbb R^{d_n}$, we have $x'\Sigma_{\eta,n}^{-1}x \leq \sigma_L^{-2}\|x\|_2^2.$

Now let $\boldsymbol\eta\in B_n$, then $ |\eta_{nj}| \leq M_\eta + 1/2 $. Additionally, by Assumption $\|\boldsymbol\mu_{\eta,n}\|_\infty\leq M_\mu$ it follows that $|\mu_{\eta,n,j}|\leq M_\mu$. Hence, for every $j= 1, \ldots, d_n$ we have
\begin{equation}
|\eta_{nj}-\mu_{\eta,n,j}| \leq
|\eta_{nj}|+|\mu_{\eta,n,j}| \leq M_\eta+M_\mu+\frac12.
\end{equation}
Therefore
\begin{equation}
\|\boldeta-\boldsymbol\mu_{\eta,n}\|_2^2 =
\sum_{j=1}^{d_n} (\eta_{nj}-\mu_{\eta,n,j})^2  \leq
d_n\left(M_\eta+M_\mu+\frac12\right)^2.
\end{equation}
Now, since $x'\Sigma_{\eta,n}^{-1}x \leq \sigma_L^{-2}\|x\|_2^2$ with $x=\boldeta-\boldsymbol\mu_{\eta,n} $,
\begin{equation}
\begin{aligned}
(\boldeta-\boldsymbol\mu_{\eta,n})'\Sigma_{\eta,n}^{-1}(\boldeta-\boldsymbol\mu_{\eta,n})&\leq\sigma_L^{-2} \|\boldsymbol\eta-\boldsymbol\mu_{\eta,n}\|_2^2 \leq
\frac{d_n}{\sigma_L^2}
\left(M_\eta+M_\mu+\frac12\right)^2.
\end{aligned}
\end{equation}
Combining this inequality with the determinant bound~\eqref{eq:det-bound} yields uniformly over $B_n$ 
\begin{equation}
\begin{aligned}
\pi_{\boldsymbol\eta,n}^{0}(\boldsymbol\eta) &= \frac{1}{(2\pi)^{d_n/2}|\Sigma_{\eta,n}|^{1/2}}\exp\left[ -\frac12(\boldsymbol\eta-\boldsymbol\mu_{\eta,n})' \Sigma_{\eta,n}^{-1}(\boldsymbol\eta-\boldsymbol\mu_{\eta,n})\right] \\
&\geq 
(2\pi\sigma_U^2)^{-d_n/2}
\exp\left\{
-\frac{d_n(M_\eta+M_\mu+1/2)^2}
{2\sigma_L^2}
\right\}.    
\end{aligned}
\end{equation}
Integrating the uniform lower bound for $\pi_{\boldsymbol\eta,n}^{0}$ over $B_n$ yields
\begin{equation}
\begin{aligned}
\Pi_{\boldsymbol\eta,n}^{0}(B_n) &= \int_{B_n} \pi_{\boldsymbol\eta,n}^{0}(\boldeta)\,d\boldeta \geq
(2\pi\sigma_U^2)^{-d_n/2}
\exp\left\{
-\frac{d_n(M_\eta+M_\mu+1/2)^2}
{2\sigma_L^2}
\right\}.
\end{aligned}
\end{equation}
The preceding bound can be rewritten as
\begin{equation}
\Pi_{\boldsymbol\eta,n}^{0}(B_n)\geq \exp(-C_2d_n),
\end{equation}
with $C_2 = \frac12\log(2\pi\sigma_U^2) + \frac{(M_\eta+M_\mu+1/2)^2} {2\sigma_L^2}$ finite and independent of $n$.

We now transfer this bound to the smoothed prior $\Pi^{\lambda_n}$. For $\lambda_n > 0$, define the weight
\begin{equation}
w_{\lambda_n}(\boldsymbol\eta) = \exp\left\{ -\frac{\operatorname{dist} (\boldsymbol\eta,\mathcal E_n)^2}{2\lambda_n} \right\}.
\end{equation}
The function $w_{\lambda_n}(\boldsymbol\eta)$ is such that
$0<w_{\lambda_n}(\boldeta)\leq1$ for every $\boldeta \in \mathbb{R}^{d_n}$. The corresponding normalizing constant of the smoothed prior $\Pi^{\lambda_n}$ is given by
\begin{equation}
Z_{\lambda_n}= \int_{\mathbb R^{d_n}} w_{\lambda_n}(\boldeta)\, 
\Pi_{\boldeta,n}^{0}(d\boldeta) =
\int_{\mathbb R^{d_n}}
w_{\lambda_n}(\boldsymbol\eta)
\pi_{\boldeta,n}^{0}(\boldsymbol\eta)
\,d\boldsymbol\eta.
\end{equation}
Given that $\Pi_{\boldeta,n}^{0}$ is a probability measure we have $0<Z_{\lambda_n}\leq 1$.
The smoothed prior is the probability measure $\Pi^{\lambda_n}$ defined by
\begin{equation}
\Pi^{\lambda_n}(A) = \frac{1}{Z_{\lambda_n}} \int_A w_{\lambda_n}(\boldsymbol\eta) \,\Pi_{\boldeta,n}^{0}(d\boldsymbol\eta), \qquad A\subseteq\mathbb R^{d_n}.
\end{equation}
Equivalently, its density is
\begin{equation}
\pi_{\boldeta,n}^{\lambda_n}(\boldsymbol\eta) = \frac{1}{Z_{\lambda_n}} w_{\lambda_n}(\boldsymbol\eta) \pi_{\boldeta,n}^{0}(\boldsymbol\eta).
\end{equation}
Since $B_n\subseteq\mathcal E_n$, we have $\operatorname{dist}
(\boldsymbol\eta,\mathcal E_n)= 0$ whenever $\boldeta \in B_n$, which implies that $w_{\lambda_n}(\boldsymbol\eta)=1$. It follows that 
\begin{equation}
\begin{aligned}
\Pi^{\lambda_n}(B_n) &=\frac{1}{Z_{\lambda_n}}\int_{B_n}w_{\lambda_n}(\boldeta)\,\Pi_{\boldeta,n}^{0}(d\boldeta) = \frac{\Pi_{\boldeta,n}^{0}(B_n)} {Z_{\lambda_n}} 
\geq \exp(-C_2d_n).
\end{aligned}
\end{equation}

\begin{equation}
(1+t)^{-(a+b)} \geq (2t)^{-(a+b)} = 2^{-(a+b)}t^{-(a+b)}.
\end{equation}
Hence, 
\begin{equation}
f_{\omega^2}(t) \geq \frac{2^{-(a+b)}}{B(a,b)} t^{a-1}t^{-(a+b)} = \frac{2^{-(a+b)}}{B(a,b)} t^{-(b+1)}.
\end{equation}

Let $I_n=[d_n,2d_n]$ and assume that $\omega^2 \in I_n$. Using that  $(1+t)^{-(a+b)} \geq (2t)^{-(a+b)}$ whenever $t\geq 1$,
\begin{equation}
\Pr(\omega^2\in I_n) \geq \frac{2^{-(a+b)}}{B(a,b)}\int_{d_n}^{2d_n}t^{-(b+1)}\,dt =  \frac{2^{-(a+b)}(1-2^{-b})} {bB(a,b)}d_n^{-b}.
\end{equation}
Define the constant $K_{a,b} = \frac{2^{-(a+b)}(1-2^{-b})} {bB(a,b)} >0$ and  $C_3=\max\{0,-\log K_{a,b}\}$, which depends only on $a$ and $b$. It follows that 
\begin{equation}
-\log\Pr(\omega^2\in I_n) \leq b\log d_n+C_3.
\end{equation}
Since the weight $w_{\lambda_n}(\boldsymbol\eta)$ depends only on $\boldsymbol\eta$, $\boldsymbol\eta$ and $\omega^2$ remain independent under the smoothed prior.

We proceed to lower bound $\Pi_n\left( \|\boldsymbol\beta_n-\boldsymbol\beta_n^*\|_2 < \Delta n^{-1/2} \right)$ by combining the bounds we have calculated. By the law of total probability,
\begin{equation}
\Pi_n(C_n) = \int\Pr\left( C_n \,\middle|\, \boldeta,\omega^2 \right) \Pi^{\lambda_n}(d\boldeta) \operatorname{Pr}_{\omega^2} (d\omega^2).
\end{equation}
The integrand is nonnegative, so restricting the domain of integration to $B_n\times I_n$ yields
\begin{equation}
\begin{aligned}
\Pi_n(C_n) &\geq
\int_{B_n} \int_{I_n} \Pr\left( C_n \,\middle|\, \boldsymbol\eta,\omega^2 \right) \Pr_{\omega^2}(d\omega^2) \Pi^{\lambda_n}(d\boldsymbol\eta)\\
&\geq \exp\left\{ -d_n(\log n+C_{1,\Delta}) \right\} \Pi^{\lambda_n}(B_n) \Pr(\omega^2\in I_n) \\
&\geq \exp\left\{ -d_n(\log n+C_{1,\Delta}) - C_2 d_n - b \log d_n-  C_3  \right\}
\end{aligned} 
\end{equation}
Using $C_n  \subseteq \left\{ \boldsymbol{\beta}_n: \|\boldsymbol{\beta}_n-\boldsymbol{\beta}_n^*\|_2 < \Delta n^{-1/2} \right\} $, we obtain
\begin{equation}
 \begin{aligned}
-\log\Pi_n\left(\|\boldsymbol\beta_n-\boldsymbol\beta_n^*\|_2 <\Delta n^{-1/2}
\right) \leq d_n(\log n+C_{1,\Delta}) +C_2d_n +b\log d_n +C_3.
\end{aligned}
\end{equation}
By Assumption~\ref{asmp:dimension} $d_n<n$ for all sufficiently
large $n$, and hence $\log d_n\leq\log n$. Also, $\log n\geq1$ and $d_n\geq1$ for all sufficiently large $n$. After enlarging $C_{1,\Delta}$, $C_2$, and $C_3$ if necessary, we may assume that they are nonnegative. It follows that
\begin{equation}
\begin{aligned}
d_n(\log n+C_{1,\Delta}) +C_2d_n +b\log d_n+C_3 \leq \left( 1+C_{1,\Delta}+C_2+b+C_3 \right) d_n\log n.
\end{aligned}
\end{equation}
Thus, defining $C_\Delta = 1+C_{1,\Delta}+C_2+b+C_3$, we obtain
\begin{equation}
-\log\Pi_n\left( \|\bbeta_n-\bbeta_n^*\|_2 < \Delta n^{-1/2}
\right) \leq C_\Delta d_n\log n,
\end{equation}
for all sufficiently large $n$. Since none of the constants in this bound depends on $\lambda_n$, the result holds uniformly over every sequence $\lambda_n>0$.
\end{proof}

\begin{proof}[Proof of the strong posterior consistency theorem]
Assumption~\ref{asmp:dimension} implies $d_n=o(n)$, as required by the prior-mass criterion. Moreover, the preceding lemma gives, for every fixed $\Delta>0$,
\begin{equation}
-\log\Pi_n\left( \|\boldsymbol\beta_n-\boldsymbol\beta_n^*\|_2 < \Delta n^{-1/2} \right) \leq C_\Delta d_n\log n.
\end{equation}
Since $d_n\log n=o(n)$,  implies $\frac{C_\Delta d_n\log n}{n} \longrightarrow 0$.  Therefore, for every fixed admissible $c>0$, we have $C_\Delta d_n\log n<c n$ for all sufficiently large $n$. Hence,
\begin{equation}
\Pi_n\left(\|\boldsymbol\beta_n-\boldsymbol\beta_n^*\|_2<\Delta n^{-1/2}\right)> \exp(-cn)
\end{equation}
eventually. This verifies the prior-mass criterion  with $\rho=1$, and strong posterior consistency follows.
\end{proof}

\subsection{Posterior contraction}

To prove contraction we follow the approach of \cite{ghosal_convergence_2007}, where we need to show that the posterior probability outside a ball of radius $M\varepsilon_n$ around $\bbeta_n^*$ vanishes for a fixed constant $M$. To achieve this we apply their general contraction theorem to the Gaussian regression model and then verify its prior-mass condition using Lemma~\ref{lem:prior-mass}, which was proved above for the proximal R2 prior.

Consider the following sequence of families 
\begin{equation}
\left\{ P_{n,\bbeta} =\mathsf{Normal}(\bZ_n\bbeta,\sigma^2I_n) :\bbeta\in\mathbb R^{d_n} \right\},
\end{equation}
where we index over $n$ and by the unknown coefficient vector $\bbeta$. In our experiment, applying Theorem~1 of \cite{ghosal_convergence_2007} in the finite -dimensional case yields the following lemma:

\begin{lemma}
\label{lem:gaussian-regression-contraction}
Suppose that Assumption~\ref{asmp:design} holds and let $Q_n$ be any prior on $\bbeta\in\mathbb R^{d_n}$. Let $\epsilon_n\to0$ satisfy $n\epsilon_n^2\to\infty$ and $d_n=o(n\epsilon_n^2)$. If, for some fixed $C>0$,
\begin{equation}
Q_n\left( \|\bbeta-\bbeta_n^*\|_2 \leq \frac{\sigma}{\Lambda_{\max}}\epsilon_n \right) \geq \exp(-Cn\epsilon_n^2)
\end{equation}
for all sufficiently large $n$, then there is a fixed $M<\infty$ such that
\begin{equation}
Q_n\left( \|\bbeta-\bbeta_n^*\|_2>M\epsilon_n \,\middle|\,\byobs_n \right) \longrightarrow0
\end{equation}
in $P_{n,\bbeta_n^*}$-probability.
\end{lemma}

\begin{proof}
Apply \citet[Theorem~1]{ghosal_convergence_2007} to the regression experiment above. In the notation of that theorem, take both semimetrics to be
\begin{equation}
\delta_n(\bbeta,\bbeta') =\frac{\|\bZ_n(\bbeta-\bbeta')\|_2}{\sigma\sqrt n}  =\frac{1}{\sigma}\|\bbeta-\bbeta'\|_n,
\end{equation}
and take $\Theta_n=\mathbb R^{d_n}$. 

We first verify the testing condition. Fix $\bbeta_1$ and write 
\begin{equation}
\boldsymbol\mu_0=\bZ_n\bbeta_n^*, \qquad \boldsymbol\mu_1=\bZ_n\bbeta_1, \qquad \boldsymbol\Delta=\boldsymbol\mu_1-\boldsymbol\mu_0.
\end{equation}
The likelihood-ratio test of $\bbeta_n^*$ against $\bbeta_1$ is
\begin{equation}
\phi_n =\mathbf{1}\left\{ \boldsymbol\Delta^\prime \left(\byobs_n-\frac{\boldsymbol\mu_0+\boldsymbol\mu_1}{2}\right)>0 \right\}.    
\end{equation}
Under $P_{n,\bbeta_n^*}$, the statistic inside braces is Gaussian with mean $-\|\boldsymbol\Delta\|_2^2/2$ and variance $\sigma^2\|\boldsymbol\Delta\|_2^2$. If $\delta_n(\bbeta,\bbeta_1)\le \delta_n(\bbeta_1,\bbeta_n^*)/4$, its mean under $P_{n,\bbeta}$ is at least $\|\boldsymbol\Delta\|_2^2/4$, by the Cauchy--Schwarz inequality, and its variance is unchanged. Gaussian tail bounds give

\begin{equation}
P_{n,\bbeta_n^*}\phi_n \le \exp\left\{-\frac{n}{8} \delta_n^2(\bbeta_1,\bbeta_n^*)\right\},
\end{equation}
and, uniformly over $\delta_n(\bbeta,\bbeta_1)\le \delta_n(\bbeta_1,\bbeta_n^*)/4$,
\begin{equation}
P_{n,\bbeta}(1-\phi_n) \leq \exp\left\{-\frac{n}{32} \delta_n^2(\bbeta_1,\bbeta_n^*)\right\}.
\end{equation}
If $\delta_n(\bbeta_1,\bbeta_n^*)>\varepsilon$, then the neighborhood $\delta_n(\bbeta,\bbeta_1)<\varepsilon/4$ required in condition~(2.2) is contained in the neighborhood used in the preceding bound. For example take $\xi=1/4$ and $K_{\mathrm{test}}=1/32$.

For a subset $A$ of the parameter space, let $N(r,A,\delta_n)$ denote the smallest number of open $\delta_n$-balls of radius $r$ whose union contains $A$. Write
\begin{equation}
B_{\delta_n}(\bbeta_n^*,R) =\{\bbeta:\delta_n(\bbeta,\bbeta_n^*)\le R\}.
\end{equation}
By Assumption~\ref{asmp:design}, $\delta_n$ is a norm on $\mathbb R^{d_n}$, so this set is an ellipsoid with volume proportional to $R^{d_n}$. A maximal $r$-separated subset of this ball is an $r$-cover, while the balls of radius $r/2$ around its points are disjoint and contained in $B_{\delta_n}(\bbeta_n^*,R+r/2)$. Comparing their volumes gives
\begin{equation}
N\left(r,B_{\delta_n}(\bbeta_n^*,R),\delta_n\right)
\le\left(1+\frac{2R}{r}\right)^{d_n}.
\end{equation}
This argument is used by \citep{ghosal_convergence_2007} in Equation~(7.7). Taking $R=\varepsilon$ and $r=\xi\varepsilon/2$ yields, for every $\varepsilon>0$,
\begin{equation}
N\left( \frac{\xi\varepsilon}{2}, \{\bbeta:\delta_n(\bbeta,\bbeta_n^*)<\varepsilon\}, \delta_n \right) \le \left(1+\frac{4}{\xi}\right)^{d_n}.
\end{equation}
Consequently, using $\xi=1/4$,
\begin{equation}
\sup_{\varepsilon>\epsilon_n}\log N\left( \frac{\xi\varepsilon}{2}, \{\bbeta:\delta_n(\bbeta,\bbeta_n^*)<\varepsilon\}, \delta_n \right) \le d_n\log17.
\end{equation}
Since $d_n=o(n\epsilon_n^2)$, the last expression is at most
$n\epsilon_n^2$ for all sufficiently large $n$, which verifies
condition~(2.4).

For probability measures $P$ and $Q$, write $K(P,Q)=E_P\log(dP/dQ)$ for their Kullback--Leibler divergence and $V_0(P,Q)=\operatorname{Var}_P\{\log(dP/dQ)\}$. Direct calculation for the Gaussian regression likelihood gives
\begin{equation}
K(P_{n,\bbeta_n^*},P_{n,\bbeta}) =\frac{n}{2}\delta_n^2(\bbeta,\bbeta_n^*), \qquad V_0(P_{n,\bbeta_n^*},P_{n,\bbeta}) =n\delta_n^2(\bbeta,\bbeta_n^*).
\end{equation}
Taking $k=2$, the Kullback--Leibler neighborhood in the cited theorem is therefore
\begin{equation}
B_n(\bbeta_n^*,\epsilon_n;2) =\{\bbeta:\delta_n(\bbeta,\bbeta_n^*)\le\epsilon_n\}.
\end{equation}
Assumption~\ref{asmp:design} gives
\begin{equation}
\delta_n(\bbeta,\bbeta_n^*) \le\frac{\Lambda_{\max}}{\sigma} \|\bbeta-\bbeta_n^*\|_2.
\end{equation}
Hence the Euclidean ball in the assumption of the lemma is contained in $B_n(\bbeta_n^*,\epsilon_n;2)$ and 
\begin{equation}
Q_n\{B_n(\bbeta_n^*,\epsilon_n;2)\} \ge \exp(-Cn\epsilon_n^2).
\end{equation}
It remains to verify the prior-ratio condition. For the shell
\begin{equation}
S_{n,j} =\{\bbeta:j\epsilon_n<\delta_n(\bbeta,\bbeta_n^*) \le2j\epsilon_n\},
\end{equation}
we have $Q_n(S_{n,j})\le Q_n(\mathbb R^{d_n})=1$. It follows that
\begin{equation}
\frac{Q_n(S_{n,j})} {Q_n\{B_n(\bbeta_n^*,\epsilon_n;2)\}} \le \exp(Cn\epsilon_n^2) \le \exp\left(\frac{K_{\mathrm{test}}}{2} j^2n\epsilon_n^2\right)
\end{equation}
whenever $j^2\ge2C/K_{\mathrm{test}}$. This is condition~(2.5).

Theorem~1 in \citep{ghosal_convergence_2007} stated with a sequence $M_n\to\infty$ multiplying the rate $\epsilon_n$. In the present setting its proof gives the stronger conclusion with a fixed multiplier. To see this directly, set $k=2$ and set the auxiliary constant called $C$ in bound~(8.6) of the cited paper equal to one. Renaming their multiplier as $L$, define
\begin{equation}
R_n(L,J) =Q_n\left( \delta_n(\bbeta,\bbeta_n^*)>LJ\epsilon_n \,\middle|\,\byobs_n \right).
\end{equation}

Bound~(8.6) then becomes
\begin{equation}
\begin{aligned}
&E_{P_{n,\bbeta_n^*}}\{R_n(L,J)\}\\  
&\qquad\le \frac{1}{n\epsilon_n^2} +2\exp\left(-\frac{K_{\mathrm{test}}L^2}{2} n\epsilon_n^2\right) +\sum_{j\ge J} \exp\left[-n\epsilon_n^2\left\{\frac{K_{\mathrm{test}}L^2j^2}{2}-2\right\}
\right].
\end{aligned}
\end{equation}
Choose a fixed $L$ such that $K_{\mathrm{test}}L^2>4$, and choose a fixed $J$ large enough that the prior-ratio condition established above holds for every $j\ge J$. The first two terms converge to zero because $n\epsilon_n^2\to\infty$. Writing $a=K_{\mathrm{test}}L^2/2$ and $t_n=n\epsilon_n^2$, the series is bounded by
\begin{equation}
\frac{\exp[-t_n(aJ^2-2)]} {1-\exp[-t_na(2J+1)]},
\end{equation}
which also converges to zero. The right-hand side of the expectation bound
therefore converges to zero. For every $\eta>0$, Markov's inequality now gives
\begin{equation}
P_{n,\bbeta_n^*}\left(R_n(L,J)>\eta \right) \le \frac{1}{\eta}E_{P_{n,\bbeta_n^*}}\{R_n(L,J)\} \longrightarrow0.
\end{equation}
Equivalently,
\begin{equation}
Q_n\left( \delta_n(\bbeta,\bbeta_n^*)>LJ\epsilon_n \,\middle|\,\byobs_n \right)\longrightarrow0
\end{equation}
in $P_{n,\bbeta_n^*}$-probability.

Finally, the lower singular-value bound in Assumption~\ref{asmp:design} gives
\begin{equation}
\delta_n(\bbeta,\bbeta_n^*) \ge\frac{\Lambda_{\min}}{\sigma} \|\bbeta-\bbeta_n^*\|_2.
\end{equation}
Consequently,
\begin{equation}
\left\{\bbeta: \|\bbeta-\bbeta_n^*\|_2> \frac{\sigma LJ}{\Lambda_{\min}}\epsilon_n \right\} \subseteq \left\{\bbeta: \delta_n(\bbeta,\bbeta_n^*)>LJ\epsilon_n \right\}.
\end{equation}
Taking $M=\sigma LJ/\Lambda_{\min}$ proves the conclusion.
\end{proof}

We proceed with the proof of the main contraction theorem.

\begin{proof}[Proof of Theorem~\ref{thm:contraction}]
Set $\varepsilon_n=\left(\frac{d_n\log n}{n}\right)^{1/2}$, then $\varepsilon_n\to0$, $n\varepsilon_n^2=d_n\log n\to\infty$, and
\begin{equation}
\frac{d_n}{n\varepsilon_n^2}=\frac{1}{\log n}\longrightarrow0.
\end{equation}

Thus the rate and dimension conditions in Lemma~\ref{lem:gaussian-regression-contraction} are satisfied.

It remains to verify the prior-mass condition for the proximal R2 prior. Set $q_0=\sigma/\Lambda_{\max}$. Because $q_0\sqrt{d_n\log n}\ge1$ eventually,
\begin{equation}
\left\{\bbeta: \|\bbeta-\bbeta_n^*\|_2<n^{-1/2}\right\} \subseteq \left\{\bbeta: \|\bbeta-\bbeta_n^*\|_2\le q_0\varepsilon_n\right\}.
\end{equation}
Lemma~\ref{lem:prior-mass}, applied with $\Delta=1$ yields
\begin{equation}
\Pi_n\left( \|\bbeta-\bbeta_n^*\|_2\le q_0\varepsilon_n \right) \geq \exp\{-C d_n\log n\} = \exp\{-C n\varepsilon_n^2\}.
\end{equation}
The constant $C$ is independent of the positive smoothing sequence $\{\lambda_n\}$. This is precisely the prior-mass condition required by Lemma~\ref{lem:gaussian-regression-contraction} and proves the Euclidean-norm conclusion.

Finally, Assumption~\ref{asmp:design} gives
\begin{equation}
\|\bbeta-\bbeta_n^*\|_n =\frac{1}{\sqrt n}\|\bZ_n(\bbeta-\bbeta_n^*)\|_2 \leq\Lambda_{\max}\|\bbeta-\bbeta_n^*\|_2.
\end{equation}
The prediction-norm conclusion follows after increasing $M$ by the fixed factor $\Lambda_{\max}$ if necessary.
\end{proof}

\subsection{Global-scale collapse}

\begin{proposition}
Let $\bz\in\mathbb R^{n\times d}$ have full column rank and define $\hat\bbeta^{\mathrm{ols}}=(\bz^\prime \bz)^{-1}\bz^\prime\byobs$. For $\bphi$ in the open simplex and $\omega^2>0$, let $D=\operatorname{diag}(\phi_1\omega^2/\rho_1,\dots,\phi_d\omega^2/\rho_d)$.  The conditional posterior mean $\bar\bbeta=\mathsf E(\bbeta\mid\byobs,\omega^2,\bphi,\sigma^2)=(\bz^\prime \bz+D^{-1})^{-1}\bz^\prime\byobs$ does not depend on $\sigma^2$ and satisfies
\begin{equation}
\|\bar\bbeta\|_2 \le \Lambda_{\max}^2(\bz)\,\frac{\omega^2}{\rho_{\min}}\, \|\hat\bbeta^{\mathrm{ols}}\|_2, \qquad \|\hat\bbeta^{\mathrm{ols}}-\bar\bbeta\|_2 \ge \frac{\|\hat\bbeta^{\mathrm{ols}}\|_2} {1+\Lambda_{\max}^2(\bz)\,\omega^2/\rho_{\min}}.    
\end{equation}
where $\rho_{\min}=\min_{j}\rho_j>0$. Therefore, for every $\varepsilon\in(0,1)$,
\begin{equation}
\omega^2\le\frac{(1-\varepsilon)\,\rho_{\min}}{\varepsilon\,\Lambda_{\max}^2(\bz)} \quad\Longrightarrow\quad \|\hat\bbeta^{\mathrm{ols}} -\bar\bbeta\|_2\ge\varepsilon\,\|\hat\bbeta^{\mathrm{ols}}\|_2,    
\end{equation}
uniformly over the allocation simplex.
\end{proposition}

\begin{proof}
Write $A=\boldsymbol{z}^{\top}\boldsymbol{z}$ and
$c=\omega^2/\rho_{\min}$. Since $\phi_j\leq 1$, every diagonal
entry of $D$ is at most $c$. Thus $\|D\|_2\leq c$ and
$D^{-1}\succeq c^{-1}I$.

The posterior mean is given by
\begin{equation}
\bar{\boldsymbol{\beta}}
=(A+D^{-1})^{-1}A\hat{\boldsymbol{\beta}}^{\mathrm{ols}}.
\end{equation}
Since $A+D^{-1}\succeq c^{-1}I$, its inverse has spectral norm at most $c$. Hence,
\begin{equation}
\|\bar{\boldsymbol{\beta}}\|_2
\leq c\|A\|_2\|\hat{\boldsymbol{\beta}}^{\mathrm{ols}}\|_2
=\frac{\omega^2\Lambda_{\max}^2(\boldsymbol{z})}{\rho_{\min}}
\|\hat{\boldsymbol{\beta}}^{\mathrm{ols}}\|_2.
\end{equation}

For the second bound, let
$\boldsymbol{r}=\hat{\boldsymbol{\beta}}^{\mathrm{ols}}-\bar{\boldsymbol{\beta}}$.
Multiplying the expression for $\bar{\boldsymbol{\beta}}$ by $A+D^{-1}$ gives
\begin{equation}
(A+D^{-1})\bar{\boldsymbol{\beta}} =A\hat{\boldsymbol{\beta}}^{\mathrm{ols}}.
\end{equation}
Rearranging yields
\begin{equation}
\hat{\boldsymbol{\beta}}^{\mathrm{ols}} =(I+DA)\boldsymbol{r}.
\end{equation}
Taking norms gives
\begin{equation}
\|\hat{\boldsymbol{\beta}}^{\mathrm{ols}}\|_2
\leq (1+\|D\|_2\|A\|_2)\|\boldsymbol{r}\|_2
\leq (1+c\|A\|_2)\|\boldsymbol{r}\|_2.
\end{equation}
Therefore,
\begin{equation}
\|\hat{\boldsymbol{\beta}}^{\mathrm{ols}}-\bar{\boldsymbol{\beta}}\|_2
\geq
\frac{\|\hat{\boldsymbol{\beta}}^{\mathrm{ols}}\|_2}
{1+\omega^2\Lambda_{\max}^2(\boldsymbol{z})/\rho_{\min}}.
\end{equation}
Under the stated threshold on $\omega^2$, the denominator is at most
$1/\epsilon$, proving the final implication. Both bounds use the allocation
only through $\phi_j\leq 1$, so they hold uniformly over the open simplex.
\end{proof}

\subsection{Predictive loss and coefficient recovery}
\label{app:predictive-loss}

This section connects the coefficient recovery results in Section~\ref{sec:experiments} to the out-of-sample $\elpd$ comparisons. Let a held-out observation be generated as $\tilde y = \tilde{\bz}^\top\bbeta^* + \tilde\varepsilon$ with $\tilde\varepsilon\sim\normal{0,\sigma^2}$ independent of $\tilde{\bz}$, where $\tilde{\bz}=h(\tilde{\bx})$ collects the main and interaction terms, and let $\boldsymbol\Sigma=\mathbb E[\tilde{\bz}\tilde{\bz}^\top]$.

For a point estimate $\hat\bbeta$ with estimation error $\boldsymbol\delta=\hat\bbeta-\bbeta^*$, the plug-in predictive density $p(\tilde y\mid\tilde{\bz})=\normal{\tilde{\bz}^\top\hat\bbeta,\sigma^2}$ has expected squared error
\begin{equation}
\mathbb E\!\left[(\tilde y-\tilde{\bz}^\top\hat\bbeta)^2\right]
=\mathbb E\!\left[(-\tilde{\bz}^\top\boldsymbol\delta+\tilde\varepsilon)^2\right]
=\boldsymbol\delta^\top\boldsymbol\Sigma\,\boldsymbol\delta+\sigma^2.
\end{equation}
Substituting this into the expected log predictive density gives
\begin{equation}
\elpd(\hat\bbeta)
=\mathbb E\!\left[\log p(\tilde y\mid\tilde{\bz})\right]
=-\tfrac12\log(2\pi\sigma^2)-\tfrac12-\frac{1}{2\sigma^2}\,\boldsymbol\delta^\top\boldsymbol\Sigma\,\boldsymbol\delta.
\end{equation}
The $\elpd$ difference between model $m$ (with error $\boldsymbol\delta_m$) and the baseline model (with error $\boldsymbol\delta_0$) is therefore
\begin{equation}
\Delta\elpd_m=\frac{1}{2\sigma^2}\left(\boldsymbol\delta_0^\top\boldsymbol\Sigma\,\boldsymbol\delta_0-\boldsymbol\delta_m^\top\boldsymbol\Sigma\,\boldsymbol\delta_m\right).
\end{equation}
If we partition the coefficients into blocks $g$ (active/inactive main effects and interactions) and neglect the off-diagonal cross-terms of $\boldsymbol\Sigma$ between blocks, this difference decomposes additively as
\begin{equation}
\Delta\elpd_m\approx\frac{1}{2\sigma^2}\sum_g\!\left(\boldsymbol\delta_{0,g}^\top\boldsymbol\Sigma_{gg}\boldsymbol\delta_{0,g}-\boldsymbol\delta_{m,g}^\top\boldsymbol\Sigma_{gg}\boldsymbol\delta_{m,g}\right).
\end{equation}
Unlike the group-wise relative RMSE, which averages errors within each coefficient group separately, $\Delta\elpd_m$ sums the design-weighted squared errors across all terms. Because the number of interaction terms $\binom{p}{2}$ grows quadratically with $p$, the improvements of $\Rtdprox$ on the interaction block outweigh its slightly higher error on inactive main effects.

\clearpage
\bibliographystyle{rss}
\bibliography{refs}

\end{document}